\documentclass{lmcs}
\pdfoutput=1
\usepackage[utf8]{inputenc}

\usepackage{lastpage}
\lmcsdoi{21}{2}{7}
\lmcsheading{}{\pageref{LastPage}}{}{}%
{Aug.~16,~2023}{Apr.~22,~2025}{}

\usepackage{microtype}
\usepackage{stmaryrd}
\usepackage{float}
\usepackage{adjustbox}
\usepackage{amsfonts}
\usepackage{amssymb}
\usepackage{amsmath,bbm}
\usepackage{xspace}
\usepackage[normalem]{ulem}
\usepackage{hyperref}
\usepackage[capitalise]{cleveref}
\usepackage{graphicx}
\usepackage{wasysym}
\usepackage{mathtools}
\usepackage{array,amsmath}
\usepackage{transparent}
\usepackage{marvosym}
\usepackage{upgreek}
\DeclareMathAlphabet\mathbfcal{OMS}{cmsy}{b}{n}

\usetikzlibrary{trees}
\usetikzlibrary{shapes,fit,backgrounds,fadings}
\usetikzlibrary{positioning,arrows,calc,decorations.markings}
\tikzset{
modal/.style={>=stealth',shorten >=1pt,shorten <=1pt,auto,node distance=1.5cm,semithick},
world/.style={circle,draw,minimum size=0.5cm,fill=gray!15},
point/.style={circle,draw,inner sep=0.5mm,fill=black},
reflexive above/.style={->,loop,looseness=7,in=120,out=60},
reflexive below/.style={->,loop,looseness=7,in=240,out=300},
reflexive left/.style={->,loop,looseness=7,in=150,out=210},
reflexive right/.style={->,loop,looseness=7,in=30,out=330}
}

\makeatletter
\pgfdeclareverticalshading{pgf@lib@fade@South}{100bp}
{color(0bp)=(pgftransparent!100); color(40bp)=(pgftransparent!100);
 color(75bp)=(pgftransparent!0); color(100bp)=(pgftransparent!0)}%
\pgfdeclarefading{South}{\pgfuseshading{pgf@lib@fade@South}}%
\makeatother

\usepackage{bold-extra} % for bold typewriter fonts
\usepackage{nimbusmononarrow} % a better distinguishable \mathtt
\usepackage[scr=boondoxo]{mathalpha}
\usepackage{scalerel} % allows to scale symbols 

\theoremstyle{plain}
\newtheorem{theorem}[thm]{Theorem}
\newtheorem{lemma}[thm]{Lemma}
\newtheorem{proposition}[thm]{Proposition}
\newtheorem{corollary}[thm]{Corollary}
\newtheorem{observation}[thm]{Observation}

\theoremstyle{definition}
\newtheorem{definition}[thm]{Definition}

\newtheorem{example}[thm]{Example}
\newtheorem{remark}[thm]{Remark}

\newcommand{\remarkstart}{\color{black}}

\newtheoremstyle{myexample}% ⟨name⟩ 
{3pt}% ⟨Space above⟩ 
{3pt}% ⟨Space below⟩
{\rm}% ⟨Body font⟩
{}% ⟨Indent amount⟩
{\bf}% ⟨Theorem head font⟩
{.}% ⟨Punctuation after theorem head⟩
{.5em}% ⟨Space after theorem head⟩
{}% ⟨Theorem head spec (can be left empty, meaning ‘normal’)⟩
\theoremstyle{myexample}

\usepackage{framed}
\definecolor{shadecolor}{rgb}{0.93,0.93,0.93}
\definecolor{darkerblue}{rgb}{0.0,0.0,0.6}
\definecolor{darkblue}{rgb}{0.0,0.0,0.9}
\definecolor{darkerred}{rgb}{0.6,0.0,0.0}
\definecolor{darkred}{rgb}{0.9,0.0,0.0}
\definecolor{darkergreen}{rgb}{0.0,0.4,0.0}
\definecolor{darkgreen}{rgb}{0.0,0.6,0.0}
\definecolor{darkercyan}{rgb}{0.0,0.4,0.6}
\definecolor{darkcyan}{rgb}{0.0,0.6,0.9}
\renewenvironment{example}%[1][]
{
	\begin{shaded*}	
		\vspace{-1.5ex}
		\begin{exa}%[#1]
		}{
		\end{exa}
		\vspace{-1.5ex}
	\end{shaded*}
}
\renewenvironment{remark}%[1][]
{
	\begin{shaded*}	
		\vspace{-1.5ex}
		\begin{rem}%[#1]
		}{
		\end{rem}
		\vspace{-1.5ex}
	\end{shaded*}
}

\crefname{rem}{Remark}{Remarks}
\crefname{exa}{Example}{Examples}

\newcommand{\footnotetextbefore}[1]{\addtocounter{footnote}{1}\footnotetext{#1}\addtocounter{footnote}{-1}
}

\theoremstyle{plain}
\newtheorem{innercustomlem}{Lemma}
\newenvironment{customlem}[1]
  {\renewcommand\theinnercustomlem{#1}\innercustomlem}
  {\endinnercustomlem}

  \newtheorem{innercustomthm}{Theorem}
\newenvironment{customthm}[1]
  {\renewcommand\theinnercustomthm{#1}\innercustomthm}
  {\endinnercustomthm}
  
  \newtheorem{innercustomprop}{Proposition}
  \newenvironment{customprop}[1]
  {\renewcommand\theinnercustomprop{#1}\innercustomprop}
  {\endinnercustomprop}

\newtheorem{innercustomobs}{Observation}
\newenvironment{customobs}[1]
  {\renewcommand\theinnercustomobs{#1}\innercustomobs}
  {\endinnercustomobs}

\theoremstyle{definition}

\newcommand{\skipit}[1]{}

\renewcommand{\vec}[1]{\boldsymbol{#1}}

\newcommand{\vecconst}[1]{\mathtt{\mathbf{#1}}}

\newcommand{\signature}{\Sigma}

\newcommand{\instance}{\mathcal{I}}
\newcommand{\inst}{\instance}
\newcommand{\database}{\mathcal{D}}
\newcommand{\db}{\database}
\newcommand{\adom}[1]{\fcn{dom}(#1)}

\newcommand{\origcol}{o}

\newcommand{\cols}{\mathbb{L}}

\newcommand{\arity}[1]{\fcn{ar}(#1)}

\newcommand{\length}[1]{|#1|}
\newcommand{\disun}{\oplus}

\newcommand{\newconst}[1]{\pred{New}_{\const{#1}}} %#1 is the constant name produced
\newcommand{\newanon}{\pred{New}_{\const{*}}} %#1 is the constant name produced
\newcommand{\cadd}[4]{#4\pred{\texttt{-}Add}^{#3}_{#1#2}} %#1, #2 are the colors #3, #4 the indices of the associated subtuples, #5 the "goal color"
\newcommand{\radd}[4]{#4\pred{\texttt{-}Add}^{#3}_{#1#2}}
\newcommand{\unradd}[2]{#2\pred{\texttt{-}Add}_{#1}}
\newcommand{\recol}[2]{#2\pred{\texttt{-}ReCol}_{#1}}

\newcommand{\ercw}[1]{\pw{#1}}
\newcommand{\pw}[1]{\fcn{pw}(#1)}
\newcommand{\Ercwidth}{\Partitionwidth}
\newcommand{\ercwidth}{\partitionwidth}

\newcommand{\hypercliquewidth}{\partitionwidth}
\newcommand{\Partitionwidth}{Par\-ti\-tion\-width}
\newcommand{\partitionwidth}{par\-ti\-tion\-width}

\newcommand{\triggers}[2]{#1 \leadsto #2}
\newcommand{\ntriggers}[2]{#1 \not\leadsto #2}
\newcommand{\cut}{\triangleright}

\newcommand{\hism}{h}

\newcommand{\treewidth}[1]{\fcn{tw}(#1)}

\newcommand{\head}{\fcn{head}}
\newcommand{\body}{\fcn{body}}

\newcommand{\fcn}[1]{\mathsf{#1}}

\newcommand{\ent}{\fcn{ent}}
\newcommand{\col}{\fcn{col}}

\newcommand{\ruleset}{\calr}

\newcommand{\bts}{\mathbf{fts}}

\newcommand{\fes}{\mathbf{fes}}

\newcommand{\fus}{\mathbf{fus}}

\newcommand{\fcs}{\mathbf{fcs}}
\newcommand{\fps}{\mathbf{fps}} %Finite partitionwidth sets

\newcommand{\hcs}{\fps}
\newcommand{\fhcs}{\hcs}

\newcommand{\urs}{\mathbf{urs}}
\newcommand{\reds}{\mathbf{reds}}
\newcommand{\fods}{\mathbf{fods}}

\newcommand{\chase}{\mathbf{Ch}}
\renewcommand{\sa}{\mathbf{Ch}_{\infty}}

\newcommand{\fchain}[1]{\chase(#1)}
\newcommand{\ksat}[2]{\chase_{#1}(#2)}
\newcommand{\sat}[1]{\ksat{\infty}{#1}}

\newcommand{\rewrs}[2]{\fcn{rew}_{#1}(#2)}

\newcommand{\calr}{\mathcal{R}}

\newcommand{\kb}{\mathbfcal{K}}
\newcommand{\sig}{\Sigma}
\newcommand{\pred}[1]{\mathtt{#1}}
\newcommand{\const}[1]{\mathtt{#1}}
\newcommand{\thephi}{\varphi}

\newcommand{\fr}{\mathit{fr}}
\newcommand{\tree}{T}

\newcommand{\ve}{\vec{e}}
\newcommand{\vx}{\vec{x}}

\newcommand{\vy}{\vec{y}}

\newcommand{\vz}{\vec{z}}
\newcommand{\vt}{\vec{t}}

\newcommand{\va}{\vec{a}}
\newcommand{\vb}{\vec{b}}
\newcommand{\vc}{\vec{c}}

\newcommand{\rpred}{\pred{R}}
\newcommand{\ppred}{\pred{P}}
\newcommand{\epred}{\pred{E}}

\newcommand{\query}{q}
\newcommand{\concat}{\mathbin{{+}\hspace{-0.5ex}{+}}}
\newcommand{\ifandonlyif}{\textit{iff} } %if and only if
\newcommand{\iffi}{\textit{iff} } %if and only if
\definecolor{piotr}{RGB}{250, 100, 100}
\definecolor{thomas}{RGB}{50, 190, 50}
\definecolor{tim}{RGB}{0, 0, 250}
\definecolor{sebastian}{RGB}{153, 50, 204}

\newcommand{\smallldots}{\ldots}

\newcommand{\defend}{\hfill\ensuremath{\Diamond}}

\newcommand{\shortto}{\hstretch{0.5}{\to}}

\renewcommand{\phi}{\varphi}

\newcommand{\probbox}[1]{\centerline{\framebox{\parbox{0.8\columnwidth}{#1}}}}

\newcommand{\narrowcdots}{\raisebox{0.3ex}{...}}

\begin{document}

\title[Querying FO Theories via Countermodels of Finite Width]{Decidability of Querying First-Order Theories\texorpdfstring{\\}{} via Countermodels of Finite Width}

\author[T.~Feller]{Thomas Feller\lmcsorcid{0000-0001-8420-6118}}
\author[T.~S.~Lyon]{Tim S. Lyon\lmcsorcid{0000-0003-3214-0828}}
\author[P.~Ostropolski-Nalewaja]{Piotr Ostropolski-Nalewaja\lmcsorcid{0000-0003-3214-0828}}
\author[S.~Rudolph]{Sebastian Rudolph\lmcsorcid{0000-0003-3214-0828}}

\address{Technische Universität Dresden, Germany}
\email{\{thomas.feller, timothy\_stephen.lyon, piotr.ostropolski-nalewaja, sebastian.rudolph\}\linebreak\mbox{\hspace{3.0cm} @tu-dresden.de}}

%%%%%%%%%%%%%%%%%%%%%%%%%%%%%%%%%%%%%%%%%%%%%%%%%%%%%%%%%
%%%%%%%%%%%%%%%%%%%%%BEGIN DOCUMENT%%%%%%%%%%%%%%%%%%%%%%
%%%%%%%%%%%%%%%%%%%%%%%%%%%%%%%%%%%%%%%%%%%%%%%%%%%%%%%%%

\begin{abstract}
We propose a generic framework for establishing the decidability of a wide range of logical entailment problems (briefly called \emph{querying}), based on the existence of countermodels that are structurally simple, gauged by certain types of \emph{width} measures (with treewidth and cliquewidth as popular examples). %, where structures of finite width ensure beneficial logical properties.
 As an important special case of our framework, we identify logics exhibiting width-finite \emph{finitely universal model sets}, warranting decidable entailment for a wide range of homomorphism-closed queries, subsuming a diverse set of practically relevant query languages. 
 As a particularly powerful width measure, we propose to employ Blumensath's \emph{partitionwidth}, which subsumes various other commonly considered width measures and exhibits highly %further very 
 favorable computational and structural properties.
 Focusing on the formalism of existential rules as a popular showcase, we explain how \emph{finite-partitionwidth sets} of rules sub\-sume other known abstract decidable classes but -- leveraging existing notions of stratification -- also cover a wide range of new rulesets. We expose natural limitations for fitting the class of finite-unifi\-cation sets into our picture and suggest several options for remedy.  
\end{abstract}

\maketitle

%SECTION 1: INTRODUCTION
%\input{intro.tex}
\section{Introduction}

The task of checking logical entailment is at the core of logic in computer science, with wide applications in many related fields including symbolic AI, databases, and veri\-fi\-ca\-tion, but also impacting neighboring disciplines like mathematics and philosophy.
Automating any kind of logical reasoning 
requires a profound understanding of the underlying entailment problem's computational properties, and in particular, its de\-ci\-da\-bi\-li\-ty status. The fact that entailment in many natural logical formalisms is undecidable has spawned  extensive research activities toward the identification of expressive yet decidable fragments and generally toward a better understanding of generic principles underlying the divide between decidability and undecidability \cite{BorgerGG1997,Vardi96,Gradel01,PrattHartmann23}. This article contributes to this endeavor.

\medskip

The concept of logical entailment appears to be the essence and primordial motivation of the study of logic. It embodies the judgment, typically noted down as $\Phi \models \Psi$, that a certain statement $\Psi$ is an inevitable consequence of another statement $\Phi$, established through formal structural analysis. Starting out as a primarily philosophical endeavor  millenia ago \cite{PriorAnalytics}, the study of logic(\hspace{-1pt}al entailment) underwent substantial mathematical formalization and rigorization in the late 19th and early 20th century \cite{Boole1854,Frege1879,Tarski1943}. With the advent of computing devices and computer science, the algorithmization of logical reasoning and the computational characteristics of entailment problems became central topics. Today, logical entailment is a notion widely employed throughout computer science, including but not restricted to the following areas:      

\begin{itemize}
\item \textbf{Formal Methods}: This discipline is concerned with describing dynamic systems and ascertaining their properties. If the dynamic system -- typically specified as a state-transition system -- is fully specified (e.g., by means of a process algebra), the task of deciding if it satisfies some property (such as safety or lifeness) is a variant of model checking (and explicitly referred to as such).  
\item \textbf{Databases}: The primary task considered in the database (DB) area is that of query answering, which, in its plain form, can be logically described as finding variable assignments (the \emph{query answers}) for a logical formula (the \emph{query}) which make that formula true over a given finite structure (the \emph{database}). 
If the logic used as query language permits to directly refer to database elements, the query answering problem can be reduced to validity of logical sentences (\emph{Boolean queries}) -- and hence to classical logical model checking.

More advanced database settings, however, deal with the phenomenon of incomplete information, where the database does not carry the full knowledge about the state of affairs and rather is to be seen as a logical description which allows for several realizations (that is, models). Such ambiguity may arise from absence of certain information (e.g., when following the open world assumption) or from providing higher-level logical constraints alongside factual data (as in \emph{deductive databases}).  
Determining which query answers are \emph{certain answers}, or which boolean queries are satisfied in all models then amounts to an entailment problem.             

Another way in which entailment arises in the area of databases is in static analysis of queries and views. The popular problem of query containment asks if every answer for $Q_1$ is an answer for $Q_2$ over every database.    
\item \textbf{Classical Knowledge Representation}: Classical logic-based knowledge representation (KR) is a subdiscipline of artificial intelligence that attempts to represent world knowledge as finite logical theories  -- referred to as \emph{knowledge bases} -- of some logical formalism -- sometimes called \emph{ontology language}, with the purpose of establishing a clear semantics and harnessing automated inferencing techniques to provide reasoning support. Rooted in mathematical logic, a fundamental task considered in classical KR is that of \emph{axiom entailment}, asking if some logical sentence (referred to as an axiom) of the formalism is a consequence of the considered knowledge base. In most cases, the considered logical language supports negation on the sentence level, whereby the axiom entailment problem can be reduced to that of \emph{knowledge base satisfiability}. On the other hand, knowledge base satisfiability is usually reducible to axiom entailment by asking if the knowledge base entails a contradictory sentence -- provided such a sentence exists in the considered formalism. Prominent KR formalisms are description logics and rule-based languages.     
\item \textbf{Contemporary Knowledge Representation}: Over the recent two or three decades, the DB and KR fields have converged significantly. One one hand, the DB side has developed an increased interest into taking available background knowledge (and the implicit knowledge implied by it) into account when querying data -- which is ever more important in view of the rise of many open, less curated data sources e.g. on the Web. On the other hand, the KR community became aware of the necessity of distinguishing between the formalism used for specifying the background knowledge (the ontology language) and that for expressing information needs (the query language). The enlarged intersection between DB and KR fields that might be referred to as \emph{contemporary KR} is primarily concerned with the task of \emph{ontology-based query answering} (OBQA), that is posing queries to knowledge bases under the certain answer semantics. At the core of this task is, again, an entailment problem. 
\end{itemize}

\medskip

More formally, in this article, we will consider the following setting: Assuming two classes $\mathfrak{F}$ and $\mathfrak{Q}$ of logical sentences, the decision problem of interest is specified as follows:

\bigskip

\probbox{
	\textbf{Problem:} $\textsc{Entailment}(\mathfrak{F},\mathfrak{Q})$\\
	\textbf{Input:} $\Phi \in \mathfrak{F}$, $\Psi \in \mathfrak{Q}$.\\
	\textbf{Output:} \textsc{yes}, if $\Phi \models \Psi$, \textsc{no} otherwise.
}

\bigskip

\noindent
We will refer to $\mathfrak{F}$ as a \emph{specification language}, to $\mathfrak{Q}$ as a \emph{query language}, and we sometimes call the above entailment problem ``query entailment'' or %even briefly 
 ``querying'', as is common in database theory and knowledge representation. We assume $\mathfrak{F}$ and $\mathfrak{Q}$ both rely on a common model theory based upon relational structures of arbitrary cardinality.\footnote{In particular, this article will not be concerned with finite-model semantics and the notion of finite entailment.} 
 More specifically, in this article, we focus on the setting where $\mathfrak{F}$ is a fragment of first-order logic (FO) and $\mathfrak{Q}$ is a fragment of universal second-order logic ($\forall$SO). There are good reasons for focusing on this particular -- yet still very %generic 
 general -- shape of entailments:

\begin{itemize}
\item Given the choice of $\mathfrak{Q}$, the considered problem occurs to be a special case of entailment in second-order logic.
However, entailment problems of this shape can be reduced to first-order entailment in a rather straighforward way, which  ensures semidecidability %and legitimizing to confine one's attention to countable models. 
and allows one to restrict one's attention to countable models (see~\Cref{sec:easy}).
 
\item The reason for posing the problem as an entailment $\Phi \models \Psi$, rather than the (un)satisfiability of $\Phi \wedge \neg \Psi$, is that in many scenarios encountered in databases, knowledge representation, and other practical settings, the specification language and the query language differ significantly so that the reduction to (un)satisfiability does not generally produce a sentence that would readily fall into a known decidable fragment \cite{%LevyR96,
CalvaneseGL98,CalvaneseEO09,RudolphG10,OrtizS12,BaranyGO13}.

\item Many contemporary query languages (starting from simple types of path queries used in graph databases \cite{FlorescuLS98} and reaching up to highly expressive query formalisms like disjunctive datalog \cite{EiterGM97}) go well beyond FO, but can be straightforwardly expressed in $\forall$SO (see~\Cref{sec:homclosedqueries}).

\item The two other ``archetypical'' decision problems in logic can be obtained as special cases of $\textsc{Entailment}(\mathfrak{F},\mathfrak{Q})$: picking $\mathfrak{Q}=\{\mathrm{false}\}$ yields (un)satisfiability checking in $\mathfrak{F}$, whereas choosing $\mathfrak{F}=\{\mathrm{true}\}$ amounts to the validity problem in $\mathfrak{Q}$.
\end{itemize}

The described problem in its unconstrained form, $\textsc{Entailment}(\mathrm{FO},\forall\mathrm{SO})$, is merely semidecidable, that is, there exists an algorithm that successfully terminates in case $\Phi \models \Psi$. Thus, in order to arrive at a decidable setting, one needs to further constrain $\mathfrak{F}$ and $\mathfrak{Q}$ in a way that ensures co-semidecidability, i.e. the existence of a procedure successfully terminating whenever $\Phi \not\models \Psi$. %One possible guiding principle in that regard can be the existence of \emph{countermodels} (that is, models of $\Phi$ that are not models of $\Psi$ and hence witness the non-entailment) of a particular, simple structure. 
We propose to achieve this by exploiting the existence of \emph{countermodels} (that is, models of $\Phi$ that are not models of $\Psi$ and hence witness the non-entailment) which are of a particular, simple shape.
To gauge ``structural difficulty'', a variety of functions, usually referred to as \emph{width measures} have been proposed and have proven beneficial for a wide range of diverse purposes (graph-theoretic, model-theoretic, algorithmic, and others). 

\begin{example} As an easy example, consider the folklore observation that the satisfiability problem of a FO fragment $\mathfrak{F}$ is decidable whenever it exhibits the \emph{finite-model property} (fmp), meaning that every satisfiable sentence $\Phi \in \mathfrak{F}$ has a finite model: In that case, unsatisfiability of $\Phi$ can be detected by enumerating all its consequences and stopping if $\mathit{false}$ is produced, while satisfiability of $\Phi$ can be spotted by enumerating all finite structures, checking modelhood, and stopping if one finite model of $\Phi$ is encountered -- then, both semidecision algorithms run in parallel yield a decision procedure.

Returning to our framework, note that the statement \textsl{``$\mathfrak{F}$ has the fmp''} can be paraphrased into \textsl{``$\textsc{Entailment}(\mathfrak{F},\{\mathrm{false}\})$ has the finite-domainsize-countermodel property''}, and we see that the function just returning the domain size of a structure can be seen as a width measure useful for our purpose. 
\end{example}

When it comes to characterizing logics that are guaranteed to have strucutrally simple models, the \emph{treewidth} measure \cite{ROBERTSON198449} has been well-established as another useful tool, e.g., for logics of the guarded family \cite{AndrekaNB98,Gradel99}. More recently, treewidth and \emph{cliquewidth} of universal models have been investigated in the context of existential rules \cite{BagLecMugSal11,CaliGK13,ICDT2023}.\footnote{Specifics will be discussed in detail in later sections.}  

\medskip

These considerations make it seem worthwhile to investigate -- on a more abstract level -- under which conditions decidability of $\textsc{Entailment}(\mathfrak{F},\mathfrak{Q})$ follows from its \textsl{``finite-width-countermodel property''} for some kind of width. This is what we set out to do in this paper. In short, our contributions are the following:

\begin{itemize}
\item In \Cref{sec:easy} and \Cref{sec:genericdecidability}, we establish a generic framework for showing decidability of $\textsc{Entailment}(\mathfrak{F},\mathfrak{Q})$ based on the existence of countermodels of finite width, and discuss how it can be used to obtain and generalize miscellaneous decidability results from the literature (where the existing proofs sometimes even employ an ad hoc version of our generic argument).   

\item In \Cref{sec:homclosedqueries}, we single out the lucrative special case where the specification language warrants finitely universal model sets of finite width, yielding decidability of entailment for a vast class of homomorphism-closed queries of high expressivity. In our view, this special case strikes a very favorable balance between the expressivity of the specification language and that of the query language. 

\item \Cref{sec:partitionwidth} exposes (a suitably reformulated and simplified version of) Blumensath's \emph{partitionwidth} \cite{Blumensath03,Blumensath06} as a very powerful width measure, superseding all others previously considered in this context, and exhibiting the extremely convenient property of being robust under diverse model-theoretic transformations.

\item From \Cref{sec:rules} on, we apply our framework in the area of existential rules% 
%(also known as tuple-generating dependencies)
, by defining \emph{finite-partitionwidth sets} of rules (\Cref{sec:finitewsets}). We show that this class of rulesets allows for the incorporation of various types of datalog layers (in the sense of stratification), yielding a flexible toolset for the analysis and creation of rulesets with favorable querying properties (\Cref{sec:strat}). We also discuss the somewhat complicated case of rulesets whose decidable querying is based on first-order rewritability (\Cref{sec:FO-rewritability}).    
\end{itemize}

%\pagebreak

%\textcolor{red}{TBC by Sebastian}

Our work draws from various subfields of computer science and mathematics. To reduce the cognitive load, many of the concepts used will be defined in the place they are needed rather than in the preliminaries; we will make clear if these notions rely on earlier work. The results presented are novel unless clearly stated otherwise.
To improve readability, we defer some technical details and most proofs to the appendix.
Several of the provided examples refer to previously established first-order fragments, also including formalisms from the description logic (DL) family. As these observations follow from general properties of these logics rather than their technical details, we refrain from introducing each and every formalism in meticulous detail. In the appendix we do, however provide a concise overview of the syntactic fragments of first- and second-order logic used (Appendix~\ref{app:standardFO}) and a slightly more comprehensive layout of syntax and semantics of DLs (Appendix~\ref{app:DLintro}).   
For a full treatise of these logics, we refer the interested reader to the standard literature provided inside these appendices.

%\clearpage

%SECTION 2: PRLIMINARIES
\section{Preliminaries}\label{sec:preliminaries}

%$\vec{\const{a}}$ $\const{a}$ $\mathtt{\mathbf{a}}$
\newcommand{\shuffle}[1]{\concat_{#1}}

\medskip \noindent {\bf Tuples and shufflings.}
    We denote finite sequences $a_1, \ldots, a_n$ of objects, also referred to as \emph{($n$-)tuples}, by corresponding bold font symbols $\va$ and write $\length{\va}$ to refer to a tuple's length $n$.  
	For tuples $\va = a_1,\ldots,a_n$ and $\vb = b_1,\ldots,b_m$ as well as $I \subseteq \{1,\ldots,n+m\}$ with $|I|=n$, we define the \emph{$I$-shuffling} of $\va$ and $\vb$, denoted $\va \shuffle{I} \vb$,
	as the interleaving combination of $\va$ and $\vb$ where $I$ indicates the positions in the resulting $(n+m)$-tuple where the entries of $\va$ go. Formally, $\va \shuffle{I} \vb$ is the tuple
	$\vc = c_1,\ldots,c_{n+m}$ where
	$$
	c_i = \begin{cases}\\[-3.7ex]
	a_{|\{1,\ldots,i\} \cap I|} & \text{ if } i \in I,\\[-0.5ex]
	b_{|\{1,\ldots,i\} \setminus I|} & \text{ if } i \not\in I.\\
	\end{cases}
	$$
For instance, $a_1a_2 \shuffle{\{2,5\}} b_1b_2b_3 = b_1a_1b_2b_3a_2$. We lift these operations to sets of tuples as usual, that is, if $A$ and $B$ are two sets of tuples, then $A \shuffle{I} B = \{ \va \shuffle{I} \vb \mid \va\in A \text{ and } \vb\in B\}$. 
%	For a tuple $\va$ and a permutation $\pi$ on the set $\{ 1, \ldots, \length{\va} \}$, we let $\pi [\va]$ denote the tuple $\vb$ satisfying $a_i = b_{\pi(i)}$ for all $1 \leq i \le \length{\va}$ and call it the \emph{permutation of $\va$ by $\pi$}. For a set $A$ of tuples of some uniform length $m$ and a permutation $\pi$ on $\{ 1, \ldots, m \}$ we let $\pi[A]$ denote the set of all permuted tuples of $A$ by $\pi$, i.e. $\pi[A] = \{ \pi[\va] \mid \va\in A \}$.

\medskip \noindent {\bf Syntax and formulae.}
We let $\mathbf{T}$ denote a set of {\em terms}, defined as the union of three countably infinite, 
mutually disjoint sets of \emph{constants}~$\mathbf{C}$, \emph{nulls}~$\mathbf{N}$, and \emph{variables}~$\mathbf{V}$. We use $\const{a}$, $\const{b}$, $\const{c}$, $\ldots$ (occasionally annotated) to denote constants, and use $x$, $y$, $z$, $\ldots$ (occasionally annotated) to denote both nulls and variables.
% in an ambiguous manner, specifying when a symbol is used to represent one as opposed to the other. 
A \emph{signature}~$\sig$ is a finite set of \emph{predicate symbols} (called \emph{predicates}), which are capitalized ($\epred$, $\rpred$, $\ldots$) with the exception of the special, universal unary predicate $\top$, assumed to include every term. %\footnote{Assuming the presence of such a built-in \emph{domain predicate} $\top$ does not affect our results, but it allows for a simpler and more concise presentation.} 
 Throughout the paper, we assume a fixed finite signature $\sig$, unless stated otherwise. For each predicate~\mbox{$\rpred \in \sig$}, we denote its 
\emph{arity} with~$\arity{\rpred}$. For a finite set of predicates $\Sigma$, we denote by $\arity{\Sigma}$ the maximum over $\arity{\rpred}$ for all $\rpred\in\Sigma$. We speak of \emph{binary signatures} whenever $\arity{\Sigma} \leq 2$. An \emph{atom} over~$\sig$ is an expression  of the form~$\rpred(\vt)$, where $\rpred \in \sig$ and~$\vt$ is an~$\arity{\rpred}$-tuple of terms. A \emph{ground atom} is an atom~$\rpred(\vt)$, where~$\vt$ consists only of constants.
An \emph{instance}~$\instance$ over~$\sig$ is a (possibly countably infinite)
set of atoms over constants and nulls, whereas a~$\emph{database}$~$\database$ is
a finite set of ground atoms. The \emph{active domain}~$\adom{\instance}$ 
of an instance~$\instance$ is the set of terms (i.e., constants and nulls) appearing in the atoms of $\instance$. %Since 
Clearly, databases and instances can be seen as model-theoretic structures,%
\footnote{The ensuing model theory enforces the unique name assumption and allows for constants to be uninter\-preted. These design choices warrant a simpler and homogenized presentation without affecting the results.} 
so we use %the satisfaction relation~%
$\models$ to indicate satisfaction of logical formulae as usual. 
%Further, as instances over binary signatures can be viewed as directed graphs whose nodes and edges are typed with predicates, we often use graph-theoretic terminology when discussing such objects. %% NOT NEEDED HERE, OR IS IT?
 We assume the reader is familiar with the syntax and semantics of first- and second-order logic. We will use upper case greek letters ($\Phi$, $\Psi$,$\,\ldots$) to denote sentences (i.e., formulae without free variables) and lower case ones ($\varphi$, $\psi$,$\,\ldots$) to denote open formulae. We write $\varphi(\vx)$ to indicate that the variable tuple $\vx$ corresponds to the free variables of $\varphi$. By slight abuse of notation, we write $\instance \models \varphi(\vt)$ instead of $\instance [\vx \mapsto \vt] \models \varphi(\vx)$ for tuples $\vt$ of terms with $|\vt|=|\vx|$. We refer to classes of sentences as \emph{fragments} and denote them by gothic letters ($\mathfrak{F}$, $\mathfrak{Q}$, $\mathfrak{L}$, $\ldots$) or, when known and fixed, by their usual abbreviations (FO, SO, MSO,\linebreak GSO, $\ldots$). An overview of the fragments discussed in this article is provided in the appendix. % \ref{app:logicoverview}.

\renewcommand{\smallldots}{...\,}
\medskip \noindent {\bf Homomorphisms and substitutions.} A \emph{substitution} is a partial function over~$\mathbf{T}$.
A~\emph{homomorphism} from a set of atoms~$\mathcal{X}$ to a set of 
atoms~$\mathcal{Y}$ is a substitution~$h$ from the terms of~$\mathcal{X}$ to the terms of~$\mathcal{Y}$ that satisfies:
(i)~$\rpred(h(t_1), \smallldots, h(t_n)) \in \mathcal{Y}$, for all~$\rpred(t_1, \smallldots, t_n) \in \mathcal{X}$%
, 
and 
(ii)~$h(\const{a}) = \const{a}$, for each~$\const{a} \in \mathbf{C}$%
. 
A~homomorphism $h$ is an \emph{isomorphism} \iffi it is bijective and $h^{-1}$ is also a homomorphism. $\mathcal{X}$ and $\mathcal{Y}$ are \emph{isomorphic} (written $\mathcal{X} \cong \mathcal{Y}$) \iffi an isomorphism from $\mathcal{X}$ to $\mathcal{Y}$ exists.

%% NOT NEEDED, I THINK...
%An instance~$\instance'$ is an \emph{induced sub-instance} of an instance $\instance$ \iffi (i) $\instance' \subseteq \instance$ and (ii) if $\rpred(t_1,\smallldots,t_n) \in \instance$ and $t_1,\smallldots,t_n \in \adom{\instance'}$, then $\rpred(t_1,\smallldots,t_n) \in \instance'$. 
\renewcommand{\smallldots}{\ldots}

\medskip \noindent {\bf (Unions of) conjunctive queries and their satisfaction.}
A {\em conjunctive query} (CQ) %, also referred to as a \emph{query}, 
 is a  formula $\query(\vy) = \exists \vx .\phi(\vx, \vy)$ with $\phi(\vx, \vy)$ being a non-empty conjunction (sometimes construed as a set) of atoms over the variables from $\vx,\vy$ and constants.
 The variables from $\vy$ are called {\em free}. A \emph{Boolean} CQ (or BCQ) is a CQ with no free variables.
 A \emph{union of conjunctive queries} (UCQ) $\psi(\vy)$ is a disjunction of CQs with free variables $\vy$. 
 We will treat UCQs as sets of CQs. %{\bf In our paper we will consider mostly Boolean queries (with constants)} with the reason being notational convenience. 
Satisfaction of a BCQ $\query = \exists \vx .\phi(\vx)$ in an instance $\instance$, written $\instance \models q$, coincides with the existence of a homomorphism from $\phi(\vx)$ to $\instance$. As per standard FO semantics, $\instance$ satisfies a union of BCQs, if it satisfies at least one of its disjuncts.
%
%\begin{definition}[Finite Controllability]
Given a FO sentence $\Phi$, we say that it is \emph{finitely controllable} if for every BCQ $q$ with $\Phi \not\models q$ there exists a finite model of $\Phi$ that does not satisfy $q$.
An FO fragment is called finitely controllable if all of its sentences are. 
%\end{definition}

%SECTION 3: GENERIC DECIABILIY A
\section{Semidecidability via Reduction to FO}\label{sec:easy}

As noted in the introduction, our general approach to establish a procedure that decides whether the entailment $\Phi \models \Psi$ holds is to couple two complementary semi-decision procedures: the first terminates successfully whenever $\Phi \models \Psi$ and the second does so whenever $\Phi \not\models \Psi$. In this section, we discuss the former, which is the easier one. Recall that we are interested in entailments of the form $\Phi \models \Psi$, where $\Phi$ is a FO sentence and $\Psi$ is a $\forall$SO sentence, both over some finite signature $\Sigma$. That is, $\Psi$ has the shape $\forall\Sigma'.\Psi^*$ where $\Psi^*$ is a FO sentence over the signature $\Sigma\cup\Sigma'$ 
(with $\Sigma'$ being disjoint from $\Sigma$ and consisting of relation names that are ``repurposed'' as second-order variables in $\Psi$).  

\begin{observation}
$\Phi \models \Psi$ \iffi $\Phi \models \Psi^*$.
\end{observation}
\begin{proof}
Recalling $\Psi = \forall\Sigma'.\Psi^*$, we note that $\Phi \models \forall\Sigma'.\Psi^*$ holds iff $\Phi \wedge  \neg \forall\Sigma'.\Psi^*$ is unsatisfiable. Now consider
$\Phi \wedge  \neg \forall\Sigma'.\Psi^* \ \ \equiv \ \ \Phi \wedge  \exists\Sigma'.\neg\Psi^* \ \ \equiv \ \ \exists\Sigma'.(\Phi \wedge  \neg\Psi^*),$
and note that $\exists\Sigma'.(\Phi \wedge  \neg\Psi^*)$ is unsatisfiable whenever $\Phi \wedge  \neg\Psi^*$ is (through second-order Skolemization). 
The latter coincides with $\Phi \models \Psi^*$, so we have shown the claimed fact.
\end{proof}

Thanks to this insight, we immediately obtain semidecidability of $\textsc{Entailment}(\mathrm{FO},\!\forall\mathrm{SO})$: Exploiting the completeness of FO \cite{God29FOLcomp}, we can enumerate all FO consequences of $\Phi$ and terminate whenever $\Psi^*$ is produced. 

Said relation to FO also helps with co-semidecidability: as every model of $\Phi \wedge \neg\Psi^*$ can immediately serve as a witness of $\Phi \not\models \Psi$, applying the downward part of Löwenheim-Skolem \cite{LSTheorem} confirms that it suffices to focus one's search for countermodels on instances with a countable domain.

%\pagebreak

%SECTION 3: GENERIC DECIABILIY B
\section{Co-Semidecidability via Width Restrictions}\label{sec:genericdecidability}

In this section, we provide a generic argument by means of which co-semidecidability (and consequently decidability) of \textsc{Entailment} can be established in situations where ``structurally well-behaved'' countermodels exist.  
As yardsticks for ``structural well-behavedness'', we will employ various measures that were originally defined (only) for finite graphs, yet, can be generalized to countable instances.  

\begin{definition}[width measures and their properties]\label{def:MSOfriendly}
A \emph{width measure} is a function $\fcn{w}$ mapping every countable instance over some finite signature to a value from $\mathbb{N} \cup \{\infty\}$. An instance $\instance$ is called $\fcn{w}$-finite if $\fcn{w}(\instance)\in \mathbb{N}$.
A width measure $\fcn{w}$ \emph{generalizes} a width measure $\widetilde{\fcn{w}}$ \iffi
%$\fcn{w}^{-1}(\infty) \subseteq \widetilde{\fcn{w}}^{-1}(\infty)$
every $\widetilde{\fcn{w}}$-finite instance is $\fcn{w}$-finite.    
A fragment $\mathfrak{F}$ has the \emph{finite-$\fcn{w}$-model property} \iffi every satisfiable $\Phi \in \mathfrak{F}$ has a $\fcn{w}$-finite model. 
Given a fragment $\mathfrak{L}$, a width measure $\fcn{w}$ is called \emph{$\mathfrak{L}$-friendly} \iffi  
there exists an algorithm that, taking a number $n \in \mathbb{N}$ as input, enumerates all $\Xi \in \mathfrak{L}$ 
that have a model $\instance$ with $\fcn{w}(\instance)\leq n$.\defend
\end{definition}

Observe that $\mathfrak{L}$-friendliness implies $\mathfrak{L}'$-friendliness for every $\mathfrak{L}'\subseteq\mathfrak{L}$. 
Next, we will briefly review known width notions and corresponding friendliness results. A summary can be found in \Cref{tab:friendly} which also includes the notion of partitionwidth, which will be defined and thoroughly discussed in \Cref{sec:partitionwidth}. 

\begin{table}%[b]
	\begin{center}
		\begin{tabular}{lcccc}
			width measure  & SO-friendly & GSO-friendly & MSO-friendly & FO-friendly \\\hline
			expansion & yes & yes & yes & yes \\ 
			treewidth & no & yes & yes & yes \\
			cliquewidth & no & no & yes & yes \\
			partitionwidth & no & no & yes & yes \\\hline
		\end{tabular}
		\caption{Width notions discussed in this article and their friendliness toward certain logics.\label{tab:friendly}}
	\end{center}
\end{table}

\medskip \noindent {\bf Expansion (domain size).}
As a very simple example, consider the \emph{expansion} function\footnote{The name is chosen for homogeneity, so it fits the pre-existing notion of \emph{finite expansion sets} discussed in \Cref{sec:finitewsets}. There is no relationship to the notion of expansion in graph theory.} $\fcn{expansion}: \inst \mapsto |\adom{\inst}|$, mapping every countable instance to the size of its domain. Clearly, the finite-$\fcn{expansion}$-model property is just the well-known finite model property. Moreover, $\fcn{expansion}$ is an SO-friendly width measure: there are only finitely many instances with $n$ elements (up to isomorphism) which can be computed and checked.

\medskip \noindent {\bf Treewidth.}
A somewhat less trivial example of a width measure is the popular notion of \emph{treewidth} \cite{ROBERTSON198449}.

\begin{definition}[treewidth]
A \emph{tree decomposition} of an instance $\inst$ is a (potentially infinite) tree $\tree = (V,E)$, where:
\begin{itemize}

\item $V \subseteq 2^{\adom{\inst}}$, i.e., each node $X \in V$ is a set of terms of $\inst$, 
\item $\bigcup_{X \in V} X = \adom{\inst}$,

\item $\rpred(t_{1}, ...\,,t_{n}) \in \inst$ implies $\{t_{1}, ...\,,t_{n}\} \subseteq X$ for some $X \in V$,

\item for each term $t$ in $\inst$, the subgraph of $T$ consisting of the nodes $X$ with $t \in X$ is connected.
\end{itemize}

The \emph{treewidth} of a tree decomposition $\tree = (V,E)$ is set to $\mathrm{max}_{X \in V} |X| - 1 $ if such a finite maximum exists and $\infty$ otherwise. The \emph{treewidth} of an instance $\inst$, denoted $\treewidth{\inst}$, is the minimal treewidth among all its tree decompositions. \defend
\end{definition}

Treewidth is a GSO-friendly width measure \cite[Thm. 9.18]{CourEngBook} and generalizes $\fcn{expansion}$.

\pagebreak

\medskip \noindent {\bf Cliquewidth.}
Yet another example of a width measure is the notion of \emph{cliquewidth} \cite{ICDT2023}, denoted $\fcn{cw}$, whose technical description is somewhat involved and constitutes a special case of the to-be-introduced notion of partitionwidth, so we defer the formal definition to \Cref{sec:hcw-vs-tw-cw}. The version of cliquewidth employed by us works for arbitrary countable structures properly generalizes earlier notions of cliquewidth for finite directed edge-labelled graphs \cite{CourEngBook} and countable un\-label\-led un\-di\-rec\-ted graphs \cite{Cou04}, as well as 
Grohe and Turán's cliquewidth notion for finite instances of arbitrary arity \cite{GroheT04}. %, the notion of multi-cliquewidth introduced by Fürer \cite{Furer15}.
Cliquewidth is MSO-friendly \cite[Thm. 11]{ICDT2023} and generalizes $\fcn{expansion}$, whereas neither of $\fcn{cw}$ and $\fcn{tw}$ generalizes the other (see also \Cref{fig:relationships} in \Cref{sec:hcw-vs-tw-cw}). 

\bigskip

The next two definitions capture the two core ingredients of our framework. The first ingredient is the guaranteed existence of width-finite countermodels.

\begin{definition}[width-based controllability]
Let $\fcn{w}$ be a width measure and let $\mathfrak{Q}$ be a fragment.
A fragment $\mathfrak{F}$ is called \emph{$\fcn{w}$-controllable for $\mathfrak{Q}$}  (briefly noted as $\mathrm{Cont}_\fcn{w}(\mathfrak{F},\mathfrak{Q})$) \iffi for every $\Psi \in \mathfrak{Q}$ and  $\Phi \in \mathfrak{F}$ with $\Phi \not \models \Psi$, there exists a $\fcn{w}$-finite model $\instance^* \models \Phi$ with $\instance^* \not\models \Psi$.\defend
\end{definition}

%Note that the finite width required by this definition does not need to be uniformly bounded: it may depend on both $\Phi$ and $\Psi$ and thus grow beyond any finite bound.

It follows from the above definitions that  $\mathrm{Cont}_\fcn{w}(\mathfrak{F},\mathfrak{Q})$ implies $\mathrm{Cont}_{\widetilde{\fcn{w}}}(\mathfrak{F}'\!,\mathfrak{Q}')$ whenever 
$\widetilde{\fcn{w}}$~generalizes~$\fcn{w}$, $\mathfrak{F}'\! \subseteq \mathfrak{F}$, and $\mathfrak{Q}'\! \subseteq \mathfrak{Q}$.   
The other ingredient deals with the existence of an equivalent translation between fragments.

\begin{definition}[effective expressibility]
Given fragments $\mathfrak{L}$ and~$\mathfrak{L}'$, we call $\mathfrak{L}$ \emph{effectively $\mathfrak{L}'$-expressible} \iffi there exists some computable function $\fcn{eqtrans} : \mathfrak{L} \to \mathfrak{L}'$ that satisfies $\Psi \equiv \fcn{eqtrans}(\Psi)$. 
\defend
\end{definition}

We next formulate this article's core decidability result.

%\begin{theorem}\label{thm:decidablequeryingold}
%Let $\mathfrak{F}$ be an FO fragment and let $\mathfrak{L} \supseteq \mathfrak{F}$ be an SO fragment closed under Boolean connectives.  
%Let $\fcn{w}$ be an $\mathfrak{L}$-friendly width measure, $\mathfrak{Q}$ be an effectively $\mathfrak{L}$-expressible $\forall$SO fragment, and $\mathfrak{F}$ be $\fcn{w}$-controllable for $\mathfrak{Q}$.
%Then, $\textsc{Entailment}(\mathfrak{F},\mathfrak{Q})$ is decidable.
%\end{theorem}

\begin{theorem}\label{thm:decidablequerying}
Let $\fcn{w}$ be a width measure,
$\mathfrak{F}$ an FO fragment, and 
$\mathfrak{Q}$ a $\forall$SO fragment 
such that $\mathfrak{F}$ is $\fcn{w}$-controllable for $\mathfrak{Q}$.
If there exists some fragment $\mathfrak{L}$ such that $\fcn{w}$ is $\mathfrak{L}$-friendly
and the fragment $\{\Phi \wedge \neg \Psi \mid \Phi \in \mathfrak{F}, \Psi \in \mathfrak{Q}\}$ is effectively 
$\mathfrak{L}$-expressible, then $\textsc{Entailment}(\mathfrak{F},\mathfrak{Q})$ is decidable.
\end{theorem}

\begin{proof}
Detecting $\Phi \models \Psi$ is semidecidable, as discussed earlier.
Hence it remains to establish a semidecision procedure for $\Phi \not\models \Psi$. 
Thanks to the assumption that $\Phi$ is from a fragment that is $\fcn{w}$-controllable for $\mathfrak{Q}$ with $\Psi \in \mathfrak{Q}$, it follows by definition that there exists a $\fcn{w}$-finite model $\instance^* \models \Phi$ with $\instance^* \not\models \Psi$.
Note that such a ``countermodel'' $\instance^*$ can be characterized by the $\mathfrak{L}$-sentence $\fcn{eqtrans}(\Phi \wedge \neg\Psi)$.  %$\Phi \wedge \neg\fcn{trans}(\Psi)$. 
 Therefore, the following procedure will terminate, witnessing non-entailment: Enumerate all natural numbers in their natural order and for each number $n$ run a parallel thread with the enumeration algorithm from \Cref{def:MSOfriendly} with input $n$ (the algorithm is guaranteed to exist due to $\mathfrak{L}$-friendliness of $\fcn{w}$). Terminate as soon as one thread produces $\fcn{eqtrans}(\Phi \wedge \neg\Psi)$.\qedhere
\end{proof}

\Cref{fig:algo} displays the corresponding decision algorithm. We note that the parallel execution of countably infinitely many ``threads'' in (2) can be realized by a technique commonly known as \emph{dovetailing}: execute one computation step for the first thread, then execute one step for threads 1 and 2, then one step for each of threads 1,2, and 3, etc. This yields a fair execution, meaning that every computation step of every thread will be executed after a finite amount of time. In particular, this execution strategy ensures that the parallel branch (2) of the algorithm terminates whenever a width-finite model of $\fcn{eqtrans}(\Phi \wedge \neg\Psi)$ exists. One should note that, beyond mere decidability, the generic framework does not give rise to any informative upper complexity bounds.

%As many of the applications of \Cref{thm:decidablequerying} will use $\mathfrak{L} \supseteq \mathrm{FO}$, 

\newcommand{\indt}{$\left.\right.$\,}
\begin{table}
\probbox{\ \\[-0ex]
	\indt\textbf{Algorithm:} $\textsc{EntailmentCheck}(\fcn{w},\fcn{eqtrans})$\\
	\indt\textbf{Requires:} $\mathrm{Cont}_\fcn{w}(\mathfrak{F},\mathfrak{Q})$, $\mathfrak{L}$-friendly $\fcn{w}$, $\fcn{eqtrans}$\\
	\indt\textbf{Input:} $\Phi \in \mathfrak{F} \subseteq \mathrm{FO}$, \ \  $\forall\Sigma'.\Psi^* \in \mathfrak{Q} \subseteq \forall\mathrm{SO}$.\\
	\indt\textbf{Output:} \textsc{yes}, if $\Phi \models \Psi$, \textsc{no} otherwise.\\[1ex]
\indt Run the following in parallel:\\
\indt (1) enumerate all $\Theta \in \mathrm{FO}$ with $\Phi \models \Theta$\\
\indt $\left.\right.$\ \ \ \ \ \ if $\Theta = \Psi^*$, terminate with \textsc{YES},\\
\indt $\left.\right.$\ \ \ \ \ \ next $\Theta$.\\
\indt (2) for all $n \in \mathbb{N}$, run the following in parallel:\\
\indt $\left.\right.$\ \ \ \ \ enumerate all $\Xi \in \mathfrak{L}$ having model $\instance$ with \indt $\fcn{w}(\instance)\leq n$,\\ 
\indt $\left.\right.$\ \ \ \ \ \ if $\Xi = \fcn{eqtrans}(\Phi \wedge \neg\Psi)$, terminate with \textsc{NO},\\
\indt $\left.\right.$\ \ \ \ \ \ next $\Xi$.\\[-1.5ex]
}
\caption{Entailment decision algorithm discussed in \Cref{thm:decidablequerying}. 
% We note that the parallel execution of countably infinitely many ``threads'' in (2) can be realized by a technique commonly known as \emph{dovetailing}: execute one computation step for the first thread, then execute one step for threads 1 and 2, then one step for each of threads 1,2, and 3, etc. This yields a fair execution, meaning that every computation step of every thread will be executed after a finite amount of time. In particular, this execution strategy ensures that the parallel branch (2) of the algorithm terminates whenever a width-finite model of $\fcn{eqtrans}(\Phi \wedge \neg\Psi)$ exists. 
\label{fig:algo}}
\end{table}

\footnotetextbefore{Obviously, decidability entails the existence of a total recursive upper-bound function obtained from exhaustively evaluating all size-bounded inputs.}
\begin{example}\label{ex:ALCHOIQb}
	The proof of the previously long-standing open problem of decidability of UCQ entailment in the description logic $\mathcal{ALCHOIQ}b$ \cite{RudolphG10} can be seen as an application of \Cref{thm:decidablequerying} with $\mathfrak{F} = \mathcal{ALCHOIQ}b$ and $\mathfrak{Q} = \mathrm{UCQ}$, where $\mathfrak{L} = \mathrm{FO}$ and $\fcn{w} = \fcn{tw}$, i.e., treewidth is the used width notion. The most intricate part of this result amounts to showing that $\mathcal{ALCHOIQ}b$ is $\fcn{tw}$-controllable for $\mathrm{UCQ}$, via an elaborate countermodel construction\remarkstart, producing so-called \emph{quasi-forest models}, which are known to be $\fcn{tw}$-finite. \color{black} To this day, no informative upper complexity bound for $\textsc{Entailment}(\mathcal{ALCHOIQ}b,\mathrm{UCQ})$ is known.\footnote{Obviously, decidability entails the existence of a total recursive upper-bound function obtained from exhaustively evaluating all size-bounded inputs.}
\end{example}

\begin{remark}
The attentive reader might wonder whether the proposed setting is unnecessarily complicated: Wouldn't it be possible to use the established reduction from the second-order entailment $\Phi \models \Psi$ with $\Psi = \forall\Sigma'.\Psi^*$ to the first-order entailment $\Phi \models \Psi^*$ for the co-semidecision part, too? This would allow one to always pick $\mathfrak{L} = \mathrm{FO}$ and $\fcn{eqtrans} : \Phi \wedge \neg \Psi \mapsto \Phi \wedge \neg \Psi^*$. The answer is no:
while $\Phi \not\models \Psi$ coincides with $\Phi \not\models \Psi^*$, the former may have a width-finite countermodel while the latter does not.

For example, let $\Phi$ be the sentence $\exists z.\forall x.\exists y.\pred{S}(x,y)$, expressing that there is at least one element and that every element has an $\pred{S}$-successor, and 
let $\Psi$ be a query checking for the existence of a directed $\pred{S}$-cycle: $\Psi = \forall \pred{P}.\Psi^*$ where
$\Psi^*$ is 
$$
\begin{array}{@{}c@{}}
\big(\forall xy.(\pred{S}(x,y){\to} \pred{P}(x,y)) \wedge \forall xyz.(\pred{P}(x,y) {\wedge} \pred{P}(y,z) {\to} \pred{P}(x,z))\big)
\  \to  \  \exists x.\pred{P}(x,x).
\end{array}
$$
Then we can see that $\Phi \not\models \Psi$ is witnessed by the countermodel $\instance = \{\pred{S}(n,n+1) \mid n \in \mathbb{N}\}$ which has treewidth of $1$, whereas any countermodel witnessing 
$\Phi \not\models \Psi^*$ has infinite treewidth. To see this, note that $\instance \cup \{\pred{P}(n,m) \mid n < m \in \mathbb{N}\}$ is a countermodel that has infinite treewidth and it is minimal in the sense that it admits an injective homomorphism into any other countermodel, whence treewidth-infinity carries over.
\end{remark}

\begin{remark}
The attentive and well-informed reader might question our restriction of $\mathfrak{F}$ to fragments of FO. Indeed, all the width notions $\fcn{w}$ considered in this paper share the property that they are not only MSO-friendly, but in fact, for any given $n \in \mathbb{N}$, it is even decidable if an MSO sentence $\Phi$ has a model $\mathcal{I}$ with $\fcn{w}(\mathcal{I})\leq n$ or not. So could we not exploit this property to design a decision procedure for $\textsc{Entailment}(\mathfrak{F},\mathfrak{Q})$ without reverting to completeness of FO as laid out in \Cref{sec:easy}? In that case we could significantly extend our setting by letting $\mathfrak{F}$ be some fragment of MSO (rather than just FO) and $\mathfrak{Q}$ be a selection of MSO sentences preceded by a prefix of (higher-ary) universal second-order quantifiers. Unfortunately, however, this is generally not possible. For this approach to work, we would need to know an a-priori upper bound on a potential countermodel's width, that is, a computable function $\fcn{maxwidth}: \mathfrak{F} \times \mathfrak{Q} \to \mathbb{N}$ satisfying that, for every $\Phi \in \mathfrak{F}$, $\Psi \in \mathfrak{Q}$ with $\Phi \not\models\Psi$, there exists a corresponding countermodel $\mathcal{I}$ with $\fcn{w}(\mathcal{I})\leq \fcn{maxwidth}(\Phi,\Psi)$. While in many cases such a function is easy to establish (e.g., for GF), this does not hold in general. In some cases, such a function does provably not exist, even if $\mathfrak{F}$ is $\fcn{w}$-controllable for $\mathfrak{Q}$. As a basic example, let $\fcn{w}$ be $\fcn{expansion}$ (i.e., the domain size), let $\mathfrak{Q} = \{\mathit{false}\}$ and let $\mathfrak{F}$ consist of all sentences of the form $\Phi \wedge \Phi'$ where $\Phi$ is an arbitrary FO sentence not using the binary predicate $\pred{S}$, while $\Phi'$ is the fixed MSO sentence (using some MSO-expressible ``macros'' to enhance readability)
$$
\begin{array}{r}
\exists x,y.\Big( \big(\neg \exists v.\pred{S}(v,x)\big) \wedge \big(\neg \exists v.\pred{S}(y,v)\big) \wedge  \big(\forall z {\not=} y. \exists ! z'. \pred{S}(z,z')\big) \wedge \big(\forall z {\not=} x. \exists ! z'. \pred{S}(z',z)\big)\quad \\
\wedge\ \big(\forall v. \forall X {\ni} x. (\forall z,z'. X(z) \wedge \pred{S}(z,z') \,{\to}\, X(z')) \, \to \, X(v) \big)\quad \\
\wedge\ \big(\forall v. \forall X {\ni} v. (\forall z,z'. X(z) \wedge \pred{S}(z,z') \,{\to}\, X(z')) \, \to \, X(y) \big)\Big), \\
\end{array}
$$
which forces every model's domain to be a finite total order with successor relation~$\pred{S}$. Due to this, every countermodel for $\Phi \wedge \Phi' \models \mathit{false}$ must have finite domain size and consequently we have established that $\mathfrak{F}$ is $\fcn{expansion}$-controllable for $\{\mathit{false}\}$. Now, toward a contradiction, suppose a computable function $\fcn{maxwidth}$ as above existed for our case. Then we could use it to decide finite unsatisfiability of any $\pred{S}$-free FO sentence $\Phi$ by deciding the entailment $\Phi \wedge \Phi' \models \mathit{false}$ using the strategy sketched above. Yet, finite satisfiability of FO sentences is known to be undecidable by Trakhtenbrot's Theorem~\cite{trakhtenbrot50}.
\end{remark}

Please note that the \textsc{EntailmentCheck} algorithm from \Cref{fig:algo} could also be run on a ``best effort basis'' with very little context information: for every $\Phi$, $\Psi$, and $\mathfrak{L}$-translation of $\Psi$, one could execute it using the most general known width notion $\fcn{w}$ available (with partitionwidth being the recommended choice). Whenever the algorithm terminates, it will provide the correct answer. It will only fail to terminate in case $\Phi \not\models \Psi$ with every countermodel having infinite width. 

We demonstrate the generality of our width-related notion of control\-la\-bility by giving a few examples of scenarios covered by it. $\mathfrak{F}$ is $\fcn{w}$-controllable for $\mathfrak{Q}$ in the following cases:

\begin{itemize}
\item $\mathfrak{F}$ has the finite-$\fcn{w}$ model property and $\mathfrak{Q}$ contains only unsatisfiable sentences. This coincides with the task of (un)satisfiability checking in $\mathfrak{F}$ and (for any width notion discussed here) subsumes all fragments known to have the finite-model property, which, among others, is the case for 
all description logics subsumed by $\mathcal{SROI}b_s$,
the majority of the prefix classes of FO with a decidable satisfiability problem \cite{BorgerGG1997}, 
the two-variable fragment of FO (FO$^2$~\cite{Mortimer75,GradelKV97}),  
the guarded fragment (GF~\cite{AndrekaNB98,Gradel99}), 
the unary negation fragment (UNF~\cite{SegoufinC13}), 
the guarded negation fragment (GNF~\cite{GuardedNegation})
and the triguarded fragment (TGF~\cite{RudolphS18,KieronskiR21}). 
An interesting case in this respect is the counting two-variable fragment of FO ($\mathcal{C}^2$~\cite{GradelOR97}). It is known to not have the finite-model property and can be easily seen to not have the finite-$\fcn{tw}$ model property either. We conjecture, however, that it has the finite-$\fcn{cw}$-model property. 
\item $\fcn{w}$ generalizes $\fcn{expansion}$,  $\mathfrak{Q}$ contains only conjunctive queries and $\mathfrak{F}$ is finitely controllable. This applies to description logics subsumed by $\mathcal{ALCHOI}b$, and to GF, UNF,\linebreak and GNF~\cite{BaranyGO13}. Contrariwise, CQ entailment for FO$^2$\!, $\mathcal{C}^2$\!, and TGF is known to be undecidable \cite{Rosati07}, so these logics must fail to be $\fcn{w}$-controllable for CQs and any FO-friendly $\fcn{w}$.       
%\item $\mathfrak{F}$ admits universal models of finite $\fcn{w}$ and $\mathfrak{Q}$ contains only sentences that are preserved under homomorphisms, (FORWARD REFERENCE!)
\end{itemize}
In the next section, we discuss an important further case, by means of which we can establish controllability for query languages way beyond (U)CQs.   

\section{The Case of Homomorphism-Closed Queries}\label{sec:homclosedqueries}

While the presented framework provides a uniform way of showing decidability of entailment problems, establishing the preconditions of \Cref{thm:decidablequerying}, and in particular controllability, can be arbitrarily difficult and may require ad hoc proof techniques in many cases.
It is therefore desirable to single out sufficient conditions for controllability that cover comprehensive specification and query languages. It turns out that fragments which are ``homomorphically benign'' in a certain sense strike a good balance in that respect.

\begin{definition}[(finite) universality, homomor\-phism-closedness]
Let $\mathscr{I}$ be a set of models of a sentence $\Phi$.
$\mathscr{I}$ is \emph{universal}, \iffi for every model $\mathcal{J} \models \Phi$ there is some $\instance \in \mathscr{I}$ that maps homomorphically into $\mathcal{J}$.
$\mathscr{I}$ is \emph{finitely universal}, \iffi for every model $\mathcal{J} \models \Phi$ there is some $\instance \in \mathscr{I}$ such that every finite subset of $\instance$ maps homomorphically into $\mathcal{J}$.  
%
%A set $\mathscr{I}$ of instances is called a \emph{(finitely) universal model set} of a sentence $\Phi$, if
%\begin{itemize}
%    \item $\instance \models \Phi$ for every $\instance \in \mathscr{I}$ and
%    \item for every model $\mathcal{J} \models \Phi$ there is some $\instance \in \mathscr{I}$ such that there is a homomorphism from (every finite subset of) $\instance$ to $\mathcal{J}$.
%\end{itemize}
%\tim{if $\mathscr{I} = \{\instance\}$, we call $\instance$ a \emph{(finitely) universal model} of $\Phi$.} %
In case $\mathscr{I} = \{\instance\}$, we call $\instance$ a \emph{(finitely) universal model} of $\Phi$.

A sentence $\Psi$ is called \emph{homomorphism-closed} \iffi whenever a model $\instance \models \Psi$ admits a homomorphism into some instance $\mathcal{J}$, then $\mathcal{J}$ is also a model of $\Psi$. \defend
\end{definition}
Clearly, every universal model (set) is also finitely universal.
The following Theorem is folklore for universal models and well-known for universal model sets \cite{BourhisMMP16}.
The generalization using finite universality is novel and a bit more involved.
\begin{theorem}[finite universality and homclosedness]\label{thm:universal}
Let $\Phi$ be an FO sentence having some finitely universal model set $\mathscr{I}$ and let $\Psi$ be a homomorphism-closed $\forall$SO sentence. Then, $\Phi \not \models \Psi$ \iffi there exists a model $\instance \in \mathscr{I}$ with $\instance \not\models \Psi$.
\end{theorem}
The proof of this theorem (see appendix) hinges on the insight that every model of some homomorphism-closed $\forall$SO sentence has a fi\-nite submodel.
This same insight can also be utilized to show the following statement (also, see the appendix for the proof), which can be understood as a sort of interpolation result and might be interesting in itself.
\begin{lemma}\label{lem:interpolation}
Let $\Phi$ be a FO sentence and $\Psi$ a homomorphism-closed $\forall$SO sentence.
Then, $\Phi \models \Psi$ holds \iffi there exists a union of BCQs $Q = \bigvee_{i=1}^n q_{i}$ satisfying $\Phi \models Q$ and $Q \models \Psi$.
\end{lemma}

\Cref{thm:universal} ensures that for homomorphism-closed queries, the search for countermodels can be confined to any finitely universal model set of $\Phi$. That means, if some finitely universal model set can be guaranteed to consist only of structurally well-behaved models, we can ensure controllability.\footnote{We note that finite universality is demonstrably better than universality: there are known cases where $\Phi$ has a finitely universal model of finite width yet any universal model has infinite width \cite{BMR2023}.} 
\begin{theorem}\label{thm:univmodelset}
Let $\fcn{w}$ be a width measure.
Let $\mathfrak{Q}$ be a $\forall$SO fragment of homomorphism-closed sentences and let $\mathfrak{F}$ be a FO fragment of which every sentence exhibits some finitely universal model set containing only $\fcn{w}$-finite structures. 
Then, $\mathfrak{F}$ is $\fcn{w}$-controllable for $\mathfrak{Q}$.
\end{theorem}
\begin{proof}
Consider arbitrary $\Psi \in \mathfrak{Q}$ and  $\Phi \in \mathfrak{F}$ with $\Phi \not \models \Psi$. By assumption, there exists a finitely universal model set $\mathscr{I}$ of $\Phi$ containing only $\fcn{w}$-finite structures and by \Cref{thm:universal}, there must be some $\instance^* \in \mathscr{I}$ with $\instance^* \not\models \Psi$. Thus $\mathfrak{F}$ is $\fcn{w}$-controllable for $\mathfrak{Q}$.
\end{proof}
Combining \Cref{thm:univmodelset} and \Cref{thm:decidablequerying}, we directly obtain the subsequent corollary.
\begin{corollary}\label{cor:dechomquerying}
	Let $\fcn{w}$ be a width measure,
	$\mathfrak{Q}$ a $\forall$SO fragment of homomorphism-closed sen\-ten\-ces, and
	$\mathfrak{F}$ an FO fragment whose every sentence has a finitely universal model set of only $\fcn{w}$-finite instances.
	If there is a fragment $\mathfrak{L}$ such that $\fcn{w}$ is $\mathfrak{L}$-friendly
	and the fragment $\{\Phi \wedge \neg \Psi \mid \Phi {\,\in\,} \mathfrak{F}, \Psi {\,\in\,} \mathfrak{Q}\}$ is effectively 
	$\mathfrak{L}$-expressible, then $\textsc{Entailment}(\mathfrak{F},\mathfrak{Q})$ is decidable.
\end{corollary}
With this corollary in place, we will next review the expressivity boundaries of the query languages for which decidability can be warranted, depending on the underlying width measures' friendliness properties. 
Note that there is a trade-off: less general width measures allow for less expressive choices for the specification language $\mathfrak{F}$ but -- due to being generally friendlier -- admit more expressive query languages $\mathfrak{Q}$.

\medskip \noindent {\bf MSO-friendly width measures.}
In this case (which comprises all width measures discussed in this article), the most expressive query language is the semantically defined class $\mathfrak{Q}$ of all \emph{h}omomorphism-closed \emph{u}niversal \emph{s}econd-\emph{o}rder / \emph{m}onadic \emph{s}econd-\emph{o}rder \emph{q}ueries -- henceforth abbreviated as HUSOMSOQs (pronounced ``hue-somm-socks''). This begs the question how comprehensive and expressive this class actually is and, in particular, which popular syntactically defined query languages it encompasses. An effective syntax characterizing the entire class is still elusive, but it subsumes two very comprehensive and expressive logical fragments (and through these many other query languages):

\begin{itemize}
\item
%\medskip \noindent 
{\bf Nested monadically defined queries.}
This well-investigated fragment (abbreviated NEMODEQs) constitutes a very expressive yet computationally manageable formalism, admitting polytime translations into MSO and datalog (and hence homomorphism-closed $\forall$SO) \cite{RudKro13}.
We will refrain from recalling the somewhat involved syntax here and instead highlight that NEMODEQs subsume (unions of) Boolean CQs,
\emph{monadic datalog queries} \cite{Cosmadakis88}, \emph{(unions of) conjunctive 2-way regular path queries}~((U)C2RPQs)~\cite{FlorescuLS98}, and nested versions thereof (e.g. \emph{regular queries}~\cite{ReutterRV17} as well as others~\cite{BourhisKR14}).

\item
%\medskip \noindent 
{\bf Monadic disjunctive datalog queries.}
 As monadic disjunctive datalog (MDDlog) queries (see \Cref{app:logicoverview}) are ``natively'' contained both in $\forall$SO and in MSO, we can pick the identity function for $\fcn{eqtrans}$ %$\fcn{trans}$
    to witness effective MSO-expressibility. Homomorphism-closedness follows from the containedness in disjunctive datalog. MDDlog queries can verify complex properties such as an instance's non-membership in arbitrary finite-domain CSPs (e.g., non-3-colorability), but are also known to capture many ontology-mediated queries based on description logics \cite{BienvenuCLW14,FeierKL19}.
\end{itemize}
In terms of expressivity, the classes of MDDlog queries and NEMODEQs are incomparable \remarkstart{}(unlike NEMODEQ, MDDlog cannot recognize graphs with directed cycles; unlike MDDlog, NEMODEQ cannot recognize non-3-colorable graphs)\color{black}. The semantic intersection of the two classes contains UCQs and monadic datalog queries.

For comprehensive investigations into FO fragments that allow for decidable HUSOMSOQ entailment on the grounds of MSO-friendly width measures, we point the reader to \Cref{ex:srfamily} and \Cref{ex:Rbio} later in the article.

\medskip \noindent {\bf GSO-friendly width measures.}
In this case, where treewidth is the prominent corresponding width notion, we can establish decidability for all homomorphism-closed queries that are simultaneously expressible in $\forall$SO and GSO. A syntactic characterization of this entire class seems not within reach, despite
some interesting results regarding the characterization of homomorphism-closed GSO in terms of infinite-domain CSPs \cite[Cor. 21]{abs-2010-05677}.   
The arguably most expressive known effectively GSO-expressible homomorphism-closed query languages in $\forall$SO are \emph{frontier-guarded disjunctive datalog} (FGDDlog) on one hand, and on the other hand \emph{nested frontier-guarded flag\,\&\,check queries}, sometimes shortened to \emph{nested guarded queries} and usually abbreviated by GQ$^+$ \cite{BourhisKR15}. 

FGDDlog is syntactically contained in $\forall$SO and its expressibility in GSO is a straightforward consequence of GSO's alternative characterization via the guarded semantics. The homomorphism-closedness of FGDDlog queries follows from their inclusion in disjunctive datalog.
FGDDlog contains all MDDlog queries (see above) and their subclasses. 

GQ$^+$ is expressible in datalog (hence homomorphism-closed and in $\forall$SO) \cite{BourhisKR15} and has also been shown to be expressible in GSO \cite[Prop. 47]{abs-2010-05677}. GQ$^+$ subsumes the class of NEMODEQs (see above) and all its subclasses.    

\footnotetextbefore{It should be noted that this corollary does not cover the case $\mathfrak{F} = \mathcal{ALCHOIQ}b$ from \Cref{ex:ALCHOIQb}, since the employed construction to convert an arbitrary countermodel $\instance$ into a finite-treewidth countermodel $\mathcal{J}$ does not ensure that $\mathcal{J}$ homomorphically maps into $\instance$. In fact, the decidability of the entailment of queries beyond UCQs from $\mathcal{ALCHOIQ}b$ is wide open.}
\begin{example}
	For numerous logical formalisms, the existence of treewidth-finite (finitely) universal model sets can be established through well-known model-transformations: For GF, UNF, and GNF %, and the guarded 2-variable fragment with counting ($\mathcal{GC}^2$~\cite{Pratt-Hartmann07}) 
	as well as a variety of very expressive description logics \cite{baader_horrocks_lutz_sattler_2017,Rudolph11}, specific techniques commonly referred to as \emph{unravelling} transform an arbitrary model $\instance$ into a tree\-width-fi\-nite model $\mathcal{J}$ that homomorphically maps into $\instance$ \cite{GradelHO02,BenediktBB16,BednarczykR19,BednarczykR23,CalvaneseEO09,BourhisKR14}. 
	Thus, due to GSO-friendliness of  
	%$\fcn{tw}$ 
	treewidth, \Cref{cor:dechomquerying} ensures in one sweep decidability of entailment of all effectively GSO-ex\-pres\-sible homomorphism-closed $\forall$SO queries for all these formalisms. 
	
	Consequently, if $\mathfrak{F}$ is any of GF, UNF, GNF, $\mathcal{ALCHOI}b$, $\mathcal{ALCHIQ}b$, $\mathcal{ALCHOQ}b$, and $\mathfrak{Q}$ is some effectively GSO-expressible fragment of homomorphism-closed $\forall$SO queries, then $\mathfrak{F}$ is treewidth-controllable for $\mathfrak{Q}$ and thus $\textsc{Entailment}(\mathfrak{F},\mathfrak{Q})$ is decidable.\footnote{It should be noted that this corollary does not cover the case $\mathfrak{F} = \mathcal{ALCHOIQ}b$ from \Cref{ex:ALCHOIQb}, since the employed construction to convert an arbitrary countermodel $\instance$ into a finite-treewidth countermodel $\mathcal{J}$ does not ensure that $\mathcal{J}$ homomorphically maps into $\instance$. In fact, the decidability of the entailment of queries beyond UCQs from $\mathcal{ALCHOIQ}b$ is wide open.}
	
	While for simple instances of $\mathfrak{Q}$ like CQs or queries with regular path expressions, decidability of $\textsc{Entailment}(\mathfrak{F},\mathfrak{Q})$ has been known for all these instances of $\mathfrak{F}$, our result not only provides a generic common ground for these prior results, but also generalizes them to the above-mentioned more expressive query languages of FGDDlog and GQ$^+$. 
\end{example}

%\pagebreak 

\medskip \noindent {\bf SO-friendly width measures.}
In this case, the semantic definition of the most expressive eligible query language simplifies to that of all homomorphism-closed $\forall$SO sentences. A very expressive syntactically defined query language that falls in this category is \emph{disjunctive datalog} (DDlog) \cite{EiterGM97}, which includes all of datalog (Dlog). 
To the best of our knowledge, no syntactic characterization of homomorphism-closed $\forall$SO has been established (whereas one has been proposed for homomorphism-closed SO as a whole \cite[Cor. 61]{BodirskyFKR21}).
So far, we are not aware of any query in homomorphism-closed $\forall$SO that is provably not expressible in DDlog, but we still find it unlikely that DDlog captures all of homomorphism-closed $\forall$SO.

%SECTION 4: HYPERCLIQUEWIDTH
%\pagebreak
\section{\Partitionwidth{}: Definition and Properties}\label{sec:partitionwidth}

The notion of \emph{\partitionwidth{}} by Blumensath \cite{Blumensath03,Blumensath06} is a parameterized measure originally developed for model-theoretic structures of arbitrary cardinality and with countable signatures. Its purpose is to classify such structures according to the ``simplicity'' of their MSO theory. For the sake of understandability, we will refrain from introducing the notion in its full generality and confine ourselves to a restricted and notation-adjusted version applicable to countable instances over finite signatures, which is sufficient for our purposes. 

Before going into the full formal details in \Cref{sec:cwdef}, let us provide an intuitive description of the underlying high-level ideas.
The notion of some instance's partitionwidth is based on the idea that we iteratively construct or assemble said instance 
from its singleton elements (e.g., terms). To this end, we utilize \emph{colors} of various arities (mostly denoted with lowercase greek latters $\kappa$, $\mu$, $\lambda$) which we dynamically assign to element tuples of the appropriate length. Initially, each singleton element $e$ is assigned the original unary color $\origcol_1$, the pair $(e,e)$ the original binary color $\origcol_2$, and so forth. During the assembly process, color assignments can be extended and manipulated in various ways. However, at any time, any tuple can only carry at most one color. Similar to the mechanism of assigning colors to tuples, it is also possible to assign them predicates from the underlying signature. As opposed to colors, such predicate assignments are immutable and several distinct predicates may be assigned to the same tuple of elements.

In some more detail, the assembly of a colored instance is carried out through the repeated execution of construction operations each of which is one of the following:
\begin{itemize}
	\item A nullary operation, denoted $\newconst{c}$ for some constant name $\const{c}$, creating a one-element instance containing the term $\const{c}$, with the original color assignments.
	\item A nullary operation, denoted $\newanon$,  creating a one-element instance with an ``unnamed'' (i.e., non-constant) element, with the original color assignments.
	\item A binary operation, denoted by the infix operator $\oplus_{\mathit{Ins}}$, where $\mathit{Ins}$ is a finite set of \emph{instructions}.
	The operation takes two colored instances as input, creates their \emph{disjoint union}, and executes color and predicate modifications according to the instructions contained in $\mathit{Ins}$. The different types of instructions and their effects are as follows:
	\begin{itemize}
		\item color-adding: the instruction $\cadd{\kappa}{\mu}{I}{\lambda}$ assigns the color $\lambda$ to any tuple obtained by $I$-shuffling a $\kappa$-colored tuple from the first input instance with a $\mu$-colored tuple from the second,   
		\item re-coloring: the instruction $\recol{\kappa}{\lambda}^0$ changes the color of all $\kappa$-colored tuples from the first input instance to $\lambda$ (in the same way, the instruction $\recol{\kappa}{\lambda}^1$ re-colors tuples originating from the second input instance),		
		\item predicate-adding: the instruction $\radd{\kappa}{\mu}{I}{\rpred}$ assigns the predicate $\rpred$ to any tuple obtained by $I$-shuffling a $\kappa$-colored tuple of the first instance with a $\mu$-colored tuple of the second,
		\item predicate-materializing: the instruction $\unradd{\kappa}{\rpred}^0$ assigns the predicate $\rpred$ to every $\kappa$-colored tuple from the first input instance (likewise, the instruction $\unradd{\kappa}{\rpred}^1$ materializes predicates out of colored tuples originating from the second input instance).         
	\end{itemize}
	\end{itemize}

%\pagebreak

After the described assembly process over colored instances has finished, we obtain the resulting instance we wanted to build by ``forgetting'' all colors and only keeping the predicates. 
Then, an instance's \ercwidth{} is determined by the minimal set of colors required to create it in the way described above.
In the finite case, an instance's ``assembly plan'' can be conveniently described by an algebraic expression built from the operators introduced above. Such an expression corresponds to a syntax tree. 

%\pagebreak

\begin{example}\label{ex:pw1}
Consider the  finite instance 
$$\instance_\mathrm{ex}=\{%\pred{E}(\const{a},\const{a}),
\pred{E}(\const{a},\const{b}),\pred{E}(\const{a},\const{c}),\pred{E}(\const{a},x),\pred{E}(\const{b},\const{a}),\pred{E}(\const{b},\const{b}),\pred{E}(\const{b},\const{c}),\pred{E}(\const{b},x) \},$$%
over $\Sigma = \{\pred{E}\}$ with just one binary predicate $\pred{E}$, whose active domain consists of the three constants $\const{a}$, $\const{b}$, and $\const{c}$, as well as a null $x$. The instance can be graphically depicted as follows (here and in the following, we omit displaying the universal unary predicate~$\top$ that holds for every term).\\[-8ex] 

\begin{center}
\begin{tikzpicture}[level/.style={sibling distance=60mm/#1, level distance=15mm}]
%%%%%%%%%%%%%%%%%%%%%%%%%%%%%%
%%%%%%%%%%%%% Graph / Clique
%%%%%%%%%%%%%%%%%%%%%%%%%%%%%%

\node[] (b) [yshift=-10cm,xshift=-0.5cm,label={[label distance=0pt]45:$ $}] {$\const{b}$};
\node[] (a) [left=of b, yshift=0mm,xshift=-15mm,label={[label distance=0pt]-45:$ $}] {$\const{a}$};
\node[] (x) [below=of b, yshift=-4mm]  
{$x$};
\node[] (c) [below=of a, yshift=-4mm] 
{$\const{c}$};

\draw [->](b) to [out=90,in=0,looseness=8] node[midway,label={[label distance=-2pt]0:$\pred{E}$}] {} (b);
\draw [->](b) to [bend left=20] node[midway,label={[label distance=-5pt]90:$\pred{E}$}] {} (a);
%\draw [->](a) to [out=90,in=180,looseness=8] node[midway,label={[label distance=-2pt]180:$\pred{E}$}] {} (a);
\draw [->](a) to [bend left=20] node[midway,label={[label distance=-5pt]90:$\pred{E}$}] {} (b);

\draw[->,color=black] (b) -- (x) node [midway,label={[label distance=-4pt]0:$\pred{E}$}] {};
\draw[->,color=black] (a) -- (x) node [midway,yshift=-7pt,label={[label distance=13pt]0:$\pred{E}$}] {};
\draw[->,color=black] (b) -- (c) node [midway,yshift=-7pt,label={[label distance=13pt]180:$\pred{E}$}] {};
\draw[->,color=black] (a) -- (c) node [midway,label={[label distance=-4pt]180:$\pred{E}$}] {};
\end{tikzpicture}
\end{center}

Assume we want to construct $\instance_\mathrm{ex}$ (up to isomorphism) via an assembly process of the kind described above. There are several ways of doing so, but in any case, one has to start by creating one-term-instances by means of appropriate $\pred{New}$ commands. So we obtain the following four colored instances (we use dotted lines for color assignments):\\[-1ex]

\noindent\begin{tabular}{@{\hspace{4.5ex}}c@{\hspace{3.5ex}}|@{\hspace{3.5ex}}c@{\hspace{3.5ex}}|@{\hspace{3.5ex}}c@{\hspace{3.5ex}}|@{\hspace{4ex}}c}
$\newconst{a}$ & $\newconst{b}$ & $\newconst{c}$ & $\newanon$ \\[1ex]
\begin{tabular}{@{}l@{}}
\begin{tikzpicture}	
	\node[] (x2) [label={[label distance=0pt,color=darkblue]0:$\!\!\narrowcdots\,o_1$}] {$\const{a}$};
	\draw [->,color=darkerblue,thick,dotted](x2) to [out=180,in=270,looseness=8] node[midway,label={[label distance=0pt]180:$o_2$}] {} (x2);
\end{tikzpicture}\\[-3ex]
\end{tabular}
&
\begin{tabular}{@{}l@{}}
\begin{tikzpicture}
	\node[] (x2) [label={[label distance=0pt,color=darkblue]0:$\!\!\narrowcdots\,o_1$}] {\phantom{$x$}\hspace{-6pt}\smash{$\const{b}$}};
	\draw [->,color=darkerblue,thick,dotted](x2) to [out=180,in=270,looseness=8] node[midway,label={[label distance=0pt]180:$o_2$}] {} (x2);
\end{tikzpicture}\\[-3ex]
\end{tabular}
&
\begin{tabular}{@{}l@{}}
\begin{tikzpicture}
	\node[] (x2) [label={[label distance=0pt,color=darkblue]0:$\!\!\narrowcdots\,o_1$}] {$\const{c}$};
	\draw [->,color=darkerblue,thick,dotted](x2) to [out=180,in=270,looseness=8] node[midway,label={[label distance=0pt]180:$o_2$}] {} (x2);
\end{tikzpicture}\\[-3ex]
\end{tabular}
&
\begin{tabular}{@{}l@{}}
\begin{tikzpicture}
	\node[] (x2) [label={[label distance=0pt,color=darkblue]0:$\!\!\narrowcdots\,o_1$}] {$x$};
	\draw [->,color=darkerblue,thick,dotted](x2) to [out=180,in=270,looseness=8] node[midway,label={[label distance=0pt]180:$o_2$}] {} (x2);
\end{tikzpicture}\\[-3ex]
\end{tabular}
\end{tabular}\\[-0ex]

Each such initial instance is endowed with the original colors $o_1$ (unary) and $o_2$ (binary). %We do not use colors $o_3$, $o_4$, and so forth, since we only possibly require original colors up to the highest arity occurring in $\Sigma$.
The name of the null introduced by $\newanon$ is irrelevant (recall that we only aim at constructing the above instance up to isomorphism), we just chose $x$ for convenience. It is, however, required that every ``call'' of $\newanon$ introduces a null that is ``fresh'', that is, not otherwise used in other independently constructed assembly components. 

%\pagebreak

Proceeding with the assembly process, we decide to combine the colored instance produced by $\newconst{a}$ with that produced by $\newconst{b}$ by taking their disjoint union and concurrently executing the following three instructions:
\begin{itemize}
\item $\cadd{\origcol_1}{\origcol_1}{\{1\}}{\lambda}$ -- assigning the binary color $\lambda$ to all pairs of terms whose first component is a $o_1$-colored term of the $\newconst{a}$ instance and whose second component  is a $o_1$-colored term of the $\newconst{b}$ instance; the only pair meeting these conditions in our case is $(\const{a},\const{b})$.
\item $\radd{\origcol_1}{\origcol_1}{\{2\}}{\epred}$ -- assigning the binary predicate $\pred{E}$ to all pairs of terms whose first component is a $o_1$-colored term of the $\newconst{b}$ instance and whose second component  is a $o_1$-colored term of the $\newconst{a}$-instance; this way, the atom $\pred{E}(\const{b},\const{a})$ is introduced.
\item $\unradd{\origcol_{2}}{\epred}^{1}$ -- assigning the binary predicate $\pred{E}$ to all $o_2$-colored pairs of terms of the $\newconst{b}$-instance; this way, the atom $\pred{E}(\const{b},\const{b})$ is introduced.
\end{itemize}
%\pagebreak  
%
The result of this composite operation is the colored instance depicted right below which is described by the expression  $\newconst{a} \oplus_{\{ 	\cadd{\origcol_1}{\origcol_1}{\{1\}}{\lambda},
	\unradd{\origcol_{2}}{\epred}^{1}, 
	\radd{\origcol_1}{\origcol_1}{\{2\}}{\epred} \}} \newconst{b}$.\\[-0ex] 

\begin{center}
\begin{tikzpicture}[level/.style={sibling distance=60mm/#1, level distance=15mm}]
		%%%%%%%%%%%%%%%%%%%%%%%%%%%%%%
		%%%%%%%%%%%%% Graph / Clique
		%%%%%%%%%%%%%%%%%%%%%%%%%%%%%%
		
\node[] (b) [label={[label distance=0pt,color=darkblue]0:$\!\!\narrowcdots\,o_1$}] {\phantom{$x$}\hspace{-6pt}\smash{$\const{b}$}};

\node[] (a) [left=of b, yshift=0mm,xshift=-15mm,label={[label distance=0pt,color=darkblue]180:$\,o_1\narrowcdots\!\!$}] {$\const{a}$};

\draw [->,color=darkerblue,thick,dotted](b) to [out=135,in=45,looseness=6] node[pos=0.2,label={[label distance=-5pt]180:$o_2$}] {} (b);		

\draw [->,color=darkerblue,thick,dotted](a) to [out=135,in=45,looseness=6] node[pos=0.2,label={[label distance=-5pt]180:$o_2$}] {} (a);		

\draw [->](b) to [out=-45,in=-135,looseness=8] node[midway,label={[label distance=+8pt]3:$\pred{E}$}] {} (b);

\draw [->](b) to [bend left=20] node[midway,label={[label distance=-5pt]90:$\pred{E}$}] {} (a);

\draw [->,color=darkerred,thick,dotted](a) to [bend left=20] node[midway,label={[label distance=-5pt]90:$\lambda$}] {} (b);
\end{tikzpicture}\\[-10ex]
\end{center}

\vspace{-1.5ex}

%\pagebreak

Next, we choose to combine the colored instances described by $\newconst{c}$ and $\newanon$ by, again, taking their disjoint union and then concurrently executing the following two instructions:
\begin{itemize}
\item $\recol{\origcol_{1}}{\kappa}^0$ -- changing the color of any $o_1$-colored term from the $\newconst{c}$-instance to $\kappa$; the unary color of $\const{c}$ is now $\kappa$. 
\item $\recol{\origcol_{1}}{\kappa}^1$ -- likewise changing the color of any $o_1$-colored term from the $\newanon$-instance to $\kappa$; the unary color of $x$ is now $\kappa$ as well.	
\end{itemize}
The result of this composite operation is the colored instance depicted right below which is described by the expression
$\newconst{c} \oplus_{\{\recol{\origcol_{1}}{\kappa}^0, \recol{\origcol_{1}}{\kappa}^1\}} \newanon$.

\begin{center}
	\begin{tikzpicture}[level/.style={sibling distance=60mm/#1, level distance=15mm}]
		%%%%%%%%%%%%%%%%%%%%%%%%%%%%%%
		%%%%%%%%%%%%% Graph / Clique
		%%%%%%%%%%%%%%%%%%%%%%%%%%%%%%
		
		\node[] (b) [label={[label distance=0pt,color=darkergreen]0:$\!\!\narrowcdots\,\kappa$}] {\phantom{$x$}\hspace{-6pt}\smash{$x$}};
		
		\node[] (a) [left=of b, yshift=0mm,xshift=-15mm,label={[label distance=0pt,color=darkergreen]180:$\kappa\,\narrowcdots\!\!$}] {$\const{c}$};
		
		\draw [->,color=darkerblue,thick,dotted](b) to [out=135,in=45,looseness=6] node[pos=0.2,label={[label distance=-5pt]180:$o_2$}] {} (b);		
		
		\draw [->,color=darkerblue,thick,dotted](a) to [out=135,in=45,looseness=6] node[pos=0.2,label={[label distance=-5pt]180:$o_2$}] {} (a);		
		
%		\draw [->](b) to [out=-45,in=-135,looseness=8] node[midway,label={[label distance=+8pt]3:$\pred{E}$}] {} (b);
		
%		\draw [->](b) to [bend left=20] node[midway,label={[label distance=-5pt]90:$\pred{E}$}] {} (a);
		
%		\draw [->,color=darkerred,thick,dotted](a) to [bend left=20] node[midway,label={[label distance=-5pt]90:$\lambda$}] {} (b);
	\end{tikzpicture}
\end{center}

Finally, we combine the two just created colored structures, represented by the expressions 
$\newconst{a} \oplus_{\{ 	\cadd{\origcol_1}{\origcol_1}{\{1\}}{\lambda},
	\unradd{\origcol_{2}}{\epred}^{1}, 
	\radd{\origcol_1}{\origcol_1}{\{2\}}{\epred} \}} \newconst{b}$ 
and
$\newconst{c} \oplus_{\{\recol{\origcol_{1}}{\kappa}^0, \recol{\origcol_{1}}{\kappa}^1\}} \newanon$, by taking their disjoint union and concurrently executing the following two instructions:
\begin{itemize}
\item $\unradd{\lambda}{\epred}^{0}$ -- assigning the binary predicate $\pred{E}$ to all $\lambda$-colored pairs of terms from the former instance; this way, the atom $\pred{E}(\const{a},\const{b})$ is introduced.
\item $\radd{\origcol_1}{\kappa}{\{1\}}{\epred}$ -- assigning the binary predicate $\pred{E}$ to all pairs of terms whose first component is an $o_1$-colored term of the former instance and whose second component  is a $\kappa$-colored term of the latter instance; this way, the atoms $\pred{E}(\const{a},\const{c})$, $\pred{E}(\const{a},x)$, $\pred{E}(\const{b},\const{c})$, and $\pred{E}(\const{b},x)$ are introduced in one go.
\end{itemize}
%
%\pagebreak
The final result of this operation is a colored structure, described by the expression
$$
(\newconst{a} 
\oplus_{\{ 
			\cadd{\origcol_1}{\origcol_1}{\{1\}}{\lambda},
			\unradd{\origcol_{2}}{\epred}^{1}, 
			\radd{\origcol_1}{\origcol_1}{\{2\}}{\epred}
	    \}} 
\newconst{b}) 
\oplus_{\{
			\unradd{\lambda}{\epred}^{0},
			\radd{\origcol_1}{\kappa}{\{1\}}{\epred}
		\}} 
(\newconst{c} 
\oplus_{\{
		\recol{\origcol_{1}}{\kappa}^0, 
		\recol{\origcol_{1}}{\kappa}^1
		\}} 
\newanon)
$$%
which -- when disregarding the coloring -- is isomorphic to the initially given instance~$\mathcal{I}_\mathrm{ex}$. 
The fact that $4$ colors needed for the construction implies that $\mathcal{I}_\mathrm{ex}$'s partitionwidth is at most $4$.
Arguably, the structure of the ``assembly expression'' above can be made more transparent by displaying its syntax tree below. Also, thinking of assembly expressions as binary trees will be crucial for the formal developments in the following sections.   

\bigskip 
\hspace{5ex}\scalebox{0.9}{\begin{tikzpicture}%[level/.style={sibling distance=60mm/#1, level distance=15mm}]
	[
	    level 1/.style={sibling distance=60mm, level distance=16mm},
	    level 2/.style={sibling distance=30mm, level distance=16mm},
	    level 3/.style={sibling distance=30mm, level distance=10mm},
	    level 4/.style={sibling distance=20mm, level distance=10mm},
	    level 5/.style={sibling distance=15mm, level distance=10mm}
	]

	\node [] (a){$\phantom{_|}\oplus\phantom{_|}$}
	child {node [] (n0) {$\phantom{_|}\oplus\phantom{_|}$}
		child {node [label={[label distance=-4pt]0:$ $}] (n00) {$\ \ \newconst{a}$}
			edge from parent [-] node [] {}
		} %end 00
		child {node [label={[label distance=-4pt]0:$ $}] (n01) {$\newconst{b}$}
			edge from parent [-] node [] {}
		} %end 01
		edge from parent [-] node [] {}
	} %end 0
	child {node [label={[label distance=-4pt]0:$ $}] (n0) {$\phantom{_|}\oplus\phantom{_|}$}
		child {node [label={[label distance=-4pt]0:$ $}] (n00) {$\ \ \newconst{c}$}
			edge from parent [-] node [] {}
		} %end 00
		child {node [label={[label distance=-4pt]0:$ $}] (n01) {$\newanon$}
			edge from parent [-] node [] {}
		} %end 11
		edge from parent [-] node [] {}
	}; %end 1

\node[] (z3) [below=of a, xshift=-0.75cm,yshift=-0.7cm,label={
	{ $_{\{\cadd{\origcol_1}{\origcol_1}{\{1\}}{\lambda},
		\unradd{\origcol_{2}}{\epred}^{1}, 
		\radd{\origcol_1}{\origcol_1}{\{2\}}{\epred}\}}$}}] { };

\node[] (z4) [below=of a, xshift=4.75cm,yshift=-0.7cm,label={
	{ $_{\{
		\recol{\origcol_{1}}{\kappa}^0, 
		\recol{\origcol_{1}}{\kappa}^1
		\}}$}}] { };
	
\node[] (z5) [below=of a, xshift=1.45cm,yshift=0.9cm,label={
	{ $_{\{
			\unradd{\lambda}{\epred}^{0},
			\radd{\origcol_1}{\kappa}{\{1\}}{\epred}
			\}}$}}] { };
\end{tikzpicture}}
\\[-2ex]
\end{example}

For finite instances, as in the example just discussed, the underlying ideas are very much in line with the notion of inductive definitions and should be easy to comprehend.
Opposed to this, for infinite instances, the much more elusive idea of an ``infinite assembly plan'' has to be implemented by resorting to assembly expressions that correspond to infinite, ``unfounded'' syntax trees. Such non-standard expressions describe construction processes that have an end but no proper beginning. According to our experience, this concept is cognitively hard to grasp and it takes a while to develop good and reliable intuitions about it. We will precede the formal treatment with an example to help developing the appropriate mindset, but certain aspects of the example will have to remain sketchy until the rigid formal development in the subsequent section.

\begin{example}\label{ex:pw2}
We consider the signature $\Sigma=\{\pred{R}\}$, which just contains the ternary predicate $\pred{R}$, and the infinite, constant-free instance 
$$
\inst_\mathrm{tern} = \{ \rpred(-1,n,n{+}1) \mid n\in \mathbb{N} \},
$$%
where, for the sake of simplicity, we use the integers $-1$, $0$, $1$, $2$, $\ldots$ as nulls from $\mathbf{N}$.
That is, $\inst_\mathrm{tern}$ contains the atoms $\rpred(-1,0,1)$, $\rpred(-1,1,2)$, $\rpred(-1,2,3)$, and so forth, as shown in the below drawing.\\[-6ex]
%\pagebreak
\begin{center}
\qquad\begin{tikzpicture}[
	Dotted/.style={% https://tex.stackexchange.com/a/52856/194703
		dash pattern=on 0.1\pgflinewidth off #1\pgflinewidth,line cap=round,
		shorten >=#1\pgflinewidth/2,shorten <=#1\pgflinewidth/2},
	Dotted/.default=3]

%%%%%%%%%%%%%	
%%%Eighth LINE
%%%%%%%%%%%%%	
\node[] [] (h1) [label={[label distance=-2pt]-90:$ $}, yshift=-0pt] {$\hspace{-1ex}0\hspace{1ex}$};
\node[] [] (h2) [right=of h1,label={[label distance=-2pt]-90:$ $}]  {$\hspace{-1ex}1\hspace{1ex}$};
\node[] [] (h3) [right=of h2,label={[label distance=-2pt]-90:$ $}]  {$\hspace{-1ex}2\hspace{1ex}$};
\node[] [] (h4) [right=of h3,label={[label distance=-2pt]-90:$ $}]  {$\hspace{-1ex}3\hspace{1ex}$};
\node[] [] (h5) [right=of h4,label={[label distance=-2pt]-90:$ $}]  {$\hspace{-1ex}4\hspace{1ex}$};
\node[] [] (h6) [right=of h5,xshift=-1cm] {$\cdots$};
%\node[] (h6) [left=of h5,xshift=-4cm] {$\radd{\origcol_{1}}{\lambda}{\{1\}}{\rpred} \ \ $};
\node[] [] (h8) [left=of h1,yshift=-0.60pt,label={[label distance=-2pt]-90:$ $}] {$-1$};

%\draw[->] (h1) -- (h2) node [midway,label={[label distance=-2pt]-90:$ $}] {};
%\draw[->] (h2) -- (h3) node [midway,label={[label distance=-2pt]-90:$ $}] {};
%\draw[-,transform canvas={yshift=6.pt,xshift=-12pt}] (h2) -- ($(h1)!3.06cm!(h3)$) node [] {};
%\draw[->] (h3) -- (h4) node [midway,label={[label distance=-2pt]-90:$ $}] {};

\node[] (h9) [right=of h8,xshift=-.05cm] { };
\node[] (h10) [left=of h1,xshift=-.05cm] { };

\begin{scope}[local bounding box=only paths]
\draw[fill=gray!45, fill opacity=.3] (h8) to[out=45,in=135] (h5);
\draw[fill=gray!45, fill opacity=.3] (h8) to[out=45,in=135] (h4);
\draw[fill=gray!45, fill opacity=.3] (h8) to[out=45,in=135] (h3);
\draw[fill=gray!45, fill opacity=.3] (h8) to[out=45,in=135] (h2);
\draw[fill=shadecolor] (h8) to[out=45,in=135] (h1); 
%\filldraw [shadecolor, line width=.45cm] (h9) -- (h10);
\end{scope}

\draw[-,transform canvas={yshift=6.8pt,xshift=-13pt}] (h1) -- ($(h1)!6.10cm!(h3)$) node [] {};

\node[] [] (R) [above=of h1,yshift=-1.0cm,xshift=.1cm]{$\pred{R}$};
\node[] [] (R) [above=of h2,yshift=-.8cm,xshift=-.2cm]{$\pred{R}$};
\node[] [] (R) [above=of h3,yshift=-.6cm,xshift=-.5cm]{$\pred{R}$};
\node[] [] (R) [above=of h4,yshift=-.4cm,xshift=-.8cm]{$\pred{R}$};
\end{tikzpicture}
\end{center}
\vspace{-2ex}

Due to the already mentioned peculiarity that assembling  $\inst_\mathrm{tern}$ does not have a proper starting point, we will argue ``backward'' by starting from $\inst_\mathrm{tern}$ and discussing how it could have been created through a combination of colored instances, then asking the same question for it's constituents, etc. %-- in a way this approach can be understood as a sort of ``disassembly ad infinitum''.

To this end, we first note that $\inst_\mathrm{tern}$ could have arisen from the two colored instances
\medskip

\newcommand{\overncdot}[1]{
	\begin{array}{c}
    \rotatebox{90}{$...$}
	\\[-5pt]
	#1	 
    \end{array}
}

\vspace{-4pt}
\begin{center}
\begin{tabular}{c@{\hspace{5ex}}|@{\hspace{5ex}}c@{\hspace{4ex}}}
\begin{tabular}{@{}l@{}}
	\begin{tikzpicture}
		\node[] (x2) [label={[label distance=-4pt,color=darkblue]-90:$\overncdot{o_1}$}] {$\smash{-}1$};
	\end{tikzpicture}\\[-2ex]
\end{tabular}
&
\begin{tabular}{@{}l@{}}
	\begin{tikzpicture}
	\node[] (f1) [label={[label distance=-4pt,color=darkercyan]-90:$\overncdot{\mu}$}] %,xshift=40pt] 
	{$0$};
	\node[] (f2) [right=of f1,label={[label distance=-4pt,color=darkergreen]-90:$\overncdot{\kappa}$}] {$1$};
	\node[] (f3) [right=of f2,label={[label distance=-4pt,color=darkergreen]-90:$\overncdot{\kappa}$}] {$2$};
	\node[] (f4) [right=of f3,label={[label distance=-4pt,color=darkergreen]-90:$\overncdot{\kappa}$}] {$3$};
	\node[] (f5) [right=of f4,label={[label distance=-4pt,color=darkergreen]-90:$\overncdot{\kappa}$}] {$4$};
	\node[] (f6) [right=of f5,xshift=-0.8cm] {$\cdots$};
	
	\draw[->,color=darkerred,thick,dotted] (f1) -- (f2) node [midway,label={[label distance=-4pt]-90:$\lambda$}] {};
	\draw[->,color=darkerred,thick,dotted] (f2) -- (f3) node [midway,label={[label distance=-4pt]-90:$\lambda$}] {};
	\draw[->,color=darkerred,thick,dotted] (f3) -- (f4) node [midway,label={[label distance=-4pt]-90:$\lambda$}] {};
	\draw[->,color=darkerred,thick,dotted] (f4) -- (f5) node [midway,label={[label distance=-4pt]-90:$\lambda$}] {};
	\end{tikzpicture}\\[-2ex]
\end{tabular}
\end{tabular}
\end{center}
\vspace{-1pt}
\noindent by taking their disjoint union and executing the instruction $\radd{\origcol_1}{\lambda}{\{1\}}{\rpred}$ -- assigning the ternary predicate $\pred{R}$ to all triples  of terms whose first component is an $o_1$-colored term of the left instance and whose second and third component taken together as a pair carry the color $\lambda$ in the right instance; this way, all the infinitely many atoms of $\inst_\mathrm{tern}$ are introduced at once.

While the left of the two input instances can be directly created through a $\newanon$ call (in this case, the only original color needed is the unary color $o_1$), the right input instance is itself infinite. For now we will denote the yet-to-be determined assembly expression for the right instance with $E$, whence the assembly expression of $\inst_\mathrm{tern}$ becomes $\newanon \smash{\oplus_{\{\radd{\,\origcol_1}{\smash{\lambda}}{\smash{\{\hspace{-0.5pt}1\hspace{-0.5pt}\}}}{\rpred}\}}}E$.
In our quest to reverse-engineer the $E$-instance further, we find that it can be garnered from the two colored instances 

\medskip

\vspace{-4pt}
\begin{center}
	\begin{tabular}{@{\hspace{5.5ex}}c@{\hspace{5ex}}|@{\hspace{5ex}}c@{\hspace{0ex}}}
		\begin{tabular}{@{}l@{}}
			\begin{tikzpicture}
				\node[] (x2) [label={[label distance=-4pt,color=darkblue]-90:$\overncdot{o_1}$}] {$0$};
			\end{tikzpicture}\\[-2ex]
		\end{tabular}
		&
		\begin{tabular}{@{}l@{}}
			\begin{tikzpicture}
				\node[] (f1) [label={[label distance=-4pt,color=darkercyan]-90:$\overncdot{\mu}$}] %,xshift=40pt] 
				{$1$};
				\node[] (f2) [right=of f1,label={[label distance=-4pt,color=darkergreen]-90:$\overncdot{\kappa}$}] {$2$};
				\node[] (f3) [right=of f2,label={[label distance=-4pt,color=darkergreen]-90:$\overncdot{\kappa}$}] {$3$};
				\node[] (f4) [right=of f3,label={[label distance=-4pt,color=darkergreen]-90:$\overncdot{\kappa}$}] {$4$};
				\node[] (f5) [right=of f4,xshift=-0.8cm] {$\cdots$};
				
				\draw[->,color=darkerred,thick,dotted] (f1) -- (f2) node [midway,label={[label distance=-4pt]-90:$\lambda$}] {};
				\draw[->,color=darkerred,thick,dotted] (f2) -- (f3) node [midway,label={[label distance=-4pt]-90:$\lambda$}] {};
				\draw[->,color=darkerred,thick,dotted] (f3) -- (f4) node [midway,label={[label distance=-4pt]-90:$\lambda$}] {};
			\end{tikzpicture}\\[-2ex]
		\end{tabular}
	\end{tabular}
\end{center}
\vspace{2pt}
\noindent by, once more, taking their disjoint union and then concurrently executing the following three instructions:
\begin{itemize}
    \item $\cadd{\origcol_1}{\mu}{\{1\}}{\lambda}$ -- assigning the binary color $\lambda$ to all pairs of terms whose first component is a $o_1$-colored term of the left instance and whose second component  is a $\mu$-colored term of the right instance; the only pair meeting these conditions in our case is $(0,1)$.
	\item $\recol{\mu}{\kappa}^1$ -- changing the color of any $\mu$-colored term from the right instance to $\kappa$.%; the unary color of $1$ is now $\kappa$. 
	\item $\recol{\origcol_{1}}{\mu}^0$ --  changing the color of any $o_1$-colored term from the left instance to $\mu$.%; the unary color of $0$ is now $\mu$.	
\end{itemize}

\noindent We find that -- once more -- the left instance can be created via $\newanon$, while the right instance is isomorphic to the $E$-instance we just disassembled. As the concrete names of the nulls used are irrelevant, we find that the assembly expression $E$ creates the right instance as well. It is easy to see that the disassembly step can be created over and over again ad infinitum, where every null from $\{0,1,\ldots\}$ will at some stage be obtained via a $\newanon$ call. We also note that, given the observation above, the expression $E$ must satisfy the ``recursive equation'' $$E = \newanon \,\smash{\oplus_{\{\cadd{\origcol_1}{\mu}{\smash{\{\hspace{-1pt}1\hspace{-1pt}\}}}{\lambda},\recol{\mu}{\kappa}^1,\recol{\origcol_{1}}{\mu}^0\}}}\,E.$$
While it is clear that this requirement cannot be realized by any finite string, one may choose to extend the setting to include ``unfounded'' expressions of infinite depth. Such expressions are difficult to grasp as linear sequences of symbols (at best, one might envision $E$ as the result of expanding $E$ with the right hand side of the above recursive equation infinitely often). But they may be easier to understand by means of the corresponding syntax trees. Below we depict the initial part of an (infinite) syntax tree representing $\inst_\mathrm{tern}$, stating that the subtree rooted in the root's right child coincides with the subtree rooted in the root's right-right grandchild, giving rise to a periodic infinite downward branch. 

\hspace{15ex}\begin{tikzpicture}%[level/.style={sibling distance=60mm/#1, level distance=15mm}]
	[
	level 1/.style={sibling distance=40mm, level distance=11mm},
	level 2/.style={sibling distance=40mm, level distance=11mm},
	level 3/.style={sibling distance=40mm, level distance=11mm},
	level 4/.style={sibling distance=40mm, level distance=11mm},
	level 5/.style={sibling distance=40mm, level distance=11mm}
	]

	\node [] (a){$\oplus\phantom{_|}\phantom{_|}$}
	child {node [] (n0) {$\newanon$}
		edge from parent [-] node [] {}
	} %end 0
	child {node [label={[label distance=-4pt]0:$ $}] (n0) {$\oplus\phantom{_|}\phantom{_|}$}
		child {node [label={[label distance=-4pt]0:$ $}] (n00) {$\newanon$}
			edge from parent [-] node [] {}
		} %end 00
	child {node [label={[label distance=-4pt]0:$ $}] (m0) {$\oplus\phantom{_|}\phantom{_|}$}
	child {node [label={[label distance=-4pt]0:$ $}] (m00) {$ $}
		edge from parent [-,path fading=South] node [] {}
	} %end 00
	child {node [label={[label distance=-4pt]0:$ $}] (n01) {$ $}
		edge from parent [-,path fading=South] node [] {}
	} %end 11
		edge from parent [-] node [] {}
	}}; %end 1

	\node[] (z5) [below=of a, xshift=0.8cm,yshift=0.93cm,label={
		{ $_{\{ \radd{\origcol_1}{\lambda}{\{1\}}{\rpred} \}}$}}] { };

	\node[] (z5) [below=of a, xshift=4.2cm,yshift=-0.16cm,label={
	{ $_{\{   	
			\cadd{\origcol_1}{\mu}{\smash{\{\hspace{-1pt}1\hspace{-1pt}\}}}{\lambda},
			\recol{\mu}{\kappa}^1,
			\recol{\origcol_{1}}{\mu}^0
			 \}}$}}] { };
	
	\node[] (z5) [below=of a, xshift=6.2cm,yshift=-1.25cm,label={
	{ $_{\{   	
			\cadd{\origcol_1}{\mu}{\smash{\{\hspace{-1pt}1\hspace{-1pt}\}}}{\lambda},
			\recol{\mu}{\kappa}^1,
			\recol{\origcol_{1}}{\mu}^0
			\}}$}}] { };

\end{tikzpicture}

We find that, despite being infinite, the assembly expression only makes use of the four colors $o_1$, $\mu$, $\kappa$, and $\lambda$. This allows us to conclude that the partitionwidth of   
$\inst_\mathrm{tern}$ is at most $4$.
\end{example}

The following section will provide full formal details on partitionwidth. In particular, we will formalize the hitherto rather sketchy idea of a potentially infinite ``assembly-plan-encoding syntax trees'' by defining such a tree as a countably infinite instance of a particular shape. 
%Before proceeding to our formal definition however, we remark that an example of such a tree (viz. a well-decorated tree), along with the instance the tree defines, is provided in \fig~\ref{fig:well-decotrated-tree-example}, which aims to make \dfn~\ref{def:infinite-binary-tree}--\ref{def:entities-tree-to-instance} more concrete for the reader.

\subsection{Formal Definition of \Partitionwidth{} of Countable Instances}\label{sec:cwdef}

\newcommand{\szero}{\pred{Succ}_0}
\newcommand{\sone}{\pred{Succ}_1}
\newcommand{\ibtree}{\mathcal{T}_\mathrm{bin}}
\newcommand{\colconsig}{(\cols, \mathrm{Cnst}, \Sigma)}
\newcommand{\decorum}{\mathrm{Dec}\colconsig}
\newcommand{\predupdate}[1]{\pred{#1}}
\newcommand{\instructionset}{\mathrm{Instr}(\cols, \Sigma)}
\newcommand{\instruction}{\mathrm{Instr}}
\newcommand{\instr}{\instruction}

\begin{definition}[Update Instructions]\label{def:instructions}
	Let $\Sigma$ be a finite signature and let $\cols$ be a finite set of \emph{colors}, where every $\lambda\in\cols$ is assigned an arity $\arity{\lambda}\in\mathbb{N}_{+}$.
	%Let $\origcol \in \cols$ with $\arity{\origcol} = 1$ a distinguished unary color.
	The set $\instructionset$ of \emph{update instructions over $(\cols, \Sigma)$} consists of the following expressions:
	\begin{itemize}
		%\item \!$\const{c}$ for any $\const{c}\in\mathrm{Const}$,
		\item \!$\cadd{\kappa}{\mu}{I}{\lambda}$ for $\kappa, \mu, \lambda\in \cols$ with  $\arity{\kappa}+\arity{\mu} = \arity{\lambda}$, $I{\,\subsetneq\,} \{ 1, ..., \arity{\lambda} \}$, and $\arity{\kappa} = |I|$,
		\item \!$\recol{\kappa}{\lambda}^i$ for $i\in \{0,1\}$ and $\lambda, \kappa\in\cols$ with $\arity{\lambda}=\arity{\kappa}$,		
		\item \!$\radd{\kappa}{\mu}{I}{\rpred}$ for $\kappa, \mu\in \cols$ and $\rpred\in \Sigma$ with $\arity{\kappa}+\arity{\mu} = \arity{\rpred}$, $I {\,\subsetneq\,} \{ 1, ..., \arity{\rpred} \}$, and $\arity{\kappa} = |I|$, 
		\item \!$\unradd{\kappa}{\rpred}^i$ for $i\in \{0,1\}$, $\kappa\in\cols$ and $\rpred\in \Sigma$ with $\arity{\rpred} = \arity{\kappa}$.

% \textcolor{red}{This is not enough! We need to have at least a 0 and a 1 option. Or, we roll this all into the Cadd.}
%		\item \!\textcolor{thomas}{$\pred{None}$}.
	\end{itemize}
	Two instructions of the form $\cadd{\kappa}{\mu}{I}{\lambda}$ or $\recol{\kappa}{\lambda}^i$ will be called \emph{incompatible}, if they disagree on $\lambda$ while all other parameters coincide.\defend
\end{definition}
Note that since $\cols$ and $\Sigma$ are all assumed to be finite, the set $\instructionset$ is finite as well.
\pagebreak

\begin{example}
Exemplars of the instructions just formally introduced have already been seen in action informally in \Cref{ex:pw1} and \Cref{ex:pw2}, which should facilitate understanding their purpose. To exemplify the intended effect of instructions of the type $\cadd{\kappa}{\mu}{I}{\lambda}$ (and, in an analogue manner, $\radd{\kappa}{\mu}{I}{\rpred}$), consider the following case. When executing the instruction \smash{$\cadd{\kappa}{\mu}{\{2,4\}}{\lambda}$} on two colored instances, where
\begin{itemize}
\item the first colored instance has domain $\{\const{a},\const{b},\const{c}\}$ wherein the pairs $(\const{a},\const{a})$, $(\const{a},\const{b})$, and $(\const{a},\const{c})$ carry the binary color $\kappa$ and no tuple carries the quaternary color $\lambda$, and
\item the second colored instance has domain $\{\const{d},\const{e},\const{f}\}$ wherein the pairs $(\const{d},\const{e})$ and $(\const{e},\const{f})$ carry the binary color $\mu$ and no tuple carries the quaternary color $\lambda$,  
\end{itemize}
the result is an instance with domain $\{\const{a},\const{b},\const{c},\const{d},\const{e},\const{f}\}$ carrying all colors and predicates of the input instances and, additionally, the color $\lambda$ is assigned to the following quadruples:
$$
(\const{d},\const{a},\const{e},\const{a}),\ \ 
(\const{e},\const{a},\const{f},\const{a}),\ \ 
(\const{d},\const{a},\const{e},\const{b}),\ \ 
(\const{e},\const{a},\const{f},\const{b}),\ \ 
(\const{d},\const{a},\const{e},\const{c}),\ \ 
(\const{e},\const{a},\const{f},\const{c}), 
$$%
which are obtained from combining every $\kappa$-pair from the first colored instance with every $\mu$-pair from the second colored instance by putting the first and second entry of the $\kappa$-pair in positions $2$ and $4$, respectively, while putting  the first and second entry of the $\mu$-pair in the remaining positions $1$ and $3$, respectively. Note that the ``internal position-order'' of each contributing pair is preserved.
\end{example}

%Additionally, all updates are finite sets.

As indicated above, the idea of ``instance assembly expressions'' -- meant to serve as blueprints for the instance construction -- plays a central role in our development. As such ``assembly expressions'' may be infinite, it is helpful to encode them by their corresponding syntax trees. Such syntax trees will themselves take the form of instances of a very particular tree-like shape. The next definition specifies how an infinite binary tree, whose nodes might possibly carry some ``markings'' or ``adornments'', will be expressed as an instance.      

\begin{definition}[(Adorned) Infinite Binary Tree]
	The \emph{infinite binary tree} is the instance
%\centerline{
	$$\ibtree = 
    \big\{\szero(s,s0) \ \big| \ s {\,\in\,} \{0,1\}^*\big\} \cup 
	\big\{ \sone(s,s1) \ \big| \ s {\,\in\,} \{0,1\}^*\big\}
	$$%
%	}
	with two binary predicates $\szero$ and $\sone$.
	We use finite sequen\-ces composed of the symbols $0$ and $1$ to identify the infinitely many nulls of $\ibtree$. The \emph{root} of $\ibtree$ is the null identified by the empty sequence, denoted $\varepsilon$.

%%%% PLEASE DON'T TOUCH THIS ANY MORE, IT JUST CREATES MORE WORK
    An \emph{adorned infinite binary tree} (short: adorned tree) is obtained by extending $\ibtree$ with arbitrarily many atoms of the form $\pred{P}(s)$ where $\pred{P}$ is a unary predicate from a finite signature and $s \in \{0,1\}^*$. Let $\mathbb{T}$ denote the set of all adorned trees. 
    \defend
%%%%
% connectedness enforces that adom is {0,1}*
%% THIS IS JUST AN ADDITIONAL SOURCE OF CONFUSION; CONNECTED USUALLY SUGGESTS "CONNECTED COMPONENT". LETS KEEP IT SIMPLE AND STRAIGHT
%
%	We say that a connected set of atoms consisting of %  $\ibtree$ and any a set of unary 
%or nullary  WHY THE HECK SHOULD WE NEED NULLARY PREDICATES??? 
%	atoms over a finite signature is an {\em adorned infinite binary tree}. 
\end{definition}

%\pagebreak
\begin{example}
The below drawing illustrates the concept of adorned infinite binary trees just defined.
It shows (the initial part of) some tree $\mathcal{T} {\,\in\,} \mathbb{T}$ adorned with unary predi\-cates $\pred{P}_1$, $\pred{P}_2$, and $\pred{P}_3$. $\pred{Succ}_{0}$ and $\pred{Succ}_{1}$ are represented as black (left) and gray (right) downward edges, respectively.

\vspace{-2ex}

%\pagebreak

\begin{center}
\begin{tikzpicture}%[level/.style={sibling distance=50mm/#1, level distance=15mm}]
	[
	    level 1/.style={sibling distance=60mm, level distance=8mm},
	    level 2/.style={sibling distance=30mm, level distance=8mm},
	    level 3/.style={sibling distance=15mm, level distance=8mm},
	    level 4/.style={sibling distance=10mm, level distance=8mm},
	    level 5/.style={sibling distance=15mm, level distance=8mm}
	]
	
	\node (a){$\varepsilon$}
	child {node (n0) {$0$}
		child {node (n00) [label={[label distance=-4pt]-2.5:${ \ : \ \pred{P}_3}$}] {$00$}
          	child {node (ntt) {$000$}
				child {node [] (n00000) { }
					edge from parent [-,color=black,line width=1pt,path fading=South] node [] {}
				}
				child {node [opacity=0.2,path fading=South,color=gray!60] (n00001) { }
					edge from parent [-,color=gray!60,line width=1pt,path fading=South] node [] {}
				}
				edge from parent [-,color=black,line width=1pt] node [] {}
			}
			child {node (nts) {$001$}
				child {node [] (n00000) { }
					edge from parent [-,color=black,line width=1pt,path fading=South] node [] {}
				}
				child {node [opacity=0.2,path fading=South,color=gray!60] (n00001) { }
					edge from parent [-,color=gray!60,line width=1pt,path fading=South] node [] {}
				}
				edge from parent [-,color=gray!60,line width=1pt] node [] {}
		}} %end 00}
		child {node [label={[label distance=-4pt]-2.5:${ \ : \ \pred{P}_1,\pred{P}_2}$}] (n01) {${01}$}
          	child {node (ntt) {$010$}
	child {node [] (n00000) { }
		edge from parent [-,color=black,line width=1pt,path fading=South] node [] {}
	}
	child {node [opacity=0.2,path fading=South,color=gray!60] (n00001) { }
		edge from parent [-,color=gray!60,line width=1pt,path fading=South] node [] {}
	}
	edge from parent [-,color=black,line width=1pt] node [] {}
}
child {node (nts) {$011$}
	child {node [] (n00000) { }
		edge from parent [-,color=black,line width=1pt,path fading=South] node [] {}
	}
	child {node [opacity=0.2,path fading=South,color=gray!60] (n00001) { }
		edge from parent [-,color=gray!60,line width=1pt,path fading=South] node [] {}
	}
	edge from parent [-,color=gray!60,line width=1pt] node [] {}
	}
			edge from parent [-,color=gray!60,line width=1pt] node [] {}
} %end 01
		edge from parent [-,line width=1pt] node [] {}
	} %end 0
	child {node (n0) [label={[label distance=-4pt]-2.5:${ \ : \ \pred{P}_2,\pred{P}_3}$}] {$1$}
	child {node (n00) {$10$}
		child {node (ntt) {$100$}
			child {node [] (n00000) { }
				edge from parent [-,color=black,line width=1pt,path fading=South] node [] {}
			}
			child {node [opacity=0.2,path fading=South,color=gray!60] (n00001) { }
				edge from parent [-,color=gray!60,line width=1pt,path fading=South] node [] {}
			}
			edge from parent [-,color=black,line width=1pt] node [] {}
		}
		child {node (nts) {$101$}
			child {node [] (n00000) { }
				edge from parent [-,color=black,line width=1pt,path fading=South] node [] {}
			}
			child {node [opacity=0.2,path fading=South,color=gray!60] (n00001) { }
				edge from parent [-,color=gray!60,line width=1pt,path fading=South] node [] {}
			}
			edge from parent [-,color=gray!60,line width=1pt] node [] {}
	}			edge from parent [-,color=black,line width=1pt] node [] {}
} %end 00}
child {node (n01) {$11$}
	child {node (ntt) {$110$}
		child {node [] (n00000) { }
			edge from parent [-,color=black,line width=1pt,path fading=South] node [] {}
		}
		child {node [opacity=0.2,path fading=South,color=gray!60] (n00001) { }
			edge from parent [-,color=gray!60,line width=1pt,path fading=South] node [] {}
		}
		edge from parent [-,color=black,line width=1pt] node [] {}
	}
	child {node (nts) [label={[label distance=-4pt]-2.5:${ \ : \ \pred{P}_1}$}] {$111$}
		child {node [] (n00000) { }
			edge from parent [-,color=black,line width=1pt,path fading=South] node [] {}
		}
		child {node [opacity=0.2,path fading=South,color=gray!60] (n00001) { }
			edge from parent [-,color=gray!60,line width=1pt,path fading=South] node [] {}
		}
		edge from parent [-,color=gray!60,line width=1pt] node [] {}
	}
	edge from parent [-,color=gray!60,line width=1pt] node [] {}
} %end 01
edge from parent [-,color=gray!60,line width=1pt] node [] {}
};
%	child {node (n0) {$1$}
%		child {node (n00) {$10$}
%			child {node [] (n000) { }
%				edge from parent [-,color=black,line width=1pt,path fading=South] node [] {}
%			} %end 000
%			child {node [opacity=0.2,path fading=South,color=gray!60] (n001) { }
%				edge from parent [-,color=gray!60,line width=1pt,path fading=South] node [] {}
%			} %end 001
%			edge from parent [-,color=black,line width=1pt] node [] {}
%		} %end 00
%		child {node (n01) {$11$}
%			child {node [opacity=0.2] (n110) { }
%				edge from parent [-,color=black,line width=1pt,path fading=South] node [] {}
%			} %end 010
%			child {node [opacity=0.2] (n111) { }
%				edge from parent [-,color=gray!60,line width=1pt,path fading=South] node [] {}
%			} %end 011
%			edge from parent [-,color=gray!60,line width=1pt] node [] {}
%		} %end 11
%		edge from parent [-,color=gray!60,line width=1pt] node [] {}
%	}; %end 1
	\end{tikzpicture}
\end{center}
\vspace{-3ex}
\end{example}

%\pagebreak

After specifying how adorned infinite binary trees are encoded as instances, the next step consists in providing the proper adornments for our purpose (referred to as \emph{decorators}) and specify local and global constraints regulating where which decorators can be put inside the tree-instance. Intuitively, the purpose of these deliberations is to ensure that there is a clear one-to-one correspondence between properly adorned tree-instances (referred to as \emph{well-decorated trees}) and assembly syntax trees as originally envisioned.

\begin{definition}[Decorators and Decorated Trees]\label{def:well-dec}

	Let $\Sigma$ and $\cols$ be as in \Cref{def:instructions} and let $\mathrm{Cnst}\subseteq \mathbf{C}$ be a finite set of constants. $\cols$ may contain special colors from the set $\{\origcol_1, \ldots, \origcol_{\arity{\Sigma}}\}$ with $\arity{\origcol_i} = i$, i.e., $\origcol_i$ is a distinguished $i$-ary color for $1\le i\le \arity{\Sigma}$.
Then, the set of \emph{decorators} %defined as
$\decorum$
is
$
\{ \newanon \} \cup \{ \newconst{c} \mid \const{c} \in \mathrm{Cnst}\} \cup \instructionset
$.

	A \emph{$\colconsig$-decorated infinite binary tree} (or \emph{decorated tree} for short) is an adorned tree where the unary predicates are from the set $\decorum$.  A decorated tree $\mathcal{T}$ is called a \emph{well-decorated tree} if 
	\begin{itemize}
		\item for every null $s\in\{ 0,1 \}^*$, if $\mathcal{T}$ contains a fact of the form $\newconst{c}(s)$ or $\newanon(s)$ (in which case $s$ will be called a \emph{pseudoleaf}), then $\mathcal{T}$ contains no other fact of the form $\pred{Dec}(s)$ with $\pred{Dec}\in\decorum$ (that is, every pseudoleaf carries exactly one ``creation command'');
		\item for any pseudoleaf $s$, any $s' \in \{0,1\}^+$, and any any $\pred{Dec}\in\decorum$, we have $\pred{Dec}(ss')\not\in \mathcal{T}$ (that is, descendants of pseudoleafs are undecorated);  
		\item for every $\const{c}\in\mathrm{Const}$, $\mathcal{T}$ contains at most one fact of the form $\newconst{c}(s)$ (that is, constant-creation commands are globally unique throughout $\mathcal{T}$);
        \item for every null $s\in\{ 0,1 \}^*$, if $\pred{Dec}(s),\pred{Dec'}(s) \in \mathcal{T}$ with $\pred{Dec},\pred{Dec'} \in \instructionset$ then $\pred{Dec}$ and $\pred{Dec'}$ are compatible (that is,  
        no node is decorated by incompatible instructions).
\defend
	\end{itemize}
\end{definition}

\begin{example}\label{ex:fintree}
Below we display (an initial part of) the well-decorated tree associated with \Cref{ex:pw1}.
Note that it is an infinite structure but the nodes below the faded-out edges do not carry any decorators. 

\vspace{-2.5ex}
\hspace{12ex}\scalebox{0.8}{
\begin{tikzpicture}%[level/.style={sibling distance=60mm/#1, level distance=15mm}]
	[
		level 1/.style={sibling distance=60mm, level distance=12mm},
level 2/.style={sibling distance=30mm, level distance=15mm},
level 3/.style={sibling distance=14mm, level distance=10mm},
%		level 4/.style={sibling distance=10mm, level distance=10mm},
%		level 5/.style={sibling distance=15mm, level distance=10mm}
%	    level 1/.style={sibling distance=40mm, level distance=10mm},
%	    level 2/.style={sibling distance=35mm, level distance=10mm},
%	    level 3/.style={sibling distance=30mm, level distance=10mm},
%	    level 4/.style={sibling distance=20mm, level distance=10mm},
%	    level 5/.style={sibling distance=15mm, level distance=10mm}
	]

	\node [label={[label distance=-4pt]0:\raisebox{8pt}{$ \, : \ \radd{\origcol_1}{\lambda}{\{1\}}{\epred}  $}}] (a){$\ \varepsilon^{\phantom{!}}_{\phantom{|}}$}
	child {node [label={[label distance=-4pt]0:\raisebox{-12pt}{ $ \, : \  \unradd{\origcol_{2}}{\epred}^{0}, \unradd{\origcol_{2}}{\epred}^{1}, $}}] (n0) {$\ 0^{\phantom{|}}$}
		child {node [label={[label distance=-4pt]0:$ \ : \ \newconst{a}$}] (n00) {$00$}
			child {node [] (n00000) { }
				edge from parent [-,color=black,line width=1pt,path fading=South] node [] {}
			}
			child {node [opacity=0.2,path fading=South,color=gray!60] (n00001) { }
				edge from parent [-,color=gray!60,line width=1pt,path fading=South] node [] {}
			}
			edge from parent [-,color=black,line width=1pt] node [] {}
		} %end 00
		child {node [label={[label distance=-4pt]0:$ \ : \ \newconst{b}$}] (n01) {$01$}
			child {node [opacity=0.2] (n010) { }
				edge from parent [-,color=black,line width=1pt,path fading=South] node [] {}
			} %end 010
			child {node [opacity=0.2] (n011) { }
				edge from parent [-,color=gray!60,line width=1pt,path fading=South] node [] {}
			} %end 011
			edge from parent [-,color=gray!60,line width=1pt] node [] {}
		} %end 01
		edge from parent [-,line width=1pt] node [] {}
	} %end 0
	child {node [label={[label distance=-4pt]0:\raisebox{-12pt}{$ \ : \ \recol{\origcol_{1}}{\lambda}^0,$}}] (n0) {$\ 1^{\phantom{|}}$}
		child {node [label={[label distance=-4pt]0:$ \ : \ \newconst{c}$}] (n00) {$10$}
			child {node [] (n000) { }
				edge from parent [-,color=black,line width=1pt,path fading=South] node [] {}
			} %end 000
			child {node [opacity=0.2,path fading=South,color=gray!60] (n001) { }
				edge from parent [-,color=gray!60,line width=1pt,path fading=South] node [] {}
			} %end 001
			edge from parent [-,color=black,line width=1pt] node [] {}
		} %end 00
		child {node [label={[label distance=-4pt]0:$ \ : \ \newanon$}] (n01) {$11$}
			child {node [opacity=0.2] (n110) { }
				edge from parent [-,color=black,line width=1pt,path fading=South] node [] {}
			} %end 010
			child {node [opacity=0.2] (n111) { }
				edge from parent [-,color=gray!60,line width=1pt,path fading=South] node [] {}
			} %end 011
			edge from parent [-,color=gray!60,line width=1pt] node [] {}
		} %end 11
		edge from parent [-,color=gray!60,line width=1pt] node [] {}
	}; %end 1

%	\node[] (z) [left=of a, xshift=-1.5cm,yshift=-0.75cm,label={[label distance=0pt]135:\scalebox{2}{$\mathcal{T} \ =$}}] { };
	%{$\underset{ \ \rotatebox{0}{=}}{\mathcal{T} \quad \ }$}}] { };

\node[] (z3) [below=of a, xshift=-0.12cm,yshift=-0.75cm,label={
	{ $\radd{\origcol_1}{\origcol_1}{\{1\}}{\epred}, \radd{\origcol_1}{\origcol_1}{\{2\}}{\epred}$}}] { };

\node[] (z4) [below=of a, xshift=5.05cm,yshift=-0.75cm,label={
	{ $\recol{\origcol_{1}}{\lambda}^1$}}] { };

\end{tikzpicture}}
\vspace{-3ex}
\end{example}

%\pagebreak

\begin{example}\label{ex:terntree}

The well-decorated tree associated with \Cref{ex:pw2} is the instance %given by the following set of atoms:
\\[-4ex]

$$
\ibtree \cup \{\radd{\origcol_1}{\lambda}{\{1\}}{\rpred}\hspace{-1pt}(\varepsilon)\} \cup
\{\cadd{\origcol_1}{\mu}{\smash{\{\hspace{-1pt}1\hspace{-1pt}\}}}{\lambda}\hspace{-1pt}(1^n)\hspace{-1pt},
\recol{\mu}{\kappa}^1\hspace{-1pt}(1^n)\hspace{-1pt},
\recol{\origcol_{1}}{\mu}^0\hspace{-1pt}(1^n), \newanon(1^{n-1}0) \mid n {\in} \mathbb{N}^+\}
$$
\vspace{-1ex}

\noindent The initial part of the tree is depicted below. %Again, the nodes below the faded-out edges do not carry any decorators. Note that 
The subtrees rooted in $1$ and $11$ are isomorphic.

\vspace{-2.5ex}
\hspace{12ex}\scalebox{0.8}{
	\begin{tikzpicture}%[level/.style={sibling distance=60mm/#1, level distance=15mm}]
		[
		level 1/.style={sibling distance=60mm, level distance=12mm},
		level 2/.style={sibling distance=30mm, level distance=15mm},
		level 3/.style={sibling distance=14mm, level distance=10mm},
%		level 4/.style={sibling distance=10mm, level distance=10mm},
%		level 5/.style={sibling distance=15mm, level distance=10mm}
		]
		
		\node (a){$_{\phantom{|}}\varepsilon_{\phantom{|}}$}
		child {node (n0) {$^{\phantom{|}}0^{\phantom{|}}$}
			child {node (n00) [label={[label distance=-4pt]-2.5:$ $}] {$00$}
				child {node (ntt) {$ $}
					edge from parent [-,color=black,line width=1pt,path fading=South] node [] {}
				}
				child {node (nts) {$ $}
					edge from parent [-,color=gray!60,line width=1pt,path fading=South] node [] {}
			}} %end 00}
		child {node [label={[label distance=-4pt]-2.5:$ $}] (n01) {${01}$}
			child {node (ntt) {$ $}
				edge from parent [-,color=black,line width=1pt,path fading=South] node [] {}
			}
			child {node (nts) {$ $}
				edge from parent [-,color=gray!60,line width=1pt,path fading=South] node [] {}
			}
			edge from parent [-,color=gray!60,line width=1pt] node [] {}
		} %end 01
		edge from parent [-,line width=1pt] node [] {}
	} %end 0
	child {node (n0) [label={[label distance=-4pt]-2.5:$ $}] {$^{\phantom{|}}1^{\phantom{|}}$}
		child {node (n00) {$10$}
			child {node (ntt) {$ $}
				edge from parent [-,color=black,line width=1pt,path fading=South] node [] {}
			}
			child {node (nts) {$ $}
				edge from parent [-,color=gray!60,line width=1pt,path fading=South] node [] {}
			}			edge from parent [-,color=black,line width=1pt] node [] {}
		} %end 00}
	child {node (n01) {$11$}
		child {node (ntt) {$\tiny\vdots\!\! $}
			edge from parent [-,color=black,line width=1pt] node [] {}
		}
		child {node (nts) [label={[label distance=-4pt]-2.5:$ $}] {$\tiny\!\!\vdots$}
			edge from parent [-,color=gray!60,line width=1pt] node [] {}
		}
		edge from parent [-,color=gray!60,line width=1pt] node [] {}
	} %end 01
	edge from parent [-,color=gray!60,line width=1pt] node [] {}
};

\node[] (z3) [below=of a, xshift=1.25cm,yshift=0.97cm,label={
	\mbox{$\ :\ \radd{\origcol_1}{\lambda}{\{1\}}{\rpred}$}}] { };
	
\node[] (lf1) [below=of a, xshift=5.2cm,yshift=-0.35cm,label={
	\mbox{$\ :\ \cadd{\origcol_1}{\mu}{\smash{\{\hspace{-1pt}1\hspace{-1pt}\}}}{\lambda},\recol{\mu}{\kappa}^1,$}}] { };

\node[] (lf1a) [below=of lf1, xshift=-0.10cm,yshift=0.77cm,label={
	\mbox{$\recol{\origcol_{1}}{\mu}^0 $}}] { };

\node[] (lf11) [below=of a, xshift=6.8cm,yshift=-1.82cm,label={
	\mbox{$\ :\ \cadd{\origcol_1}{\mu}{\smash{\{\hspace{-1pt}1\hspace{-1pt}\}}}{\lambda},\recol{\mu}{\kappa}^1,$}}] { };

\node[] (lf11a) [below=of lf11, xshift=-0.10cm,yshift=0.77cm,label={
	\mbox{$\recol{\origcol_{1}}{\mu}^0 $}}] { };

\node[] (lf0) [below=of a, xshift=-2.12cm,yshift=-0.25cm,label={
	\mbox{$\ :\ \newanon$}}] { };

\node[] (lf10) [below=of a, xshift=2.53cm,yshift=-1.72cm,label={
	\mbox{$\ :\ \newanon$}}] { };

	\end{tikzpicture}}
\vspace{-2ex}

\end{example}

%\pagebreak

After having precisely characterized the type of adorned tree-instances that serve as ``instance construction blueprints'', what remains to be done is to formally define the connection between the ``blueprint instances'' (that is, well-decorated trees) and the instances they represent. This is the purpose of the following definition. It states that, on an intuitive level, when presented with a well-decorated tree $\mathcal{T}$, each tree node $s$ will be associated with a colored instance whose active domain, denoted $\fcn{ent}^\mathcal{T}\!(s)$, are $s$'s pseudoleaf descendants (renamed to constants, if applicable). The function $\fcn{col}^{\mathcal{T}}_s$ gives a full account of the $s$-instance's coloring by associating every color of some arity $k$ with the set of $k$-tuples from $\fcn{ent}^\mathcal{T}\!(s)$ that carry that color. Finally, $\inst^{\mathcal{T}}_s$ provides all information regarding the $s$-instance's atoms. For pseudoleafs, their associated colored instance is given explicitly, for the other nodes $s$, the associated colored instance is obtained as a combination of the instances associated to $s$'s children nodes $s0$ and $s1$. The specifics of this combination are governed by the instruction-decorators attached to $s$, in line with the intuitions provided in the previous section.\footnote{By design, descendants of pseudoleafs are irrelevant for the described assembly process and can be entirely disregarded as spurious nodes. Technically, $\fcn{ent}^\mathcal{T}$ assigns them an empty active domain and therefore, they are associated with empty colored instances.} This way, the well-decorated tree $\mathcal{T}$ can indeed be understood as a ``bottom-up assembly workflow'', where every node comes with a colored instance representing an intermediate state of the assembly process, which will be combined with other colored instances further up in the tree. The final result thereof is the instance $\inst^{\mathcal{T}}_\varepsilon$ that is attached to $\mathcal{T}$'s root node $\varepsilon$.    

\begin{definition}\label{def:entities-tree-to-instance}
Let $\mathcal{T}$ be a $\colconsig$-well-decorated tree.
We define the function $\fcn{ent}^\mathcal{T}\!:\{0,1\}^* \to 2^{\mathrm{Cnst} \cup \{0,1\}^*}$ mapping each null $s \in \{0,1\}^*$ to its \emph{entities} (a set of nulls and constants) as follows:
$$\fcn{ent}^\mathcal{T}\!(s) = \{ ss' \mid \newanon(ss') \in \mathcal{T}\} \cup \{ \const{c} \mid \newconst{c}(ss') \in \mathcal{T}, \const{c}\in \mathrm{Cnst}\}.$$%
Every tree node $s$ of $\mathcal{T}$ is assigned a \emph{color} map
$$\fcn{col}^{\mathcal{T}}_s\! \colon \cols \to 2^{(\fcn{ent}^{\mathcal{T}}\!(s))^{\ast}},$$%
indicating which tuples of $s$'s entities are accordingly colored: The color maps are defined such that $\fcn{col}^{\mathcal{T}}_{s}\!(\lambda)$ equals
\begin{itemize}
\item	
$\{ (\const{c}, \ldots, \const{c}) \}  \text{\ iff } \newconst{c}(s)\in \mathcal{T},\ \const{c}\in \mathrm{Cnst}, \text{ and } \lambda=\origcol_{|(\const{c}, \ldots, \const{c})|}$,
\item	
	$\{ (s, \ldots, s) \}  \text{\ iff } \newanon(s)\in \mathcal{T},\  s\in \{0,1\}^*, \text{ and } \lambda=\origcol_{|(s, \ldots, s)|}$,
\item and in all other cases
    \begin{align*}
	& \hspace{-10ex}\bigcup_{\hspace{8.5ex}\cadd{\kappa}{\mu}{I}{\lambda}(s)\in\mathcal{T}} \hspace{-8.5ex}\col_{s0}^{\mathcal{T}}(\kappa)\concat_{I}\col_{s1}^{\mathcal{T}}(\mu)
	\  \! \cup \! \hspace{-8ex}\bigcup_{\hspace{8.5ex}\recol{\kappa}{\lambda}^{i}(s)\in\mathcal{T}}\hspace{-8.5ex}\col_{si}^{\mathcal{T}}(\kappa)
 \ \cup \!\!\! \bigcup_{\ \ i\in \{0,1\}\!\!}
	\small
	\begin{cases}\footnotesize
		\emptyset & \text{if } \recol{\lambda}{\kappa}^{i}(s)\in\mathcal{T} \text{ for any } \kappa\in\cols \\
		\col_{si}^{\mathcal{T}}(\lambda) & \text{otherwise}.
	\end{cases}
\end{align*}
\end{itemize}
%
%We note that $\fcn{col}^{\mathcal{T}}_s$ satisfies $\length{\ve} = \arity{\lambda}$ for all $\lambda\in\cols$ and $\ve\in \fcn{col}^{\mathcal{T}}_s\!(\lambda)$. 
%
Every node $s$ is assigned a set of $\Sigma$-atoms over its entities as indicated by the sets $\mathrm{Atoms}_s$. These sets are defined such that $\mathrm{Atoms}_s$ equals:
\begin{itemize}
	\item \!$\{ \top(\const{c}) \}$ if $\newconst{c}(s)\in\mathcal{T}$ and $\const{c}\in\mathrm{Const}$,
	\item \!$\{ \top(s) \}$ if $\newanon(s)\in\mathcal{T}$,
	\item and otherwise
		\begin{align*}
			 \hspace{-10ex}\bigcup_{\hspace{8.5ex}\radd{\kappa}{\mu}{I}{\rpred}(s)\in\mathcal{T}} \hspace{-8.5ex}\big\{  \rpred(\ve) \mid  \ve  \in \col_{s0}^{\mathcal{T}}(\kappa)\concat_{I}\col_{s1}^{\mathcal{T}}(\mu) \big\} \ 
			 \cup 
			  \hspace{-6.5ex}\bigcup_{\hspace{8ex}\unradd{\kappa}{\rpred}^i(s)\in\mathcal{T}}\hspace{-8ex} \big\{ \rpred(\ve) \mid  \ve\in\col_{si}^{\mathcal{T}}(\kappa) \big\}.
		\end{align*}
\end{itemize}
Let 
$\inst^{\mathcal{T}}_s = \bigcup_{s'\in\{ 0,1 \}^*}\mathrm{Atoms}_{ss'}$ denote the instance associated to the node $s$ of $\mathcal{T}$. Finally we define the \emph{instance $\inst^{\mathcal{T}}$ represented by $\mathcal{T}$} as $\inst^{\mathcal{T}}_{\varepsilon}$ where $\varepsilon$ is the root of $\mathcal{T}$.\defend
\end{definition}

\begin{remark}
The well-versed reader might sense a certain oddity about the above definition. While $\fcn{ent}^\mathcal{T}\!(s)$ is clearly specified for every $s$, the coloring $\fcn{col}^{\mathcal{T}}_s$ and atom set~$\inst^{\mathcal{T}}_{s}$ associated to $s$ are ``defined'' in terms of the colorings and atomsets of $s$'s  child nodes $s0$ and $s1$, which raises justified concerns about the well-foundedness of their definition, in view of the infinity of $\mathcal{T}$. To resolve the oddity, we will now briefly argue why the definition is not problematic.
We observe that both $\fcn{col}^{\mathcal{T}}_s$ and  $\inst^{\mathcal{T}}_{s}$ are fully determined by providing for every tuple $\bold{t}$ over $\fcn{ent}^\mathcal{T}(s)$ the information which colors and predicates that tuple carries.
Now for any such tuple $\bold{t}$, each of its entries $t$ corresponds to one specific pseudoleaf of $\mathcal{T}$, where $t$ is being ``introduced''. By the way the assignments of colorings and atom sets are specified, the colors and predicates carried by $\bold{t}$ are fully determined by any finite initial segment (aka prefix) of $\mathcal{T}$ that contains all the pseudoleafs introducing the entries of $\bold{t}$. Therefore, for every tuple $\bold{t}$ and node $s$, the colors and predicates that $\bold{t}$ carries in the colored $s$-instance are uniquely determined. Consequently, $\fcn{col}^{\mathcal{T}}_s$ and  $\inst^{\mathcal{T}}_{s}$ as a whole are uniquely determined for every $s$.  
Summing up: while the tree $\mathcal{T}$ we operate on is infinite as a whole, any tuple-specific coloring and predicate assignment is uniquely determined through a (bottom-up) inductive definition on some finite prefix of $\mathcal{T}$.

Consequently, an equivalent alternative definition of $\inst^\mathcal{T}$ is obtained by looking at the infinite sequence $\mathcal{T}_0,\ \mathcal{T}_1,\ \mathcal{T}_2,\ldots$ of finite binary decorated trees obtained from $\mathcal{T}$ by truncating it at depth $0,\ 1,\ 2,\ldots$. One can then use a very minor adaptation of \Cref{def:entities-tree-to-instance} to define the corresponding instances $\inst^{\mathcal{T}_0},\ \inst^{\mathcal{T}_1},\  \inst^{\mathcal{T}_2},\ldots$ in a well-founded bottom-up manner. Finally, one lets $\inst^\mathcal{T} = \bigcup_{i\in \mathcal{N}}\inst^{\mathcal{T}_i}$. In fact, this is the strategy Courcelle uses to define clique-width for countable graphs \cite{Cou04}.

%	
%Discuss the issue of non-well-foundedness.
%
%This is a very good question, as indeed the uniqueness of the definition of color map is not immediate.
%To make the argument more explicit one could proceed as follows:
%(a) define the notion of *a* color map for T rather than *the* color map of T, using the given definition
%(b) show that a color map always exists, e.g. by a fixed-point construction
%(c) show that the color map must be unique by the following argument:
%(i) observe that the color map $col^T_s$ is fully determined by telling for every tuple tup of $ent^T(s)$ what this tuples's colors are
%(ii) observe that for any such tuple tup, its colors are fully determined by any finite initial segment (aka prefix) of T that contains all the pseudoleafs occurring in tup
%(iii) consequently for every tuple in question, its colors are uniquely determined
%(iv) consequently, the map $col^T_s$ as a whole is uniquely determined for every s
%Summing up: while the tree T we operate on is infinite as a whole, any tuple-specific coloring is uniquely determined by a (bottom-up) inductive definition on some finite prefix of T. 
\end{remark}

\begin{example}
	Let $\mathcal{T}$ be as given in \Cref{ex:fintree}. We then find the following values for the node-wise entities, colorings, and atom sets. %(while $s$ ranges over all binary strings of length $>2$).
	
	\begin{itemize}
		\item $\begin{array}[t]{@{}l@{}l@{}} 
			    \fcn{ent}^\mathcal{T}\! = \ & \big\{
				\varepsilon {\,\mapsto\,\hspace{-1pt}} \{\const{a},\const{b},\const{c},\hspace{-1pt}11\}\hspace{-0.5pt},
				0 {\,\mapsto\,\hspace{-1pt}} \{\const{a},\const{b}\}\hspace{-0.5pt},
				1 {\hspace{-0.5pt}\,\mapsto\,\hspace{-1pt}} \{\const{c},\hspace{-1pt}11\}\hspace{-0.5pt},
				00 {\,\mapsto\,\hspace{-1pt}} \{\const{a}\}\hspace{-0.5pt},
				01 {\hspace{-0.5pt}\,\mapsto\,\hspace{-1pt}} \{\const{b}\}\hspace{-0.5pt},
				10 {\,\mapsto\,\hspace{-1pt}} \{\const{c}\}\hspace{-0.5pt},
				11 {\hspace{-0.5pt}\,\mapsto\,\hspace{-1pt}} \{11\},\\[1pt]
				& \ \  s {\,\mapsto\,\hspace{-1pt}} \emptyset \mbox{\ \ whenever \ } |s| > 2 \big\} \end{array}$ \\[-0.2ex] 
\item\hspace{-3ex} $\begin{array}[t]{@{\bullet\ \,}l@{\ =\ \big\{ o_1 \mapsto\,}l@{\   \lambda \mapsto\, }l@{  \ o_2 \mapsto\, }l}
 \fcn{col}^{\mathcal{T}}_{\varepsilon} &
  \{\const{a},\const{b}\}, &
  \{\const{c},11\},
  &
  \{(\const{a},\const{a}),(\const{b},\const{b}),(\const{c},\const{c}),(11,11)\} 
  \big\}\\[2pt]
\fcn{col}^{\mathcal{T}}_{0} &
   \{\const{a},\const{b}\}, &
   \emptyset, &
   \{(\const{a},\const{a}),(\const{b},\const{b})\} 
   \big\}\\[2pt]
 \fcn{col}^{\mathcal{T}}_{1} &
\emptyset, &
\{\const{c},11\}, &
\{(\const{c},\const{c}),(11,11)\}\big\}\\[2pt]
\fcn{col}^{\mathcal{T}}_{00} &
\{\const{a}\}, &
\emptyset, &
\{(\const{a},\const{a})\}\big\}\\[2pt]
\fcn{col}^{\mathcal{T}}_{01} &
\{\const{b}\}, &
\emptyset, &
\{(\const{b},\const{b})\}\big\}\\[2pt]
\fcn{col}^{\mathcal{T}}_{10} &
\{\const{c}\}, &
\emptyset, &
\{(\const{c},\const{c})\}\big\}\\[2pt]
\fcn{col}^{\mathcal{T}}_{11} &
\{11\}, &
\emptyset, &
\{(11,11)\}\big\}\\[2pt]
\fcn{col}^{\mathcal{T}}_{s} &
\emptyset, &
\emptyset, &
\emptyset\big\}\quad \mbox{ for all $s$ with $|s|>2$,} \\
\end{array}$\\[1ex]
\item $\inst^{\mathcal{T}}_{s} = \mathrm{Atoms}_{s} = \emptyset$ for all $s \not= \varepsilon$,\\[-2.3ex] 
\item $\inst^{\mathcal{T}}_{0} = \mathrm{Atoms}_{0} = \{\epred(\const{a},\const{a}),\epred(\const{b},\const{b}),\epred(\const{a},\const{b}),\epred(\const{b},\const{a})\}$,\\[-2.3ex]
\item \hspace{5.4ex} $\mathrm{Atoms}_{\varepsilon} = \{\epred(\const{a},\const{c}),\epred(\const{a},11),\epred(\const{b},\const{c}),\epred(\const{b},11)\}$,\\[-2.3ex]
\item $\inst^{\mathcal{T}} = \inst^{\mathcal{T}}_{\varepsilon} = \{\epred(\const{a},\const{a}),\epred(\const{b},\const{b}),\epred(\const{a},\const{b}),\epred(\const{b},\const{a}), \epred(\const{a},\const{c}),\epred(\const{a},11),\epred(\const{b},\const{c}),\epred(\const{b},11)\}$
	\end{itemize}
	It is now easy to check that these values are in line with \Cref{def:entities-tree-to-instance}.
	The isomorphism $\{\const{a} \mapsto \const{a}, \const{b} \mapsto \const{b}, \const{c} \mapsto \const{c}, 11 \mapsto x\}$ witnesses that $\inst^\mathcal{T}$ is indeed isomorphic to $\inst_\mathrm{ex}$. 
\end{example}

%\pagebreak

\begin{example}
Let $\mathcal{T}$ be as given in \Cref{ex:terntree}. We then find the following values for the node-wise entities, colorings, and atom sets (letting $i$ range over $0, 1, 2, \ldots$, while $s$ ranges over $\{0,1\}^+$, and $s'$ over $\{0,1\}^*$.

\begin{itemize}
\item $\fcn{ent}^\mathcal{T}(\varepsilon) = \{0,10,110,1110,\ldots\} = \{1^j0 \mid 0 \leq j \}$,\\[-2ex] 
\item $\fcn{ent}^\mathcal{T}(s) = \{ss' \mid s' \in \{0,1\}^*\} \cap \fcn{ent}^\mathcal{T}(\varepsilon)$,\\[-1ex]
\item $\fcn{col}^{\mathcal{T}}_{\mathrlap{\varepsilon}\phantom{1^{i+1}}} = \{o_1 \mapsto \mathrlap{\{0\},}\qquad\quad\! \mu \mapsto \mathrlap{\{10\},}\qquad\quad\ \   
\kappa \mapsto \mathrlap{\{1^j0 \mid j\geq 2\},}\qquad\qquad\qquad\ \,\ \    \lambda \mapsto \{(1^j0,1^{j+1}0) \mid j\geq 1\}\}$,\\[-2ex]
\item $\fcn{col}^{\mathcal{T}}_{\mathrlap{1^{i+1}}\phantom{1^{i+1}}} = \{o_1 \mapsto \mathrlap{\emptyset,}\qquad\quad\! \mu \mapsto \mathrlap{\{1^{i+1}0\},}\qquad\quad\ \  
\kappa \mapsto \mathrlap{\{1^{j}0 \mid j>i{+}1\},}\qquad\qquad\qquad\ \,\ \   \lambda \mapsto \{(1^j0,1^{j+1}0) \mid j> i\}\}$,\\[-2ex]
\item $\fcn{col}^{\mathcal{T}}_{\mathrlap{1^i0}\phantom{1^{i+1}}} = \{o_1 \mapsto \mathrlap{\{1^i0\},}\qquad\quad\! \mu \mapsto \mathrlap{\emptyset,}\qquad\quad\ \   
\kappa \mapsto \mathrlap{\emptyset,}\qquad\qquad\qquad\ \,\ \   \lambda \mapsto \emptyset\}$,\\[-2ex]
\item $\fcn{col}^{\mathcal{T}}_{\mathrlap{s'0s}\phantom{1^{i+1}}} = \{o_1 \mapsto \mathrlap{\emptyset,}\qquad\quad\! \mu \mapsto \mathrlap{\emptyset,}\qquad\quad\ \   
\kappa \mapsto \mathrlap{\emptyset,}\qquad\qquad\qquad\ \,\ \   \lambda \mapsto \emptyset\}$,\\[-1ex]
\item $\inst^{\mathcal{T}} = \inst^{\mathcal{T}}_{\varepsilon} = \mathrm{Atoms}_{\varepsilon} = \{\rpred(0,1^i0,1^{i+1}0 \mid i\geq 1)\}$ and $\inst^{\mathcal{T}}_{s} = \mathrm{Atoms}_{s} = \emptyset$ for all $s \not= \varepsilon$. 
\end{itemize}
It can now be readily checked that these values are in line with \Cref{def:entities-tree-to-instance}.
Also, the isomorphism $\fcn{iso}: \{0,10,110,1110,\ldots\} \to \{-1,0,1,2,\ldots\}$ with $1^i0 \mapsto i{-}1$ for all $i\geq 0$ witnesses that $\inst^\mathcal{T}$ is indeed isomorphic to $\inst_\mathrm{tern}$. 
\end{example}

Having introduced the machinery for representing countable instances by means of well-decorated trees (which are themselves instances of a specific convenient shape), we are finally ready to formally define the notion of \partitionwidth{}.

\begin{definition}
Given an instance $\instance$ over a finite set $\mathrm{Cnst}$ of constants and a countable set of nulls as well as a finite signature $\Sigma$, the \emph{\ercwidth{}} of $\instance$, written $\ercw{\instance}$, is defined to be the smallest natural number $n$ such that $\instance$ is isomorphic to an instance represented by some $\colconsig$-well-decorated tree such that $|\cols|  =n$.
If no $\colconsig$-well-decorated tree with such a number exists, we let \mbox{$\ercw{\instance}= \infty$}.\defend
\end{definition}

\subsection{Partitionwidth and Monadic Second-Order Logic}

%\textcolor{red}{Piotr: Introduce L-Interpretations here following Blumensath but translated to the database-style model theory database instances.}

Next, we recall known features of partitionwidth related to MSO and utilize them to show several properties which are useful for our purposes. A crucial role in this regard is played by the well-known notion of \emph{MSO-interpretations} 
(also referred to as MSO-transductions or MSO-definable functions in the literature \cite{ArnborgLS88,Courcelle91a,Engelfriet90,Courcelle94,CourEngBook}). 

\begin{definition}[MSO-interpretations]
Let $\Sigma$ and $\Sigma'$ be signatures, and $\mathrm{Cnst}$ and $\mathrm{Cnst}'$ be finite sets of constants. An {\em MSO-interpretation} is a sequence $\mathfrak{I}$ 
of MSO-formulae over $(\Sigma, \mathrm{Cnst})$, containing for every $\rpred \in \Sigma' \cup \{\top\}$ one formula $\phi_\rpred$ with free variables $x_1,\ldots,x_{\arity{\rpred}}$; for every $\const{c} \in \mathrm{Cnst}'$ a formula $\phi_\const{c}$ with free variable $x_1$; and one formula $\phi_=$ with free variables $x_1,x_2$.   
Given an instance $\mathcal{B}$ over $(\Sigma, \mathrm{Cnst})$, let 
$D = \{t \in \adom{\mathcal{B}} \mid \mathcal{B} \models \varphi_\top(t)\}$ and let 
$\fcn{f}: D \to \mathbf{N}\cup \mathrm{Cnst}'$ be an arbitrary but fixed function satisfying 
$\fcn{f}(t)=\fcn{f}(t')$ \iffi $\mathcal{B}\models \varphi_=(t,t')$,
and also $\fcn{f}(t)=\const{c}$ \iffi $\mathcal{B}\models \varphi_\const{c}(t)$.
If no such a function exists, $\mathfrak{I}(\mathcal{B})$ is undefined, otherwise let
$$
\mathfrak{I}(\mathcal{B})=\textstyle\bigcup_{\rpred \in \Sigma' \cup \{\top\}}\{\rpred(\fcn{f}(\vt)) \mid \mathcal{B} \models {\varphi}_{\rpred}(\vt),\ \vt \in D^{\arity{\rpred}} \}.
$$

An instance $\mathcal{A}$ is \emph{MSO-interpretable} in an instance $\mathcal{B}$ (written $\mathcal{A} \leq_{\text{MSO}} \mathcal{B}$) if there exists an MSO-interpretation $\mathfrak{I}$ such that $\mathcal{A} \cong \mathfrak{I}(\mathcal{B})$. It is \emph{(MSO-)tree-interpretable} if $\mathcal{A} \leq_{\text{MSO}} \mathcal{T}$ for some adorned tree $\mathcal{T}$.\defend
\end{definition}

According to standard results about MSO-interpretations, for every interpretation $\mathfrak{I}$ from $(\Sigma, \mathrm{Cnst})$ to $(\Sigma', \mathrm{Cnst}')$ and every MSO formula $\Xi$ over $(\Sigma', \mathrm{Cnst}')$ one can effectively construct a formula $\Xi^{\mathfrak{I}}$ over $(\Sigma, \mathrm{Cnst})$ such that for any two instances $\mathcal{A}, \mathcal{B}$ satisfying 
%$\mathfrak{I} : \mathcal{A} \leq_{\text{MSO}} \mathcal{B}$ 
$\mathcal{A} \cong \mathfrak{I}(\mathcal{B})$,
the following holds:
$$\mathcal{A} \models \Xi \iff \mathcal{B} \models \Xi^{\mathfrak{I}}.$$%
%
%The formula $\Xi^{\mathfrak{I}}$ can be effectively constructed by replacing every relation symbol $\rpred$ by its $\mathfrak{I}$-definition $\phi_{\rpred}$, replacing every constant $\const{c}$ by $\phi_{\const{c}}$, replacing every $=$ by $\phi_{=}$, and by relativising every quantifier to $\phi_{\top}$, where set quantifiers are further relativised to sets closed under $\phi_{=}$. 
%
Moreover, as MSO-interpretations can be composed, the MSO-interpretability relation ($\leq_{\text{MSO}}$) is transitive:
%\begin{observation}\label{obs:MSO-interpretations-transitive}
whenever $\mathcal{A} \leq_{\text{MSO}} \mathcal{B}$ and $\mathcal{B} \leq_{\text{MSO}} \mathcal{C}$ hold, then also $\mathcal{A} \leq_{\text{MSO}} \mathcal{C}$.
%\end{observation}
Blumensath showed that, for every set of decorators, the corresponding mapping $\mathcal{T} \mapsto \inst^{\mathcal{T}}$ described in \Cref{def:entities-tree-to-instance}
can be realized by an appropriate MSO-interpretation \cite[Prop. 17]{Blumensath06}.
\begin{lemma}\label{lem:computable-mso-tree-interpretations}
For every set of decorators $\decorum$, one can compute an MSO-interpretation $\mathfrak{I}$ such that for every $\colconsig$-well-decorated tree $\mathcal{T}$, we get $\inst^{\mathcal{T}} \cong \mathfrak{I}(\mathcal{T})$.
\end{lemma}
\remarkstart
In view of this result, establishing MSO-friendliness is then rather straightforward.
\color{black}
\begin{theorem}\label{thm:cliquewidth-k-MSO-decidable}
	\Ercwidth{} is an MSO-friendly width measure.
\end{theorem}
\begin{proof}
Let $n$ be a fixed natural number. Take a MSO sentence $\Xi$ over some signature $\Sigma$ and a set of constants $\mathrm{Cnst}$. We will show how to decide if there exists a model $\inst$ of $\Xi$ such that $\ercw{\inst} \leq n$. This implies the claimed result as we can enumerate all MSO sentences and check model existence.%In turn this will grant us the theorem as MSO formulae form a countable set.

Up to renaming of colors, there are only finitely many distinct color sets $\cols$ with $|\cols| \leq n$, where each color is assigned an arity between $1$ and $\arity{\Sigma}$. 
By definition, $\Xi$ has a model of partitionwidth $\leq n$ \iffi for one of these $\cols$ there exists a $\colconsig$-well-decorated tree $\mathcal{T}$ satisfying $\inst^\mathcal{T} \models \Xi$. With the corresponding MSO-interpretation $\mathfrak{I}$ from \cref{lem:computable-mso-tree-interpretations}, this holds \iffi $\mathcal{T} \models \Xi^\mathfrak{I}$.
It is easy to see that there exists a MSO-sentence $\Theta$ which is satisfied by a given decorated tree \iffi it is well-decorated (cf. \Cref{def:well-dec}).
Consequently, $\Xi$ is satisfiable in a model of partitionwidth $\leq n$ \iffi $\Xi^\mathfrak{I}\wedge \Theta$ is satisfiable in any $\colconsig$-decorated tree (for an appropriate $\cols$).
Let $\Xi' = \exists \decorum. (\Xi^\mathfrak{I}\wedge \Theta)$, i.e. it is the sentence obtained by reinterpreting all unary predicates of $\Xi^{\mathfrak{I}}\wedge \Theta$ as MSO set variables and quantifying over them existentially. Clearly, $\Xi'$ is a MSO sentence over the signature $\{\szero,\sone\}$ which is valid in $\ibtree$ \iffi some well-decoration of $\ibtree$ makes $\Xi^{\mathfrak{I}}$ true. Thus, we have reduced our problem to checking the validity of finitely many MSO sentences (one for each $\cols$) in $\ibtree$, which is decidable by Rabin's Tree Theorem \cite[Thm 1.1]{Rabin69}.
\end{proof}

%\newpage

By \Cref{lem:computable-mso-tree-interpretations}, any $\fcn{pw}$-finite $\instance$ is tree-interpretable.
Remarkably, Blumensath showed that the converse also holds \cite[Thm. 41]{Blumensath06}.

\begin{theorem}\label{thm:pw-coincides-mso-interpretability-trees} For every instance $\inst$ the following holds:
$\instance$ is $\fcn{pw}$-finite \iffi there is some  $\mathcal{T} \in \mathbb{T}$ such that $\instance \leq_{\text{MSO}} \mathcal{T}$.   
\end{theorem}

This coincidence of tree-interpretability and finiteness of par\-ti\-tion\-width shows that the latter is uniquely positioned among the width measures. We will exploit it in various ways; 
%
% \subsection{Preservation under MSO-interpretations}
%
% One much desired property of a width notion is preservation of its finiteness under various model-theoretic transformations. 
%Informally speaking, this allows for a single, well-behaved representation of theories expressed by multiple formalisms ``layered'' on top of each other\footnote{For a more applied example on this matter see \cref{sec:strat}.}. 
most importantly, we can use it to show that finiteness of partitionwidth is preserved under MSO-in\-ter\-pre\-tations.

\begin{theorem}
If $\mathcal{A} \leq_{\text{MSO}} \mathcal{B}$ and $\mathcal{B}$ is $\fcn{pw}$-finite, then $\mathcal{A}$ is $\fcn{pw}$-finite.
\end{theorem}
\begin{proof}
Given a $\fcn{pw}$-finite $\mathcal{B}$, \Cref{thm:pw-coincides-mso-interpretability-trees} implies  $\mathcal{B} \leq_{\text{MSO}} \mathcal{T}$ for some $\mathcal{T}\in\mathbb{T}$.
Then, $\mathcal{A} \leq_{\text{MSO}} \mathcal{B}$ and transitivity of $\leq_{\text{MSO}}$ yield $\mathcal{A} \leq_{\text{MSO}} \mathcal{T}$ and thus, again by \Cref{thm:pw-coincides-mso-interpretability-trees}, it follows that $\mathcal{A}$ is $\fcn{pw}$-finite.
\end{proof}

\subsection{Robustness under Controlled Modifications}\label{sec:preservation}

We will now turn our attention to certain modifications of instances that preserve width finiteness properties. They all have in common that the instance's domain stays the same. We start with an easy, yet useful observation (the proof of which can be found in the appendix).

\begin{definition}
Let $\signature$ be a signature, $\pred{U} \in \signature$ a unary predicate, $\inst$ an instance over $\signature$, and $T \subseteq \adom{\inst}$ a set of domain ele\-ments.
We define the \emph{unary replacement of $\inst$} w.r.t. $\pred{U}$ and~$T$ as the instance $\inst^{\pred{U}:=T} = \inst \setminus \{\pred{U}(t) \in \inst\} \cup \{\pred{U}(t) \mid t \in T \}$.\defend
\end{definition} 

\begin{observation}\label{prop:unary-replacements}
Finiteness of treewidth, cliquewidth, and \hypercliquewidth{} are all preserved under unary replacements. 
\end{observation}

Finite \hypercliquewidth{} exhibits further convenient properties, which will turn out to be quite useful later on: it is preserved when extending an instance by defined relations, where definitions can be given through first-order formulas or regular path expressions. This is a rather straightforward consequence from the fact that all such definitorial extensions are particular cases of MSO-interpretations (see the appendix for the detailed argument).

\begin{definition}\label{def:extensions}
Let $\signature$ be a signature, and $\inst$ be an instance over $\signature$. 
We define \emph{two-way regular path expressions} (2RPEs) over $\signature$ to be regular expressions over $\signature_{\smash{\mathrm{bin}}}^\leftrightarrow$, the set containing the symbols $\ppred^\rightarrow$ and $\ppred^\leftarrow$ for each binary $\ppred\in \signature$.
We associate with any such 2RPE $E$
the binary relation $E^\instance$ over $\adom{\inst}$:
\begin{align*}
	(\pred{R}^\rightarrow)^\instance = &\ \{(t_1,t_2) \mid \pred{R}(t_1,t_2) \in \instance\} \\[-0.5ex]
	(\pred{R}^\leftarrow)^\instance = &\  \{(t_2,t_1) \mid \pred{R}(t_1,t_2) \in \instance\} \\[-0.5ex]
	(E_1E_2)^\instance = &\  \{(t_1,t_3) \mid (t_1,t_2)\in (E_1)^\instance, (t_2,t_3)\in (E_2)^\instance\}  \\[-0.5ex]
	(E_1 \cup E_2)^\instance = &\  (E_1)^\instance \cup (E_2)^\instance \\[-0.5ex]
	(E^*)^\instance = &\  \{(t,t) \mid t \in \adom{\instance}\} \cup (E)^\instance \cup (EE)^\instance \cup \ldots
\end{align*}
Given a binary predicate $\ppred \in \signature$, we define the \emph{2RPE-extension of $\inst$} w.r.t. $\ppred$ and $E$ as the instance $\inst^{\ppred:=\ppred\cup E} = \instance \cup \{\ppred(t,t') \mid (t,t')\in E^\instance \}$.  
Moreover, given a $|\vx|$-ary predicate $\rpred \in \signature$ and a first-or\-der formula $\varphi(\vx)$ with free variables $\vx$, we define the \emph{FO-extension of $\inst$} w.r.t. $\rpred$ and $\varphi$ as the instance $\inst^{\rpred:=\rpred\cup\varphi} = \inst \cup \{\rpred(\vt) \mid \inst \models \varphi(\vt)\}$.  \defend
%Any such $\inst^{\rpred:=\phi}$ will be called an FO-extension of $\inst$.   
\end{definition} 

\begin{theorem}\label{thm:closureproperties}
If $\inst$ has finite \hypercliquewidth{}, then so does any 2RPE-extension and any FO-extension of $\inst$.
\end{theorem}

%\clearpage

%As an aside, note that iterating and interleaving these two types of definitorial extensions allow us to precisely realize extensions defined by \emph{binary transitive closure logic}, sometimes denoted $FO^*$, which extends first order logic by the transitive closure operator $TC_{x,y}$ where the formula $TC_{x,y}\varphi(x,y)$ corresponds to the transitive closure of the binary relation expressed by $\varphi(x,y)$.  

\subsection{Relationship to Other Width Notions}\label{sec:hcw-vs-tw-cw}

In this section we are going to clarify how the notion of finite \ercwidth{} relates to those of finite treewidth and finite cliquewidth. 
\Cref{fig:relationships} summarizes these insights.

\begin{figure}[t]
\newcommand{\nlarge}{\small}
\newcommand{\roundy}{15pt}
%\vspace{-0.5ex}
\centerline{
\scalebox{1.00}{%\hspace{-1.6pt}
\begin{tikzpicture}
%	\draw
%	(0,0) {[rounded corners=15pt] --
%		++(2,0)  -- 
%		++(0,1)} --
%	++(-2,0) --
%	cycle;
	\draw[rounded corners=\roundy]
	(0,0.4) rectangle ++(12.7,4.6);
	\draw[rounded corners=\roundy]
    (0,0.4) rectangle ++(8.2,4.6);
	\draw[rounded corners=\roundy]
    (0,0.4) rectangle ++(4.2,4.6);
	\draw[rounded corners=\roundy]
    (0,0.4) rectangle ++(12.7,1.45);
	\draw[rounded corners=\roundy]
    (0,0.4) rectangle ++(8.2,3.0);
	\draw[rounded corners=\roundy]
    (0,0.4) rectangle ++(1.0,3.0);
    \node at (0.49,1.7) {\rotatebox{270}{\nlarge finite \ \ domain}};
%    \node at (0.6,2.85) {\large domain};
    \node at (2.1,4.75) {\nlarge \hspace{8ex} finite treewidth};
    \node at (2.2,4.0) {\nlarge $\inst_\mathrm{tern} = \{ \rpred(-1,n,n{+}1) \}$};
    \node at (2.5,2.5) {\nlarge $\{ \rpred(n,n{+}1,n{+}2)\}$};
    \node at (2.75,1.1) {\nlarge $\{ \ppred(n,n{+}1) \}$};
    \node at (6,4.7) {\nlarge \hspace{5ex} finite \hypercliquewidth{}};
    \node at (6.15,4.0) {\nlarge $\{ \rpred(-1,n,n{+}m) \}$};
    \node at (6,3.1) {\nlarge \hspace{7ex} finite cliquewidth};
    \node at (6.1,2.5) {\nlarge $\{ \rpred(n,n{+}1,n{+}m) \}$};
    \node at (6.2,1.1) {\nlarge $\{ \ppred(n,n{+}m) \}$};
    \node at (10.45,1.6) {\nlarge \hspace{8.5ex} binary signature};
    \node at (10.4,1.1) {\nlarge $    	
    	   \begin{array}{l}
    		\\
    		\{ \ppred(\langle n,\!m \rangle, \langle n{+}1,\!m \rangle)  \}\ \cup\\[-1pt]
    		\{ \ppred'(\langle n,\!m \rangle, \langle n,\!m{+}1 \rangle) \}
    	    \end{array}
            $};
    \node at (10.45,4.75) {\nlarge \hspace{15ex} countable};
    \node at (10.4,4.0) {\nlarge $
    	\begin{array}{l}
    		\\
    	\ \{ \ppred(\langle n,\!-1 \rangle, \langle n{+}1,\!-1 \rangle)  \}\ \cup\\
    	\ \{ \rpred(\langle n,\!-1 \rangle,\langle m,\!-1 \rangle, \langle n,\!m \rangle ) \}
        \end{array}
    	$};
\end{tikzpicture}}}

%\vspace{-1ex}

\caption{Set diagram displaying the inclusions between classes of countable instances satisfying diverse types of width-finiteness, conditioned on signature arity. Example instances witness the strictness of the inclusions ($n$ and $m$ range over $\mathbb{N}$).\label{fig:relationships}}
\vspace{-1ex}
\end{figure}

Without providing full formal details, we recall that the notion of \emph{cliquewidth} in its version for countable instances \cite{ICDT2023} follows a very similar idea as \ercwidth{}, but only allows for unary colors and uses a different set of assembly operations: nullary operations $\pred{c}_\kappa$ and $\pred{*}_\kappa$, creating singleton instances of $\kappa$-colored constants or nulls, respectively;
a binary operator $\oplus$ creating the (plain) disjoint union; unary operators $\fcn{Recolor}_{\kappa\shortto\mu}$ assigning all $\kappa$-colored terms the new color $\mu$; and another type of unary operators $\fcn{Add}_{\pred{R},\kappa_1,...,\kappa_n}$, creating atoms $\rpred(t_1,\ldots,t_n)$ whenever each term $t_i$ carries the corresponding color $\kappa_i$ for all $i$. As an example, the instance from \Cref{ex:pw2} can be created through the expression
$$ \fcn{Add}_{\pred{E},\kappa,\mu}\big(\fcn{Add}_{\pred{E},\kappa,\kappa}(\pred{a}_\kappa \oplus \pred{b}_\kappa)  \oplus \fcn{Recolor}_{\kappa\shortto\mu}( \pred{c}_\kappa \oplus \pred{*}_\kappa )\big). $$%
Earlier investigations showed that, while finite treewidth implies finite cliquewidth for binary signatures, the two notions are incomparable for signatures of higher arity \cite{ICDT2023}: notably, the example $\inst_\mathrm{tern}$ from \Cref{ex:pw2} exhibits a treewidth of 2 but infinite cliquewidth.
As demonstrated by \Cref{ex:pw2}, however, $\inst_\mathrm{tern}$ does have finite \ercwidth{}. In fact, the class of instances with finite \ercwidth{} subsumes both the class of treewidth-finite and the class of cliquewidth-finite instances, independent of the signature's arity. With the help of \Cref{thm:pw-coincides-mso-interpretability-trees}, this result can be derived from the fact that these classes consist of tree-interpretable instances (the detailed argument can be found in the appendix).

\begin{theorem}\label{thm:tree-width-implies-clique-width}
Let $\inst$ be a countable instance over a finite signature of arbitrary arity. Then the following hold:
\begin{itemize}
    \item If $\inst$ has finite treewidth, then $\inst$ has finite \ercwidth{}.
    \item If $\inst$ has finite cliquewidth, then $\inst$ has finite \ercwidth{}.
\end{itemize}
\end{theorem}

%\subsection{\Ercwidth{} and~Cliquewidth}\label{sec:hcw-vs-cw}

%We have shown in \cref{sec:I_tern_fin_hcw} that $\inst_{\mathrm{tern}}$ has finite \ercwidth{} while having infinite cliquewidth.

%In this section we will prove two results:
%Firstly, the class of instances with finite \ercwidth{} strictly subsumes the class of instances with finite cliquewidth.
%Secondly, the notions of finite \ercwidth{} and finite cliquewidth coincide for binary signatures.

Both inclusions of \Cref{thm:tree-width-implies-clique-width} are strict; there even are instances of finite \hypercliquewidth{} whose treewidth and cliquewidth are both infinite. The instance 
$$\inst_\mathrm{tern*} = \{ \rpred(-1,n,n{+}m) \mid n,m\in \mathbb{N}\}$$%
has infinite treewidth since any two natural numbers co-occur in some atom, enforcing that any tree-decomposition must contain some node $X$ with $\mathbb{N} \subseteq  X$.  
It also has infinite cliquewidth, which can be argued along the same lines as for $\inst_\mathrm{tern}$ \cite{ICDT2023}.
Yet, $\inst_\mathrm{tern*}$ has finite \ercwidth{}, as it is isomorphically represented by the decorated tree corresponding to the (infinite) assembly expression
$
\const{*} \oplus_{\{
\radd{o_1}{\lambda}{\{1\}}{\rpred}
\}} E
$
where
$$
E = \big( \const{*} \oplus_{\{
\recol{o_1}{\kappa}^0, 
\recol{o_2}{\lambda}^0,
\cadd{o_1}{\kappa}{\{1\}}{\lambda}
\}} E \big).
$$

It should, however be noted that, finite \ercwidth{} and finite cliquewidth coincide for binary signatures (see the appendix for the proof).

\begin{proposition}\label{th:hcw=cw-bin}
	Any countable instance over a binary signature has finite cliquewidth \iffi it has finite \ercwidth{}.
\end{proposition}

At the end of this section, let us outline how the just obtained insights regarding partitionwidth can be harnessed to show decidability of very expressive querying in some of the most expressive contemporary description logics.

%\pagebreak

\begin{example}\label{ex:srfamily}
The description logics $\mathcal{SRIQ}b_s$, $\mathcal{SROI}b_s$, and $\mathcal{SROQ}b_s$ are extensions of the mentioned $\mathcal{ALCHIQ}b$, $\mathcal{ALCHOI}b$, and $\mathcal{ALCHOQ}b$, respectively. One additional feature of $\mathcal{SR}*$ knowledge bases is an \emph{RBox}, which contains \emph{role-chain axioms} like transitivity statements for binary relations such as  $\pred{Friend}\circ\pred{Friend} \sqsubseteq \pred{Friend}$ but also more complex relationships like $\pred{Friend}\circ\pred{Enemy} \sqsubseteq \pred{Enemy}$. Such axioms prevent the existence of treewidth-finite (finitely) universal model sets, but not of partitionwidth-finite ones, as can be shown via known model-theoretic insights: Taking some model, it can be unravelled into a ``pre-model'' of finite treewidth (and thus partitionwidth), as one would do in the RBox-free case. This structure can then be extended to a model proper, coping with the still violated RBox axioms by applying RPE-extensions \cite{Kazakov08,CalvaneseEO09}, which -- as we now know -- preserve finiteness of partitionwidth. It is easy to see that the model thus obtained maps homomorphically into the original one. Hence, we obtain decidability of HUSOMSOQs for all these description logics. For query languages up to UC2RPQs, this had already been known (even with tight complexity bounds \cite{CalvaneseEO09,BednarczykR19}), but for more expressive query languages, decidability had hitherto been open.
\end{example}

%SECTION ?: SECOND PRLIMINARIES
\section{Application to Existential Rules}\label{sec:rules}
In the remainder of the paper, we show how our decidability framework can be fruitfully applied to the very popular logical formalism of \emph{existential rules} -- also referred to as \emph{tuple-generating dependencies} (TGDs)~\cite{AbiteboulHV95}, \emph{conceptual graph rules}~\cite{SalMug96}, Datalog$^\pm$~\cite{Gottlob09}, and \emph{$\forall \exists$-rules}~\cite{BagLecMugSal11} -- which has become a de-facto standard of ontological querying over databases. As entailment of even the most simple queries is undecidable for unconstrained existential rules \cite{ChandraLM81}, there has been ample research into fragments that restore decidable querying. This includes investigations based on width notions, although existing work almost entirely focuses on (single) universal models of finite treewidth.\footnote{Two exceptions are our prior work proposing \emph{cliquewidth-finiteness} as alternative to tree\-width-finite\-ness for universal models \cite{ICDT2023}, and a work investigating \emph{finitely} universal models of finite treewidth~\cite{BMR2023}.} Thus, the upcoming part of the paper will significantly advance this line of research. Before we introduce relevant notions of existential rules, the reader may want to have a look at the following example.

\begin{example}\label{ex:Rbio}
We present the ruleset $\ruleset_\mathrm{bio}$ about biological relationships describes female/male individuals% 
%($\pred{Fem}$, $\pred{Mal}$, $\pred{Ind}$)
, related by mother-, father-, parent- and ancestorhood, as well as compatibility% 
%($\pred{Mth}$, $\pred{Fth}$, $\pred{Par}$, $\pred{Anc}$, $\pred{Cmp}$)
. A ternary relation ($\pred{CommonDescendant}$) indicates two individuals' ``earliest'' joint descendant, and helps to describe hereditary patterns of some genetic risk% %($\pred{Risk}$)
. As usual, we omit universal quantifiers and use $\leftrightarrow$ as a shorthand for two rules.
{\small
	\begin{align}
		\pred{Individual}(x) \to & \ \exists yz.\pred{Mother}(x,\!y) {\,\wedge\,} \pred{Father}(x,\!z) \tag{1}\\[-1pt]
		\pred{Mother}(x,\!y) \leftrightarrow & \ \pred{Parent}(x,\!y) {\,\wedge\,} \pred{Female}(y) \ \  \hspace{5ex} \pred{Father}(x,\!y) \leftrightarrow \pred{Parent}(x,\!y) {\,\wedge\,} \pred{Male}(y) \tag{2}\\[-1pt]
		\pred{Female}(x) \to  & \ \pred{Individual}(x) \hspace{21.3ex}  \pred{Male}(x) \to  \pred{Individual}(x) \tag{3}\\[-1pt]
		\pred{Compatible}(x,\!y) \leftrightarrow & \ \pred{Female}(x) \wedge \pred{Male}(y)   \tag{4}\\[-1pt]
		\pred{Parent}(x,\!y) \to & \ \pred{Ancestor}(x,\!y) \hspace{5.8ex} \pred{Ancestor}(x,\!y) {\,\wedge\,} \pred{Ancestor}(y,\!z) \to \pred{Ancestor}(x,\!z) \tag{5}\\[-1pt]
		& \!\hspace{-21ex} \pred{Mother}(x,\!y) {\,\wedge\,} \pred{Ancestor}(y,\!y') {\,\wedge\,} \pred{Father}(x,\!z) {\,\wedge\,} \pred{Ancestor}(z,\!z') \to \pred{CommonDescendant}(x,\!y'\!,\!z') \tag{6}\\[-1pt]
		& \!\hspace{18ex}  \pred{Risk}(x) {\,\wedge\,} \pred{Mother}(x,y) \to \pred{Risk}(y) \tag{7}\\[-1pt]
		& \hspace{-6ex} \pred{CommonDescendant}(x,\!y,\!z) {\,\wedge\,} \pred{Risk}(y) {\,\wedge\,} \pred{Risk}(z) \to \pred{Risk}(x) \tag{8}
	\end{align}}
\vspace{-4ex}
\end{example}

\pagebreak

The rest of the section is devoted to introducing the necessary preliminaries on existential rules, the new contributions commence in \Cref{sec:finitewsets}.

\medskip \noindent {\bf Existential rules and knowledge bases.}
An \emph{(existential) rule} is a first-order sentence of the  
form~$\rho = \forall{\vx\vy}. \big( \phi(\vx, \vy)
\rightarrow \exists{\vz}.\psi(\vy, \vz) \big)$,
where~$\vx$, $\vy$, and $\vz$ are mutually disjoint tuples of variables, and both 
the \emph{body}~$\phi(\vx, \vy)$ and the \emph{head}~$\psi(\vy, \vz)$ of~$\rho$  (denoted with~$\body(\rho)$ and~$\head(\rho)$, respectively) are conjunctions (possibly empty, sometimes construed as sets of their conjuncts) of atoms containing %constants and 
the indicated variables.
The \emph{frontier} $\fr(\rho)$ of $\rho$ is the set of variables $\vy$ shared between body and head. For better readability, we often omit the universal quantifiers prefixing existential rules. 
We say that a rule~$\rho$ is 
%\mbox{(i)~\emph{$n$-ary}} \iffi all predicates appearing in~$\rho$ are of arity at most~$n$, 
%\mbox{(ii) \emph{single-headed}} \iffi $\head(\rho)$ contains at most one atom, (iii) 
\emph{full} \iffi $\rho$ does not contain an existential quantifier%
%, and (iv) \emph{non-datalog} \iffi $\rho$ contains an existential quantifier
. We call a finite set of existential rules $\ruleset$ a {\em ruleset}. Satisfaction of a rule $\rho$  (a ruleset $\ruleset$) by an instance $\instance$ according to standard FO semantics is written $\instance \models \rho$ ($\instance \models \ruleset$, respectively). Given a database $\database$ and a ruleset $\ruleset$, we define the pair $(\db,\ruleset)$ to be a \emph{knowledge base}. Then, an instance $\instance$ is a \emph{model} of the knowledge base (i.e., of $\database$ and $\ruleset$), written $\instance\models (\database,\ruleset)$, \iffi $\database \subseteq \instance$ and $\instance \models \ruleset$. 

%\pagebreak
\medskip \noindent {\bf Trigger, rule application, and Skolem chase.}
An existential rule $\rho = \phi(\vx,\vy) \rightarrow \exists \vz. \psi(\vy,\vz)$ is \emph{applicable} to an instance $\instance$ \ifandonlyif there is a homomorphism $\hism$ mapping $\phi(\vec{x},\vec{y})$ to $\instance$. We then call $(\rho, \hism)$ a \emph{trigger} of $\instance$.
The \emph{application} of a trigger $(\rho, \hism)$ in $\instance$ yields the instance $\fchain{\inst,\rho,\hism} = \instance \cup \bar\hism(\psi(\vec{y},\vec{z}))$, where $\bar\hism$ is the extension of $\hism$ that maps every variable $z$ from $\vz$ to a null denoted \smash{$z_{\rho,\hism(\vy)}$.} 
Note that $\bar\hism(\psi(\vec{y},\vec{z}))$ is then unique for a given rule and the mapping of its frontier variables, with the effect that applying $\rho$ twice with homomorphisms that agree on the frontier variables does not produce more atoms, that is,
$\fchain{\fchain{\inst,\rho,\hism},\rho,\hism'} = \fchain{\inst,\rho,\hism}$ whenever $\hism(\vy) = \hism'(\vy)$. Moreover, applications of different rules or the twofold application of the same rule with a different frontier-mapping to some instance $\instance$ are independent from each other and the order of their application is irrelevant. This allows us to define the parallel one-step application of all applicable rules as 
$$\ksat{1}{\inst,\ruleset} = 
\hspace{-9ex}\bigcup_{\hspace{9ex}\rho\in \ruleset,\,(\rho, \hism) \text{\,trigger of\,} \instance}\hspace{-9ex} 
\fchain{\inst,\rho,\hism}.$$     

To define the \emph{(breadth-first) Skolem chase sequence} we let 
$\ksat{0}{\inst,\ruleset} = \inst$ as well as $\ksat{i+1}{\inst,\ruleset} = \ksat{1}{\ksat{i}{\inst,\ruleset}, \ruleset}$, and obtain the \emph{Skolem chase} $$\sat{\inst,\ruleset} = \bigcup_{i\in \mathbb{N}} \ksat{i}{\inst,\ruleset},$$which is a universal model of $(\inst,\ruleset)$ \cite{fkmp05,marnette09}. We note that the Skolem chase of a countable instance is countable, as  is the number of overall rule applications performed to obtain it. 
In accordance with standard FO-semantics, an instance $\inst$ and ruleset $\ruleset$ \emph{entail} a BCQ $\query {\,=\,} \exists \vx .\phi(\vx)$, written $(\inst,\! \ruleset) {\,\models\,} \query$  \iffi $\phi(\vx)$ maps homomorphically into each model of $\inst$ and $\ruleset$. This coincides with  the existence of a homomorphism from $\phi(\vx)$ into any universal model of $\inst$ and $\ruleset$ (e.g., the Skolem chase $\sat{\inst, \ruleset}$).

\medskip \noindent {\bf Rewritings and finite-unification sets.}
Given a ruleset $\ruleset$ and a CQ $\query(\vy)$, we say that a UCQ $\psi(\vy)$ is a {\em rewriting} of $\query(\vy)$ under the ruleset $\ruleset$ \ifandonlyif for any database $\database$ and any tuple of its elements $\vecconst{a}$ the following holds:
$$\sat{\database,\ruleset} \models \query(\mbox{$\vecconst{a}$}) \ \ifandonlyif \ \database \models \psi(\vecconst{a}).$$%
A ruleset $\ruleset$ is a {\em finite-unification set} ($\fus$) \iffi for every CQ, there exists a UCQ rewriting~\cite{BagLecMugSal11}, in which case it is also always possible to compute it \cite{fusalgo}. This property is also referred to as \emph{first-order rewritability}.
If a ruleset $\ruleset$ is $\fus$, then for any given CQ $\query(\vy)$, we fix one of its rewritings under $\ruleset$ and denote it with $\rewrs{\ruleset}{\query(\vy)}$.

%

%SECTION ?: PRLIMINARIES
\section{Abstracting from Databases: Finite w-Sets}\label{sec:finitewsets}

When investigating existential rules, one is commonly interested in semantic properties of rulesets that hold independently of the associated database. This motivates the following definition.

\begin{definition}[finite-$\fcn{w}$ set]\label{def:finitewset}
	A ruleset $\ruleset$ is called a \emph{finite-$\fcn{w}$ set} \iffi for every database $\database$, there exists a $\fcn{w}$-finite universal model $\instance^*$ of $(\db,\ruleset)$.\defend
\end{definition}

Various previously defined classes of rulesets with decidable querying can be seen as instances of this framework: \emph{finite-expansion sets} ($\fes$) are obtained by using the $\fcn{expansion}$ function returning an instance's domain size as a width measure; \emph{finite-treewidth sets} ($\bts$) are obtained when using treewidth;\footnote{In the literature, finite-treewidth sets are also often referred to as ``bounded treewidth sets'' and denoted \textbf{bts}. Beyond being misleading in this case, the usage of the \textbf{bts} notion is also ambiguous; it may also refer to structural properties of a specific type of monotonic chase.} our more recently proposed \emph{finite-cliquewidth sets}~($\fcs$) are the outcome of employing cliquewidth \cite{ICDT2023}. Due to the established relationships between the width measures, we immediately obtain the following result regarding the novel class of \emph{finite-\ercwidth{} sets} ($\fhcs$).

\begin{corollary}
$\fhcs$ subsumes all of $\fes$, $\bts$, and $\fcs$.
\end{corollary}

%\clearpage

By virtue of entirely subsuming $\bts$ (\cref{thm:tree-width-implies-clique-width}), $\fhcs$ contains all fragments of the guarded family, ranging from the most basic (guarded rules) to the most elaborate ones (such as glut-frontier-guarded rules) \cite{BagLecMugSal11,CaliGK13,KrotzschR11}, as well as the various $\fes$ rule classes based on acyclicity \cite{GHKKMMW2013}. 
By virtue of subsuming $\fcs$, $\fhcs$ inherits the corresponding decidable subclasses~\cite{ICDT2023}; in particular, any $\fus$ ruleset consisting of single-headed rules over a binary signature is $\hcs$.
While for all of these, decidability of querying was previously only established for query formalisms of significantly lower expressivity (though sometimes beyond CQs and with precise complexities \cite{ijcai2017-110}), \Cref{cor:dechomquerying} implies that decidability also holds for more expressive queries like NEMODEQs and MDDlog queries. Yet, the novel $\fhcs$ notion also encompasses a great variety of rulesets which aren't contained in any of the previously established classes.

\begin{example}
As an example, consider again the ruleset $\ruleset_\mathrm{bio}$ from \Cref{ex:Rbio}. $\ruleset_\mathrm{bio}$ does not fall under any of the previously known decidable classes: even with the simple database $\database_\const{a}$ just containing the fact $\pred{Individual}(\const{a})$, any universal model of $(\database_\const{a},\ruleset_\mathrm{bio})$ will have an infinite domain (due to indefinite ``value invention'' through the rules (1)--(4)), infinite treewidth (due to the transitivity of the $\pred{Ancestor}$ relation along infinite $\pred{Parent}$ chains), and infinite cliquewidth (due to the ternary $\pred{CommonDescendent}$ relation, which cannot be created by a unary coloring scheme).
$\ruleset_\mathrm{bio}$ is also not $\fus$: the UCQ rewriting of the CQ $\pred{Ancestor}(x,y)$ under $\ruleset_\mathrm{bio}$ would have to contain an infinite number of disjuncts: $\exists z.\pred{Ancestor}(x,z) \wedge \pred{Ancestor}(z,y)$ as well as $\exists z,z'.\pred{Ancestor}(x,z) \wedge \pred{Ancestor}(z,z')\wedge \pred{Ancestor}(z',y)$ and so forth.
Yet,
%, toward the end of this paper (\cref{sec:strat}), 
the results obtained in the next sections will enable us to show that $\ruleset_\mathrm{bio}$ is indeed $\fhcs$ and thus allows for decidable expressive querying, as will be discussed in \Cref{ex:Rbio-good}.
\end{example}

It should be noted that, similar to the general case discussed in the first part of the paper, no general complexity bounds exist for any of the width-based existential rules classes: already the special case of plain CQ entailment checking over $\fes$ rulesets is known to exhibit arbitrarily high data complexities \cite{BourgauxCKRT21}. 
Moreover, it is undecidable to determine if a given ruleset is in $\fes$, $\bts$, $\fcs$, or $\fhcs$, as can be shown via slight modifications of a known argument for $\bts$  \cite{BagLecMugSal11} \remarkstart employing a reduction from CQ entailment for unrestricted existential rules\color{black}.

%\clearpage

\section{Extending $\fhcs$ through Stratification}\label{sec:strat}

For analyzing rulesets toward showing closure properties and decidable CQ entailment, the idea of stratification based on rule dependencies has proved helpful. 
We use a notion of rule dependency going back to Baget~\cite{Bag04} and Baget et al.~\cite{BagLecMugSal11}.   

\begin{definition}[rule dependencies, cuts] Given rules $\rho$ and $\rho'$, we say that $\rho'$ \emph{depends} on $\rho$, written $\triggers{\rho}{\rho'}$, \iffi there exists an instance $\instance$ such that 
\begin{itemize}
%    \item $\rho'$ is not applicable to $\instance$ via any homomorphism,
    \item $\rho$ is applicable to $\instance$ via a homomorphism $h$, 
    \item $\rho'$ is applicable to $\chase(\instance,\rho,h)$ via a homomorphism $h'$, and
    \item $\rho'$ is not applicable to $\instance$ via $h'$.
\end{itemize}
If a rule $\rho'$ does not depend on a rule $\rho$, we write $\ntriggers{\rho}{\rho'}$. We define an \emph{$n$-cut} to be a partition $\{\ruleset_{1}, \ldots, \ruleset_{n}\}$ of a ruleset $\ruleset$ such that for every rule $\rho \in \ruleset_{i}$ and $\rho' \in \ruleset_{j}$, $\ntriggers{\rho'}{\rho}$, for $1 \leq i < j \leq n$. We use the notation $\ruleset = \ruleset_{1} \cut \cdots \cut \ruleset_{n}$ to denote that $\{\ruleset_{1}, \ldots, \ruleset_{n}\}$ is an $n$-cut of $\ruleset$.
\remarkstart For classes of rulesets $\mathbf{class}$ and $\mathbf{class}'$  we let $\mathbf{class} \cut \mathbf{class}'$ denote the class containing all rulesets representable as $\ruleset \cut \ruleset'$ where $\ruleset$ is in $\mathbf{class}$ while  $\ruleset'$ is in $\mathbf{class}'$.\color{black}\defend
\end{definition}

Intuitively, $\triggers{\rho}{\rho'}$ means that applying $\rho$ to some instance may create a new trigger for $\rho'$ and thus make $\rho'$ applicable via a new homomorphism.

%
%$\ruleset = \ruleset_{1} \cut \ruleset_{2}$ where $\ruleset_{1}$ is $\bts$ and $\ruleset_{2}$ is $\fus$, 
It was shown that \mbox{$\fes \cut \fes = \fes$}, \mbox{$\fes \cut \bts = \bts$}, \mbox{$\fus \cut \fus = \fus$}, and also that $\bts \cut \fus$ exhibits decidable CQ entailment \cite{BagLecMugSal11}. % 
%then CQ entailment is decidable with $\ruleset$ \cite[\thm~19]{BagLecMugSal11}. 
%We will later generalize such results utilizing our new notion of \emph{finite \hypercliquewidth{} sets}; 
%
The following easy lemma comes in handy when working with cuts (see the appendix for the proof). %; its proof can be found in \app~\ref{app:section-1-proofs}.

\begin{lemma}\label{lem:n-cut-layered-chase}
For any ruleset $\ruleset = \ruleset_{1} \cut \cdots \cut \ruleset_{n}$ and any database $\db$, the following holds:
$$
\sa(\db,\ruleset) = \sa(\ldots\sa(\sa(\db,\ruleset_{1}),\ruleset_{2}) \ldots, \ruleset_{n}).
$$
\end{lemma}

Toward a more principled way to establish $\fhcs$ness for complex rulesets, we now define categories of rulesets, whose chase causes the instance modifications described in  \cref{sec:preservation}. 
Then, in view of \Cref{prop:unary-replacements} and \Cref{thm:closureproperties}, it is immediate that the $\fhcs$ class can integrate ``layers'' of these introduced types of rulesets.  

\begin{definition}
A ruleset $\ruleset$ is called a
\begin{itemize}
\item 
\emph{unary replacement set} ($\urs$) iff  for all instances $\inst$, $\sa(\instance,\ruleset)$ can be obtained from $\instance$ via unary replacements; 
\item 
\emph{2RPE-definition set} ($\reds$) iff  for all instances $\inst$, $\sa(\instance,\ruleset)$ can be obtained from $\instance$ via 2RPE-extensions;
\item 
\emph{FO-definition set} ($\fods$) iff  for all instances $\inst$, $\sa(\instance,\ruleset)$ can be obtained from $\instance$ via FO-extensions. 
\defend
\end{itemize}
\end{definition}

\begin{corollary}\label{prop:nicerelations}
The following relationships hold:
%$\fes \cut \fhcs = \fhcs$; 
$\fhcs \cut \urs = \fhcs$;
$\fhcs \cut \reds = \fhcs$;
$\fhcs \cut \fods = \fhcs$. 
\end{corollary}

Any set of full rules where all rule heads contain only unary predicates is  $\urs$.
The $\reds$ class includes transitivity rules of the form $\ppred(x,y) \wedge \ppred(y,z) \to \ppred(x,z)$, which have been known as notoriously difficult to include into decidable rule classes \cite{BagetBMR15,GanzingerMV99,KieronskiR21,SzwastT04}.
More complex cases of $\reds$ can be derived from regular grammars and have in fact been described and used under the name \emph{(regular) RBoxes} in expressive description logics \cite{SROIQ}.
An easy example for $\fods$ are concept products \cite{BourhisMP17,RudolphKH08} like $\pred{A}(x) \wedge \pred{B}(y) \to \ppred(x,y)$, but $\fods$ goes far beyond that: most notably, it contains any full ruleset that is $\fus$.\footnote{In the literature, the class of full $\fus$ rulesets is also referred to as \emph{bounded datalog}, but we will refrain from that naming to avoid confusion with datalog queries, which are second-order sentences.}

\begin{example}\label{ex:Rbio-good}
At this point, let us return to our example ruleset $\ruleset_\mathrm{bio}$ from \Cref{ex:Rbio}: Rules (1)–(4) are a union of frontier-one (1) and full $\fus$ (2)-(4) and thus $\hcs$, by the forthcoming \cref{prop:frone-fusddl}. The two rules under (5) constitute a $\reds$ ruleset stratified on top of (1)-(4). Rule (6) adds another stratum which is $\fods$, while rules (7) and (8) add a $\urs$ stratum. Thus, we obtain a $\hcs \cut \reds \cut \fods \cut \urs$ ruleset which is $\hcs$ by \cref{prop:nicerelations} and thus allows for decidable HUSOMSOQ entailment.  
\end{example}

\section{The Case of First-Order Rewritability}\label{sec:FO-rewritability}

In this section, we will observe that the prospects of accommodating $\fus$ into the universal-model-based finite-width decidability framework are limited. In fact, already a simple concrete fragment of $\fus$ exhibits undecidable HUSOMSOQ entailment. Within these factual boundaries, there is still a lot to be achieved depending on one's preferences:
\begin{itemize}
\item  
As just discussed, many existential-quantifier-free rulesets -- and in particular \emph{any} that is $\fus$ and full -- can be stratified on top of $\fhcs$ rulesets
in a way that preserves $\fhcs$ness (and thus decidability of HUSOMSOQ entailment). Sometimes even a tighter integration is possible.
\item  
If one is willing to settle for decidability of CQ rather than full-on HUSOMSOQ entailment, one might resort to the very comprehensive class $\fhcs\cut\fus$ that properly subsumes both $\fhcs$ and $\fus$.
\item
On a more speculative note, there is justified hope that all of $\fus$ can be captured by a width-based machinery that does not rely on universal models. For some popular (and otherwise defiant) concrete subclass of $\fus$ we can corroborate this hope. 
\end{itemize}

\subsection{The Price of Expressive Querying}\label{sec:prize}
%\medskip \noindent {\bf The Price of Expressive Querying. } 
%
In prior work~\cite{ICDT2023}, we identified a $\fus$ ruleset $\ruleset_\mathrm{grid}$  whose chase creates an infinite grid, which immediately implies undecidability of entailment of monadic datalog queries (which
are subsumed by HUSOMSOQs). By contraposition, $\ruleset_\mathrm{grid}$ cannot be $\fhcs$, nor can it be a finite-$\fcn{w}$ set for \emph{any} MSO-friendly width measure $\fcn{w}$.
As $\ruleset_\mathrm{grid}$ was constructed ad hoc and is not contained in any of the known ``concrete'' $\fus$ classes, we were hopeful that good width properties could be established for most of these.
However, we found that the simple ruleset $\ruleset'_{\smash{\mathrm{grid}}}$ consisting of the rules
$$
\pred{A}(x) \to \exists y. \pred{E}(x,\!y)
\qquad\qquad
\pred{E}(x,\!y) \to \pred{A}(y)
\qquad\qquad
\pred{A}(x) \wedge \pred{A}(y) \to \exists z. \pred{R}(x,\!y,\!z)
$$%
falls in the rather basic $\fus$ class of \emph{sticky rules} \cite{CaliGP12}, while it still shares the mentioned unpleasant properties with $\ruleset_\mathrm{grid}$. That is, it also creates a grid-like structure albeit in a slightly more covert way. Note that $\rpred$ can be conceived as a ``pairing function'': observing that  $\sa(\{\pred{A}(0)\},\ruleset'_{\smash{\mathrm{grid}}})$ is isomorphic to 
$$\inst'_\mathrm{grid} = \{ \pred{A}(n), \pred{E}(n, n{+}1), \rpred(n,m,\langle n,\!m \rangle) \mid n,m \in \mathbb{N} \},$$%
$\rpred$'s third entries represent grid points; horizontal neighbours are accessible via
{$\varphi_\pred{H}(z,\!z')=\rpred(x,\!y,\!z) {\,\wedge\,} \pred{E}(x,\!x') {\,\wedge\,} \rpred(x'\!,\!y,\!z')$}, vertical ones via
$\varphi_\pred{V}(z,\!z')=\rpred(x,\!y,\!z) {\,\wedge\,} \pred{E}(y,\!y') {\,\wedge\,} \rpred(x,\!y'\!,\!z')$. We conclude that entailment of monadic datalog queries over sticky rules is undecidable and that sticky rules cannot be captured by any notion of finite-$\fcn{w}$ set for any MSO-friendly $\fcn{w}$.   

\subsection{A Novel $\fps$ Class}
%\medskip \noindent {\bf A Concrete $\fps$ Class. } 
%
%It is beyond the scope of this paper to fully explore our framework's various capabilities to define novel concrete classes of rules that are $\fhcs$ and hence ensure decidable DiDaMSOQ entailment. We will just provide one showcase here. 
%
We already saw that, as opposed to $\fus$ in general, the existential-quantifier-free fragment of $\fus$ shows better compatibility with $\fhcs$. Even beyond stratification, it allows us to define a novel $\fhcs$ fragment, for which decidability of querying does not follow from any of the previously established approaches (not even when restricting to CQs). We recall that a rule $\rho$ is \emph{frontier-one} if $|\fr(\rho)|=1$. 

\begin{proposition}\label{prop:frone-fusddl}
	Let $\ruleset = \ruleset_{\smash{\mathrm{fr1}}} \cup \ruleset_{\smash{\mathrm{ff}}}$ be a ruleset where $\ruleset_{\smash{\mathrm{fr1}}}$ is frontier-one and $\ruleset_{\smash{\mathrm{ff}}}$ is full and $\fus$. Then $\ruleset$ is $\fhcs$ and thus allows for decidable HUSOMSOQ entailment.  
\end{proposition}
The detailed proof can be found in the appendix. Note that the interaction between $\ruleset_{\smash{\mathrm{fr1}}}$ and $\ruleset_{\smash{\mathrm{ff}}}$ is not restricted in any way; no stratification is required in this case.
%The result is established by exploiting that $\ruleset_{\smash{\mathrm{fr1}}}$ can be rewritten under $\ruleset_{\smash{\mathrm{bdl}}}$ into a frontier-one ruleset $\ruleset'_{\smash{\mathrm{fr1}}}$ such that $\ruleset' = \ruleset'_{\smash{\mathrm{fr1}}} \cup \ruleset_{\smash{\mathrm{bdl}}} = \ruleset'_{\smash{\mathrm{fr1}}} \cut \ruleset_{\smash{\mathrm{bdl}}}$ has the same universal models as $\ruleset$. Then, the claim follows from the last statement of \cref{prop:nicerelations}.
%
\Cref{prop:frone-fusddl} ceases to hold, however, when considering arbitrary unions of a \emph{guarded} and a full $\fus$ rulesets: note that replacing the third rule in $\ruleset'_{\smash{\mathrm{grid}}}$ 
%from \cref{sec:prize} 
by the two rules 
$$ 
\pred{A}(x) {\,\wedge\,} \pred{A}(y) {\,\to\,} \pred{P}(x,\!y)
\qquad\qquad\qquad 
\pred{P}(x,\!y) {\,\to\,} \exists z. \pred{R}(x,\!y,\!z)
$$%
yields a ruleset of the latter type for which HUSOMSOQ entailment is undecidable for the same reasons as for $\ruleset'_{\smash{\mathrm{grid}}}$.

\subsection{A Broad Class for Decidable CQ Entailment}
%\medskip \noindent {\bf A Broad Class for Decidable CQ Entailment. } 
%
We close this section by noting that the rule class $\fhcs \cut \fus$ exhibits decidable CQ entailment.

\begin{proposition}\label{thm:fcs-fus-decidable}
	Let $\ruleset = \ruleset_\mathrm{fps} \cut \ruleset_\mathrm{fus}$ where $\ruleset_\mathrm{fps}$ is $\fhcs$ and $\ruleset_\mathrm{fus}$ is $\fus$. Then the query entailment problem $(\db,\ruleset) \models \query$ for databases~$\db$, and (Boolean) CQs $\query$ is decidable.
\end{proposition}

\begin{proof} The claim follows directly from the insight that the stratification allows us to reduce  $(\db,\ruleset_\mathrm{fps} \cut \ruleset_\mathrm{fus}) \models \query$ to
$(\db,\ruleset_\mathrm{fps}) \models \rewrs{\ruleset_\mathrm{fus}}{\query}$.
\end{proof}

Observe that $\fhcs \cut \fus$ fully subsumes both $\fhcs$ and $\fus$ without arity restrictions and hence contains the vast majority of the known rules classes for which decidability of CQ entailment has been hitherto established. Unlike for ``pure'' $\fhcs$, however, it remains unclear if this ``hybrid'' class corresponds to a structural property of the underlying universal models.

\subsection{Beyond Universal Models}
%\medskip \noindent {\bf Beyond Universal Models. } 
%
As just discussed, the $\fus$ class generally does not exhibit universal models of finite width. This, however, does not rule out the possibility that $\fus$ is $\fcn{w}$-controllable for the class of (U)CQs for some appropriate $\fcn{w}$. In fact, it is a long-standing well-known conjecture that all $\fus$ rulesets are finitely controllable \cite{GogaczM13a,GogaczM13}. Should that turn out to be true, $\fus$ would be $\fcn{expansion}$-controllable for UCQs and be captured by our decidability framework. What gives some hope in this respect is the fact that finite controllability was already shown to hold for the prominent concrete $\fus$ subclass of sticky rules \cite{GogaczM13a,GogaczM17}. %whence at least this notorious concrete $\fus$ class does fall under our general framework after all.

%\section{The matter of perspective}
%  \begin{itemize}
%    \item ontology mediated querying
%    \item syntax and semantics
%    \item transformations, and their decidability
%  \end{itemize}

\section{Conclusion}

We are optimistic that the generic methods and results presented in this article can advance the understanding of general model-theoretic foundations of decidability of entailment problems. We deem it quite likely that many more useful width notions can be defined (or discovered in the existing literature), leading to novel avenues for establishing decidability of expressive entailment problems. Thus, we plan to investigate other (existing or newly designed) width notions regarding their friendliness properties. Since less expressive $\mathfrak{L}$ may admit more general $\mathfrak{L}$-friendly $\fcn{w}$, there is a Pareto-optimality boundary to explore. For controllability $\mathrm{Cont}_\fcn{w}(\mathfrak{F},\mathfrak{Q})$, there is a three-way tradeoff between the generality of $\fcn{w}$, the expressivity of $\mathfrak{F}$, and the expressivity of $\mathfrak{Q}$. 

\smallskip

One particularly interesting challenge in that respect would be to establish a very general FO-friendly width notion, ideally one that both 
\begin{itemize}
\item
generalizes all the width notions considered in this article and 
\item
still reflects structural model-theoretic properties that are ``cognitively graspable''.
\end{itemize}
The second -- admittedly rather vague -- condition aims to rule out very technical width notions which can generally be obtained by ``reverse-engineering'' known decision procedures, but which would not advance our model-theoretic understanding of decidability.  

\medskip

We also are confident that our provided ``toolbox'' for establishing a fragment's partition\-width-finiteness in a stratification-based, modular manner will be helpful in obtaining specific decidability results for a wide spectrum of logical formalisms in diverse areas related to computational logic. 
Indeed, several instances of $\textsc{Entailment}$ that are known to be decidable still lack a width-based decidability argument. Notable examples include satisfiability of logics without the finite model property, such as 
GF+TG and TGF+TG (the guarded and triguarded fragment extended by transitive guards \cite{SzwastT04,KieronskiM20}) as well as $\mathcal{C}^2$, but of course also BCQ entailment from $\fus$ knowledge bases. Our goal is to accommodate these and other cases in our framework.

Finally, as we have already demonstrated in this paper, we are confident that the width-based framework can guide the search for novel, very expressive, syntactically defined combinations of specification and query languages with a decidable entailment problem. Of course, for cases thus identified, subsequent investigations would have to determine the precise corresponding complexities.

%For future work, we are going to flesh out concrete, syntactically defined fragments of $\hcs$ to improve our understanding of the new ``decidability territory'' opened up by our results. 
%We will also search for a model-theoretic characterization of $\hcs \cut \fus$.

\section*{Acknowledgements}

This article constitutes a central contribution to the European Research Council (ERC) Consolidator Grant 771779 \emph{A Grand Unified Theory of Decidability in Logic-Based Knowledge Representation} (DeciGUT).
The work presented herein was only possible thanks to the continuous funding of all of the authors over several years.

We are indebted to various colleagues for discussing our work, and for commenting on earlier drafts of this paper. First and foremost, our thanks go to Achim Blumensath, who patiently answered our numerous questions regarding his ingenious work on partitionwidth.
Beyond that, we are grateful for feedback from Bartosz Bednarczyk, Michael Benedikt, David Carral, Bruno Courcelle, Lukas Gerlach, Erich Grädel, Carsten Lutz, Jerzy Marcinkowski, Vincent Peth, Andreas Pieris, Ian Pratt-Hartmann, and certainly many others that we forgot. 

We finally thank the anonymous reviewers of this article and precursor works for their constructive feedback that led to numerous improvements.

\clearpage

\bibliographystyle{alphaurl}
\bibliography{bibliography}
%\clearpage

%Table of abbreviations

\clearpage
%\onecolumn

% Appendix
%\noindent{\huge{\scaleobj{1.1}{\bf Appendices}}}
%\medskip\nobalance
\appendix

\section{Overview of Discussed Logical Fragments}\label{app:logicoverview}
%\vspace{-2ex}
\subsection{Logical Fragments in Standard Syntax}\label{app:standardFO}
\newcommand{\sigP}{\pred{P}}
We give a brief description of the logics we discuss in this paper. 
\emph{First-order logic (FO)} formulae are built from atomic formulae (that is, atoms) using 
existential quantification $\exists$, universal quantification $\forall$, negation $\neg$, conjunction $\wedge$, and disjunction $\vee$ in the usual way.
We use the symbol $\rightarrow$ for implication and $\leftrightarrow$ for equivalence as the usual shortcuts for better readability. 
In this article, we discuss the following FO fragments:
\begin{itemize}
	\item
	\emph{Prefix classes}. 
	For these logics, we assume that the sentences are in prenex normal form and that the quantifier prefix is restricted by a regular expression (e.g., $\exists^*\forall^*$ for the Bernays-Schönfinkel class), following 
	standard nomenclature as used by Börger, Grädel and Gurevich \cite{BorgerGG1997}. 
	%E.g., we write \EAAE{} for the logic of all first-order formulae that start with an arbitrarily long (possibly empty) sequence of existential quantifiers, followed by exactly two universal quantifiers, again followed by an arbitrarily long sequence of existential quantifiers in front of a quantifier-free formula. 
	\item 
	\emph{The guarded fragment} (GF). %, \cite{GraedelHirschOtto}). 
	All occurrences of $\forall$ have the form 
	$\forall \bold{x}.(\sigP(\bold{y}) {\,\rightarrow\,} \varphi[\bold{y}])$, 
	%(short: $\forall \sigP(\bold{x}) \impl \varphi[\bold{x}]$)
	and occurrences of $\exists$ the form 
	$\exists \bold{x}.(\sigP(\bold{y}) \wedge \varphi[\bold{y}])$
	for $\bold{x} \subseteq \bold{y}$ and predicates $\sigP{}$.
	%(short: $\exists \sigP(\bold{x}) \wedge \varphi[\bold{x}]$)
	This shape is called \emph{guarded quantification} and $\sigP(\bold{y})$ the respective \emph{guard}.
	\item
	\emph{The triguarded fragment} (TGF). Every quantification over formulae with three or more free variables must be guarded. Equality is disallowed.
	\item
	\emph{The unary negation fragment} (UNF). The symbol $\forall$ is disallowed 
	%$\Rightarrow$, and $\Leftrightarrow$ are disallowed 
	and every occurrence of $\neg$ must be as
	$\neg \varphi$ with $\varphi$ containing at most one free variable.
	\item
	\emph{The guarded negation fragment} (GNF). The symbol $\forall$ is disallowed 
	%$\Rightarrow$, and $\Leftrightarrow$ are disallowed 
	and every occurrence of $\neg$ must be in the form of
	$\sigP(\bold{x}) \wedge \neg \varphi[\bold{x}]$.
	\item
	\emph{The $2$-variable fragment} (FO$^2$). All variables used must be from $\{x,y\}$.
	\item
	\emph{The $2$-variabe fragment of FO with counting quantifiers} ($\mathcal{C}^2$): Same as FO$^2$ but we allow quantifiers $\exists^{\leq n}$, $\exists^{= n}$, $\exists^{\geq n}$ for $n\in \mathbb{N}$, specifying the number of satisfying instantiations of the quantified variable. When measuring formula length, we assume binary encoding of $n$. $\mathcal{C}^2$ can be expressed in FO (however at the cost of using more variables and an exponential blowup in formula length).
\end{itemize}

\noindent All the considered logics are fragments of \emph{second-order logic} (SO), which extends first-order logic by \emph{second-order quantifiers}, of the form $\exists \sigP$ or $\forall \sigP$ for a predicate symbol $\sigP$. 
%A formula's \emph{quantifier rank} is the maximal nesting depth of its quantifiers.  
Of particular interest will be the following SO fragments:
\begin{itemize}
	\item 
	\emph{Guarded second-order logic} (GSO). A second-order quantifier block followed by a GF formula \cite{GradelHO00}. An alternative but equally expressive way of defining GSO is as arbitrary SO sentences with modified semantics: the second-order quantification only ranges over ``guarded relations'', that is, relations containing only tuples all of whose elements co-occur in one tuple of a predicate extension. Note that as a consequence, GSO subsumes FO.
	\item 
    \emph{Monadic second-order logic} (MSO). Second-order quantification is only allowed over unary predicate symbols. In terms of expressivity, MSO subsumes FO and is subsumed by GSO.
	\item 
	\emph{Universal second-order logic} ($\forall$SO). The syntactic fragment of SO where all second-order quantifiers are universal and prenex.
	\item 
    \emph{Disjunctive datalog} (DDlog) \cite{EiterGM97}. 
    A disjunctive datalog sentence over a signature $\Sigma$ is a $\forall$SO sentence of the shape 
    $\forall \Sigma'.\neg\bigwedge_i \forall \vec{x}_i. \big(\varphi_i \to \psi_i \big)$, %$\forall \Sigma'.\neg\bigwedge_i \big(\forall \vec{x}_i \varphi_i \to \psi_i \big),$
    where $\Sigma'$ consists of fresh predicates, every $\varphi_i$ is a (possibly empty) conjunction of atoms over $\Sigma \cup \Sigma'$, every $\psi_i$ is a (possibly empty) disjunction of atoms over $\Sigma'$, and every $\vec{x}_i$ comprises all the variables occurring in $\varphi_i \to \psi_i$.
	\item 
    \emph{Frontier-guarded disjunctive datalog} (FGDDlog). 
    A disjunctive datalog (DDL) sentence where in every subformula $\forall \vec{x}_i. \big(\varphi_i \to \psi_i\big)$, the conjunction $\varphi_i$ contains one atom over $\Sigma$ which jointly carries all the variables occurring in $\psi_i$. Every FGDDlog sentence can be transformed into an equivalent GSO formula \cite{GradelHO00}.
    \item 
    \emph{Monadic disjunctive datalog} (MDDlog). The fragment of DDlog where all predicates in $\Sigma'$ are unary. This makes it a syntactic fragment of MSO. 
    \item 
    \emph{Datalog} (Dlog). The fragment of DDlog where every $\psi_i$ consists of at most one disjunct.
	\item 
    \emph{Frontier-guarded datalog} (FGDlog). The syntactic intersection of FGDDlog and Dlog. 
    \item 
    \emph{Monadic datalog} (MDlog). The syntactic intersection of MDDlog and Dlog.

\end{itemize}

\subsection{Description Logics}\label{app:DLintro}
%We briefly introduce syntax and semantics of $\SROIQbs$.
\newcommand{\Preds}{\Sigma}
\newcommand{\Consts}{\mathbf{C}}
\newcommand{\ebnfeq}{::=}
\newcommand{\conc}[1]{#1}
\newcommand{\rol}[1]{#1}
\newcommand{\ind}[1]{\mathtt{#1}}
\newcommand{\inda}{\ind{a}}
\newcommand{\indb}{\ind{b}}
\newcommand{\indc}{\ind{c}}
\newcommand{\indd}{\ind{d}}
\newcommand{\rolexpR}{\rol{R}}
\newcommand{\rolexpS}{\rol{S}}
\newcommand{\rolR}{\rol{R}}
\newcommand{\rolnR}{\mathtt{R}}
\newcommand{\rolS}{\rol{S}}
\newcommand{\rolnS}{\mathtt{S}}
\newcommand{\rolT}{\mathtt{t}}
\newcommand{\rolU}{\mathtt{U}}
\newcommand{\rolV}{\rol{V}}
\newcommand{\rolW}{\rol{W}}
\newcommand{\conA}{\mathtt{A}}
\newcommand{\connA}{\mathtt{A}}
\newcommand{\conB}{\mathtt{B}}
\newcommand{\conC}{\conc{C}}
\newcommand{\conD}{\conc{D}}
\newcommand{\conE}{\conc{E}}
\newcommand{\conF}{\conc{F}}
\newcommand{\conS}{\mathtt{M}}
\newcommand{\atleast}[1]{\mathord{\geqslant}#1\,}
\newcommand{\atmost}[1]{\mathord{\leqslant}#1\,}
\newcommand{\Self}{\text{\sf{Self}}}
\newcommand{\ssb}{\sqsubseteq}
\newcommand{\SROIQ}{\mbox{$\mathcal{SROIQ}b_s$}}
\newcommand{\define}[1]{\emph{#1}}
\newcommand{\connames}{\Preds_1}
\newcommand{\rolnames}{\Preds_2}
\newcommand{\indnames}{\mathbf{C}}
\newcommand{\Inv}{\fcn{inv}}
\newcommand{\Dis}{\mathrm{Dis}}
\newcommand{\itemb}{\item}
\newcommand{\Vars}{\mathbf{V}}
\newcommand{\llang}[1]{\mathscr{#1}}

Description logics (DLs) are a very prominent family of formalisms used in Knowledge Representation. While using a different syntax than FOL, all DLs discussed in this article can be understood as FOL fragments via a syntactic translation into standard FOL syntax.\footnote{Alternatively, DLs can be seen as logics in their own right, with a special model-theoretic semantics, or as extended multi-modal logics. All these formulations are equivalent and for our purpose it is most convienient to view DLs as FOL fragments employing a non-standard syntax.} For the sake of conciseness, we will omit certain ``syntactic extras'' from the presentation in case they are not relevant for the logical expressivity of the DLs. For a full treatise on DLs, we refer the interested reader to the standard literature \cite{baader_horrocks_lutz_sattler_2017,Rudolph11}.

Like in other FOL fragments, DLs are built on sets of constants and of predicates from some signature. Constants from $\Consts$ are referred to as \emph{individual names}, predicates from $\Preds$ are subdivided into unary predicates ($\Preds_1$, referred to as \emph{concept names}) and binary predicates ($\Preds_2$, referred to as \emph{role names}), whereas predicates of higher arity are disallowed. $\Preds_2$ is subdivided into \emph{simple role names} $\Preds^\mathrm{s}_2$ and \emph{non-simple role names} $\Preds^\mathrm{ns}_2$, the latter containing as a distinguished predicate the \emph{universal role} $\rolU$ and being strictly ordered by some strict order $\prec$.\footnote{In the original definition of $\mathcal{SROIQ}b_s$, simplicity of roles and $\prec$ are not given a priori, but meant to be implicitly determined by the set of axioms. Our choice to fix them explicitly upfront simplifies the presentation without restricting expressivity.}
Then, the set $ \llang{R}^\mathbf{s}$ of \emph{simple role expressions} is defined by
$$
\rolexpR_1,\rolexpR_2 \ebnfeq \rolnS \mid \rolnS^- \mid \rolexpR_1\cup\rolexpR_2 \mid \rolexpR_1\cap\rolexpR_2 \mid \rolexpR_1\setminus\rolexpR_2,
$$%
%{\small$$\rolexpR_1,\rolexpR_2 \ebnfeq \rolS \mid \rolS^- \mid \rolexpR_1\cup\rolexpR_2 \mid \rolexpR_1\cap\rolexpR_2 \mid \rolexpR_1\setminus\rolexpR_2,$$}
with $\rolnS {\,\in\,} \Preds^\mathrm{s}_2$, while the set of (arbitrary) \emph{role expressions} is $ \llang{R} {\,=\,}  \llang{R}^\mathrm{s} \cup \{\rolnR, \rolnR^- \mid \rolnR \in \Preds^\mathrm{ns}_2\}$. 
%The order $\prec$ is then extended to $\mathcal{R}$ by making all elements of $\mathcal{R}^\mathrm{s}$ $\prec$-minimal. 
For convenience, we introduce the function $\Inv$ that ``inverts'' roles, i.e. we set $\Inv(\rolnR)\coloneqq\rolnR^-$ and $\Inv(\rolnR^-) \coloneqq \rolnR$ in order to simplify notation. In the sequel, we will use the symbols $\rolR, \rolS$, possibly with subscripts, to denote role expressions.

Having these sets at hand, we can now turn to the three building blocks of \SROIQ{} knowledge bases: RBox, TBox and ABox.

\medskip\noindent\textbf{RBox. } 
A \define{role inclusion axiom} (RIA, sometimes also referred to as
\define{role chain axiom}) is a statement of the form $\rolR_1 \circ
\ldots \circ \rolR_n \ssb \rolR$ where $\rolR_1,\ldots,\rolR_n,\rolR$ are roles. As a special case thereof (for $n=1$), we obtain \define{direct role inclusions} $\rolR \ssb \rolS$. A finite set of such RIAs is called a \define{role hierarchy}.
%
%We let $ \llang{R}^\mathrm{n}$ denote the set of all non-simple roles of a role hierarchy and call all the other roles \define{simple} denoted by $ \llang{R}^\mathrm{s}=  \llang{R} \setminus  \llang{R}^\mathrm{n}$.
%
A role hierarchy is \define{regular} if every RIA is of one of the forms
$$\rolR\circ\rolR\ssb \rolR,\quad
\Inv(\rolR) \ssb \rolR, \quad
\rolS_1\circ\ldots\circ\rolS_n\ssb \rolR, \quad
\rolR\circ\rolS_1\circ\ldots\circ\rolS_n\ssb \rolR, \quad
\rolS_1\circ\ldots\circ\rolS_n\circ\rolR\ssb \rolR,$$%
such that $\rolR\in\rolnames$ is a (non-inverse) non-simple role name $\rolnR \in \Preds^\mathrm{ns}_2$, and $\rolS_i\prec\rolR$ for $i=1, \ldots, n$ whenever $\rolS_i$ is non-simple.
%
%A \define{role disjointness statement} is a statement of the form $\Dis(\rolS, \rolS')$, where $\rolS$ and $\rolS'$ are simple roles. 
A \SROIQ{}
\define{RBox} (usually denoted by $\mathbfcal{R}$) is %the union of a finite set of role characteristics together with
a regular role hierarchy.

\medskip\noindent\textbf{TBox. } 
Given a \SROIQ{} RBox $\mathbfcal{R}$, we now inductively define \define{concept expressions} (also simply called \define{concepts}) as follows:
\begin{itemize}
\itemb
the symbols $\top$ and $\bot$ as well as all concept names $\connA\in \connames$ are concept expressions,
\itemb
$\{\mathtt{a_1}, \ldots, \mathtt{a_n}\}$ is a concept expression for any finite set $\{\mathtt{a_1}, \ldots, \mathtt{a_n}\}\subseteq \indnames$ of individual names,
\itemb
if $C$ and $D$ are concept expressions then so are $\neg C$ as well as $C\sqcap D$ and $C\sqcup D$,
\itemb
if $\rolR$ is a role and $C$ is a concept expression, then $\exists \rolR. C$ and $\forall \rolR. C$ are concept expressions,
\itemb
if $\rolR$ is a simple role, $n$ is a natural number and $C$ is a concept expression, then  ${\geqslant}n\rolR. C$ and ${\leqslant}n\rolR. C$ as well as $\exists\rolR.\Self$ are also concept expressions.
\end{itemize}
We will denote the set of all concept expressions thus defined by $ \llang{C}$. Throughout this section, the symbols $\conC$, $\conD$ will be used to denote concept expressions.
A \define{general concept inclusion axiom} (short: GCI) has the form $\conC \ssb \conD$ where $\conC$ and $\conD$ are concepts.
Finally, a \SROIQ{} \define{TBox} (usually denoted by $\mathbfcal{T}$) is a finite set of GCIs.

\medskip\noindent\textbf{ABox. } 
An \define{individual assertion} can have any of the following forms:
$$\conC(\inda), \qquad \rolR(\inda,\indb), \qquad \neg\rolR(\inda,\indb), \qquad \inda\approx \indb, \qquad \inda\not\approx \indb,
$$%
with $\inda,\indb\in\indnames$ individual names, $\conC\in \llang{C}$ a concept expression, and $\rolR\in \llang{R}$ a role.
A \SROIQ{} \define{ABox} (usually denoted by $\mathbfcal{A}$) is a finite set of individual assertions. 
%We call an ABox \define{extensionally reduced} if the only concepts and roles occurring therein are concept names and roles names, respectively.

\medskip

Finally, a \SROIQ{} \define{knowledge base} $\kb$ is the union of an RBox $\mathbfcal{R}$ and a TBox $\mathbfcal{T}$ as well as an ABox $\mathbfcal{A}$ for $\mathbfcal{R}$. The elements of $\kb$ are referred to as \define{axioms}. 
%Given a knowledge base $\kb$ we write $\indnames(\kb)$, $\connames(\kb)$, and $\rolnames(\kb)$ to denote those individual names, concept names, and role names which occur in $\kb$, respectively.

\bigskip

In order to establish the connection between DL syntax and FOL, we now define an easy syntactic translation $ \fcn{fol}$. Every $\mathcal{SROIQ}b_s$ knowledge base $\kb$ thus translates via $ \fcn{fol}$ to a FOL sentence $ \fcn{fol}(\kb)$. We define $\fcn{fol}(\kb)=\bigwedge_{\alpha\in \kb} \fcn{fol}(\alpha),$ i.e., we build the conjunction over the separately translated axioms of the knowledge base. How exactly $ \fcn{fol}(\alpha)$ is defined depends on the type of the axiom $\alpha$:
\[
\begin{array}{r@{\ \ }l@{\ \ }l}
	\fcn{fol}(\rolR_1\circ\ldots\circ\rolR_n\ssb\rolR) & = & \forall x_0\ldots x_n ( \bigwedge_{\scriptscriptstyle 1\leq i\leq n} \fcn{fol}_{ \llang{R}}(\rolR_i,x_{i-1},x_i)) \to  \fcn{fol}_{ \llang{R}}(\rolR,x_0,x_n)\\[1mm]
	%	 \fcn{fol}(\Dis(\rolR,\rolR')) & = & \forall x_0x_1 ( \fcn{fol}_{ \llang{R}}(\rolR,x_0,x_1) \to \neg  \fcn{fol}_{ \llang{R}}(\rolR',x_0,x_1))\\[1mm]
	\fcn{fol}(\conC\ssb\conD) & = & \forall x_0( \fcn{fol}_{ \llang{C}}(\conC,x_{0})\to  \fcn{fol}_{ \llang{C}}(\conD,x_{0}))\\[1mm]
	\fcn{fol}(\conC(\inda)) & = &  \fcn{fol}_{ \llang{C}}(\conC,x_{0})[x_0/\inda]\\[1mm]
	\fcn{fol}(\rolR(\inda,\indb)) & = &  \fcn{fol}_{ \llang{R}}(\rolR,x_{0},x_1)[x_0/\inda][x_1/\indb]\\[1mm]
	\fcn{fol}(\neg\rolR(\inda,\indb)) & = & \neg \fcn{fol}(\rolR(\inda,\indb))\\[1mm]
	\fcn{fol}(\inda\approx \indb) & = & \inda = \indb \\[1mm]
	\fcn{fol}(\inda\not\approx \indb) & = & \neg(\inda = \indb) \\
\end{array}
\]

These definitions make use of auxiliary translation functions $ \fcn{fol}_ \llang{R}:  \llang{R} \times \Vars \times \Vars \to \mathrm{FOL}$ for roles and $ \fcn{fol}_ \llang{C}:  \llang{C} \times \Vars \to \mathrm{FOL}$ for concepts (where we assume $\Vars=\{x_0,x_1,\ldots\}$ for the set of variables):
\[
\begin{array}{rll}
	 \fcn{fol}_{ \llang{R}}(\rolU,x_i,x_j) & = & \mathrm{true} \\[1mm]
	 \fcn{fol}_{ \llang{R}}(\rolnS,x_i,x_j) & = & \rolnS(x_i,x_j) \\[1mm]
	 \fcn{fol}_{ \llang{R}}(\rolnS^-,x_i,x_j) & = & \rolnS(x_j,x_i) \\[1mm]
	 \fcn{fol}_{ \llang{R}}(R_1 \cup R_2,x_i,x_j) & = & \fcn{fol}_{ \llang{R}}(\rolR_1,x_i,x_j) \vee \fcn{fol}_{ \llang{R}}(\rolR_2,x_i,x_j) \\[1mm]
	 \fcn{fol}_{ \llang{R}}(R_1 \cap R_2,x_i,x_j) & = & \fcn{fol}_{ \llang{R}}(\rolR_1,x_i,x_j) \wedge \fcn{fol}_{ \llang{R}}(\rolR_2,x_i,x_j) \\[1mm]
	 \fcn{fol}_{ \llang{R}}(R_1 \setminus R_2,x_i,x_j) & = & \fcn{fol}_{ \llang{R}}(\rolR_1,x_i,x_j) \wedge \neg\fcn{fol}_{ \llang{R}}(\rolR_2,x_i,x_j)
\end{array}
\]
\[
\begin{array}{rll}
	 \fcn{fol}_{ \llang{C}}(\connA,x_i) & = & \connA(x_i)\\[1mm]
	 \fcn{fol}_{ \llang{C}}(\top,x_i) & = & \mathrm{true}\\[1mm]
	 \fcn{fol}_{ \llang{C}}(\bot,x_i) & = & \mathrm{false}\\[1mm]
	 \fcn{fol}_{ \llang{C}}(\{\inda_1,\ldots,\inda_n\},x_i) & = &\bigvee_{\scriptscriptstyle 1\leq j\leq n} x_i=\inda_j\\[1mm]
	 \fcn{fol}_{ \llang{C}}(\neg\conC,x_i) & = & \neg \fcn{fol}_{ \llang{C}}(\conC,x_i)\\[1mm]
	 \fcn{fol}_{ \llang{C}}(\conC\sqcap\conD,x_i) & = &  \fcn{fol}_{ \llang{C}}(\conC,x_i)\wedge  \fcn{fol}_{ \llang{C}}(\conD,x_i)\\[1mm]
	 \fcn{fol}_{ \llang{C}}(\conC\sqcup\conD,x_i) & = &  \fcn{fol}_{ \llang{C}}(\conC,x_i)\vee  \fcn{fol}_{ \llang{C}}(\conD,x_i)\\[1mm]
	 \fcn{fol}_{ \llang{C}}(\exists\rolR.\conC,x_i) & = & \exists x_{i+1}.\big( \fcn{fol}_{ \llang{R}}(\rolR,x_i,x_{i+1})\wedge  \fcn{fol}_{ \llang{C}}(\conC,x_{i+1})\big)\\[1mm]
	 \fcn{fol}_{ \llang{C}}(\forall\rolR.\conC,x_i) & = & \forall x_{i+1}.\big( \fcn{fol}_{ \llang{R}}(\rolR,x_i,x_{i+1}) \to  \fcn{fol}_{ \llang{C}}(\conC,x_{i+1})\big)\\[1mm]
	 \fcn{fol}_{ \llang{C}}(\exists\rolR.\Self,x_i) & = &  \fcn{fol}_{ \llang{R}}(\rolR,x_i,x_i)\\[1mm]
	 \fcn{fol}_{ \llang{C}}(\mathord{\geqslant}n \rolR.\conC,x_i) & = & \exists x_{i+1}\dots x_{i+n}.\big(\bigwedge_{\scriptscriptstyle i+1\leq j < k \leq i+n}(x_j\not=x_k)\\[1mm]
	& & \hspace{20mm}\wedge\bigwedge_{\scriptscriptstyle i+1\leq j\leq i+n}( \fcn{fol}_{ \llang{R}}(\rolR,x_i,x_j)\wedge  \fcn{fol}_{ \llang{C}}(\conC,x_{j})\big)\\[1mm]
	 \fcn{fol}_{ \llang{C}}(\mathord{\leqslant}n \rolR.\conC,x_i) & = & \neg \fcn{fol}_{ \llang{C}}(\mathord{\geqslant}{(n+1)} \rolR.\conC,x_i)\\[1mm]
\end{array}
\]

Obviously, $\fcn{fol}_ \llang{R}$ assigns to a role a FOL formula with (at most) two free variables whereas $ \fcn{fol}_ \llang{C}$ assigns to a concept a FOL formula with (at most) one free variable. 
%Now we are ready to translate all $\mathcal{SROIQ}b_s$ axiom types:

\medskip

After defining syntax and FOL-translation of $\mathcal{SROIQ}b_s$ we now obtain all the DLs discussed in this article as fragments of $\mathcal{SROIQ}b_s$ by disallowing certain expressive features.

\begin{itemize}
\item $\mathcal{SRIQ}b_s$ is obtained from $\mathcal{SROIQ}b_s$ by disallowing concept expressions of the form $\{\mathtt{a_1}, \ldots, \mathtt{a_n}\}$.
\item $\mathcal{SROI}b_s$ is obtained from $\mathcal{SROIQ}b_s$ by disallowing concept expressions of the form
${\geqslant}n\rolR. C$ and ${\leqslant}n\rolR. C$. 
\item $\mathcal{SROQ}b_s$ is obtained from $\mathcal{SROIQ}b_s$ by disallowing inverse roles, i.e. the constructor $(\cdot)^-$ inside role expressions.	
\end{itemize} 
Finally, the description logics $\mathcal{ALCHOIQ}b$, $\mathcal{ALCHIQ}b$, $\mathcal{ALCHOI}b$, and $\mathcal{ALCHOQ}b$ are obtained from $\mathcal{SRIOQ}b_s$, $\mathcal{SRIQ}b_s$, $\mathcal{SROI}b_s$, and $\mathcal{SROQ}b_s$, respectively, by disallowing the usage of non-simple role names and concept expressions of the form $\exists\rolR.\Self$.

A structural analysis unveils that $\mathcal{ALCHOI}b$ is a fragment of GF. It is also well-known that $\mathcal{ALCHOIQ}b$ and all its fragments can be expressed in $\mathcal{C}^2$ \cite{RudolphKH08b}.

\section{Proofs for Section 5}
%\subsection{Proofs for \cref{{sec:homclosedqueries}}}

\begin{lemma}\label{lem:interesting}
Let $\Psi$ be a homomorphism-closed $\forall$SO sentence.
Then for any instance $\instance$ holds: $\instance \models \Psi$ \iffi there is some finite $\instance' \subseteq \instance$ satisfying $\instance' \models \Psi$.
\end{lemma}

\begin{proof}
Let $\Psi = \forall \Sigma'.\Psi^*$. 
Note that for any $\instance' \subseteq \instance$, the identity function is a homomorphism from $\instance'$ to $\instance$.

Then, the ``if'' part follows immediately from the fact that $\Psi$ is homomorphism-closed. For the ``only if'' part, assume $\instance \models \Psi$. Let now $\instance^c$ denote the (very likely infinite) set of ground facts obtained from $\instance$ by uniformly replacing all nulls by fresh constant names. Note that $\instance$ has a homomorphism into every model of the sentence set $\instance^c$. Thus, thanks to $\Psi$ being homomorphism-closed by assumption, we can conclude that $\instance^c$ logically entails $\Psi$.
This means that $\instance^c \cup \{\neg\Psi\}$ is unsatisfiable, therefore also $\instance^c \cup \{\neg\Psi^*\}$ must be an unsatisfiable set of FO sentences. By compactness of FOL, this implies that there must be a finite $\instance^c_\mathrm{fin} \subseteq \instance^c$ for which $\instance^c_\mathrm{fin} \cup \{\neg\Psi^*\}$ is already unsatisfiable, implying unsatisfiability of $\instance^c_\mathrm{fin} \cup \{\neg\Psi\}$. Thus we obtain that  
$\instance^c_\mathrm{fin}$ logically entails $\Psi$. Replacing the introduced constant names back by nulls, we obtain a finite instance $\instance_\mathrm{fin} \subseteq \instance$, as desired.
\end{proof}

%\begin{proof}
%Let $\Psi = \forall \Sigma'.\Psi^*$. 
%Let $\mathfrak{C}$ denote the (countable) set of semantically distinct BCQs $q$ satisfying $q \models \Psi$.

%The ``if'' part follows immediately by construction. For the ``only if'' part, assume $\instance \models \Psi$. Let now $\instance^c$ denote the (very likely infinite) set of ground facts obtained from $\instance$ by uniformly replacing all nulls by fresh constant names. Note that $\instance$ has a homomorphism into that every model of the sentence set $\instance^c$. Thus, thanks to $\Psi$ being homomorphism-closed by assumption, we can conclude that $\instance^c$ logically entails $\Psi$.
%This means that $\instance^c \cup \{\neg\Psi\}$ is unsatisfiable, therefore also $\instance^c \cup \{\neg\Psi^*\}$ must be an unsatisfiable set of FO sentences. By compactness of FOL, this implies that there must be a finite $\instance^c_\mathrm{fin} \subseteq \instance^c$ for which $\instance^c_\mathrm{fin} \cup \{\neg\Psi^*\}$ is already unsatisfiable, implying unsatisfiability of $\instance^c_\mathrm{fin} \cup \{\neg\Psi\}$. Thus we obtain that  
%$\instance^c_\mathrm{fin}$ logically entails $\Psi$. Turning the finite set $\instance^c_\mathrm{fin}$ into a conjunction, replacing the introduced constant names back by variables, and existentially quantifying over them, we obtain a CQ $q$ with $q \models \Psi$ by construction (and therefore $q \in \mathfrak{C}$), but also $\instance \models q$ (due to the fact that $q$ is just a subset of $\instance$ and thus the identity function yields the desired homomorphism).
%\end{proof}

\begin{customthm}{\ref{thm:universal}}
Let $\Phi$ be an FO sentence having some finitely universal model set $\mathscr{I}$ and let $\Psi$ be a homomorphism-closed $\forall$SO sentence. Then $\Phi \not \models \Psi$ \iffi there exists a model $\instance \in \mathscr{I}$ with $\instance \not\models \Psi$.
\end{customthm}

\begin{proof}
The ``if'' part is an easy consequence of the fact the $\instance$ serves as a countermodel.

For the ``only if'' part, assume $\Phi \not \models \Psi$, that is, there is a countermodel $\mathcal{J} \models \Phi$ with $\mathcal{J} \not\models \Psi$. Then, because of $\mathcal{J} \models \Phi$, there must be some $\instance \in \mathscr{I}$ of which every finite subset maps homomorphically into $\mathcal{J}$. We conclude our argument by showing that $\instance \not\models \Psi$. Toward a contradiction, suppose $\instance$ were a model of $\Psi$. Then \Cref{lem:interesting} would imply the existence of some finite $\instance' \subseteq \instance$ with $\instance' \models \Psi$. Yet then, noting that $\instance'$ maps homomorphically to $\mathcal{J}$ and $\Psi$ is homomorphism-closed, we would conclude $\mathcal{J} \models \Psi$ contradicting the initial assumption.
\end{proof}

\begin{customlem}{\ref{lem:interpolation}}
Let $\Phi$ be a FO sentence and $\Psi$ a homomorphism-closed $\forall$SO sentence.
Then $\Phi \models \Psi$ \iffi there exists a union of BCQs $Q = \bigvee_{i=1}^n q_{i}$ satisfying $\Phi \models Q$ and $Q \models \Psi$.
\end{customlem}

\begin{proof}
The ``if'' direction is immediate.
For the ``only if'' direction, let $\mathfrak{C}$ denote the set of BCQs $q$ satisfying $q \models \Psi$.
Now assume $\Phi \models \Psi$. By \Cref{lem:interesting}, this means that every model of $\Phi$
must be a model of some $q \in \mathfrak{C}$. Yet, this means that $\{\Phi\} \cup \{\neg q \mid q \in \mathfrak{C}\}$ is an unsatisfiable set of FO sentences. Again, by compactness of FO, we can conclude that $\{\Phi\} \cup \{\neg q \mid q \in \mathfrak{C}'\}$ must be unsatisfiable for some finite $\mathfrak{C}' \subseteq \mathfrak{C}$. Yet, this means that $\Phi$ logically entails $\bigvee_{q \in \mathfrak{C}'} q$ which is a union of BCQs, so we have identified the claimed $Q$.
\end{proof}

%%NEW APP SECTION
%\section{Section 2 Proofs}\label{app:section-1-proofs}

%%NEW APP SECTION
\section{Detailed Material on Section 6}\label{app:section-3-proofs}

\subsection{MSO-Definition of Well-Decoratedness (Theorem 6.17)}
%\subsection{MSO-Definition of Well-Decoratedness (\cref{thm:cliquewidth-k-MSO-decidable})}

Recall that a decorated tree $\mathcal{T}$ is called a \emph{well-decorated tree} \iffi
\begin{enumerate}
	\item for every null $s\in\{ 0,1 \}^*$, if $\mathcal{T}$ contains a fact of the form $\newconst{c}(s)$ or $\newanon(s)$ (in which case $s$ will be called a \emph{pseudoleaf}), then $\mathcal{T}$ contains no other fact of the form $\pred{Dec}(s)$ with $\pred{Dec}\in\decorum$,
	\item for every pseudoleaf $s$ and $s' \in \{0,1\}^+$, $\pred{Dec}(ss')\not\in \mathcal{T}$ for any $\pred{Dec}\in\decorum$  (i.e., descendants of pseudoleafs are undecorated),  
	\item for every $\const{c}\in\mathrm{Const}$, $\mathcal{T}$ contains at most one fact of the form $\newconst{c}(s)$ (i.e., constants are unique),
	\item for every null $s\in\{ 0,1 \}^*$, if $\pred{Dec}(s),\pred{Dec'}(s) \in \mathcal{T}$ with $\pred{Dec},\pred{Dec'} \in \instructionset$ then $\pred{Dec}$ and $\pred{Dec'}$ are compatible (i.e.,  
	no node is decorated by incompatible instructions).
\end{enumerate}
Here, fixing $\cols, \Sigma$, and $\mathrm{Cnst}$, we will define a MSO-sentence $\Theta$ such that for any $\colconsig$-decorated tree $\mathcal{T}$, $\mathcal{T}$ is well-decorated \iffi $\mathcal{T}\vDash \Theta$.
\newline
Let us denote $\decorum$ with $\Gamma$. We write $\varphi_{\mathrm{ancestor}}(x,y)$ to abbreviate $$\neg(x = y) \wedge \forall X. \big(x\in X \land \forall z, z'. \big((z\in X \land (\pred{Succ}_0(z, z')\lor \pred{Succ}_{1}(z,z')) )\rightarrow z'\in X\big) \rightarrow y \in X\big).$$We first define the following four sentences, where $\const{c}\in\mathrm{Cnst}, \const{d}\in\mathrm{Cnst}{\cup}\{ \ast \}$, and then use these to define $\Theta$ below:
\begin{align*}
\Theta^{\mathrm{(1)}} =\ &  \bigwedge_{\mathclap{\hspace{6ex}\const{d}\in\mathrm{Cnst}\cup\{ \ast \}}} \forall x.\big( \newconst{d}(x)\rightarrow \bigwedge_{\mathclap{\hspace{7ex}\pred{Dec}\in\Gamma{\setminus}\{ \newconst{d} \}}} \lnot\pred{Dec}(x) \big)\\
\Theta^{\mathrm{(2)}} =\ & \bigwedge_{\mathclap{\hspace{6ex}\const{d}\in\mathrm{Cnst}\cup\{ \ast \}}} \Big(\forall xy. \big(\newconst{d}(x) \wedge \varphi_{\mathrm{ancestor}}(x,y) \big) \rightarrow \bigwedge_{\mathclap{\hspace{2ex}\pred{Dec}\in\Gamma}}\lnot\pred{Dec}(y)\Big)\\
\Theta^{\mathrm{(3)}} =\ & \bigwedge_{\mathclap{\hspace{2ex}\const{c}\in\mathrm{Cnst}}} \Big( \forall xy. \big(\newconst{c}(x) \land \newconst{c}(y)\big) \rightarrow x = y \Big)\\
\Theta^{\mathrm{(4)}} =\ & \forall x. \hspace{-16ex}\bigwedge_{\hspace{14ex}{\recol{\kappa}{\lambda}^{i},\recol{\kappa}{\lambda'}^{i}\in\Gamma, \lambda\neq\lambda'}}\hspace{-16ex}\lnot\big(\recol{\kappa}{\lambda}^{i}(x)\land \recol{\kappa}{\lambda'}^{i}(x)\big) \land \hspace{-16ex}\bigwedge_{\hspace{14ex}{\cadd{\kappa}{\mu}{I}{\lambda}, \cadd{\kappa}{\mu}{I}{\lambda'}\in\Gamma, \lambda\neq \lambda'}}\hspace{-16ex}\lnot\big(\cadd{\kappa}{\mu}{I}{\lambda}(x) \land \cadd{\kappa}{\mu}{I}{\lambda'}(x)\big)
\end{align*}
Finally, we define 
$$ \Theta = \Theta^{\mathrm{(1)}} \land \Theta^{\mathrm{(2)}} \land \Theta^{\mathrm{(3)}} \land \Theta^{\mathrm{(4)}}.$$

%\subsection{Proofs for \NoCaseChange{\cref{{sec:preservation}}}}
\subsection{A Normal Form for Well-Decorated Trees}
\renewcommand{\disun}{\uplus}

In order to simplify the proofs of \cref{sec:preservation} and \cref{sec:hcw-vs-tw-cw} we now expose a normal form for well-decorated trees.
\begin{definition}\label{def:mso-orig}
	Let $\mathcal{T}$ be a $(\cols,\mathrm{Cnst})$-well-decorated tree.
	For $\const{c}\in \mathrm{Cnst}$, let $\const{c}^\mathcal{T}$ denote the unique $s \in \{0,1\}^*$ with $\pred{c}(s)\in \mathcal{T}$. Moreover, we define the injection $\fcn{orig}: \adom{\inst^\mathcal{T}} \to \{0,1\}^*$ 
	by letting $\fcn{orig}(e)=e^\mathcal{T}$ if $e\in \mathrm{Cnst}$, and $\fcn{orig}(e)=e$ otherwise.\defend
\end{definition}
Let $u_1, \ldots, u_n$ be nodes in a well-decorated tree $\mathcal{T}$. We denote by $\fcn{cca}(u_1, \ldots, u_n)$ the \emph{closest common ancestor of $u_1, \ldots, u_n$}, which is the unique node such that (i) $\fcn{cca}(u_1, \ldots, u_n)$ is an ancestor of $u_1, \ldots, u_n$ and (ii) every other common ancestor of $u_1, \ldots, u_n$ is a strict ancestor of $\fcn{cca}(u_1, \ldots, u_n)$.
%Furthermore, we say that in a well-decorated tree $\mathcal{T}$, $s$ is a \emph{copseudoleaf}, if $s0$ or $s1$ is a pseudoleaf of $\mathcal{T}$.

\begin{definition}
	Let $\Sigma$ be a finite signature, $\mathrm{Const}$ a finite set of constants and $\cols$ a finite set of colors. Let $\mathcal{T}$ be a $\colconsig$-well-decorated tree. We call $\mathcal{T}$ to be in \emph{prenormal form}, if for every fact $\rpred(\ve)\in\inst^{\mathcal{T}}$ the following holds: 
	\begin{enumerate}
		\item If $\ve = (e, \ldots, e)$ for one term $e$ from $\inst^{\mathcal{T}}$, then for $s = \fcn{orig}(e)$ there are $i\in\{ 0,1 \}$, $s'\in\{ 0,1 \}^*$ with $s = s'i$ such that $\unradd{\origcol_{\arity{\rpred}}}{\rpred}^{i}(s')\in\mathcal{T}$.
		\item Otherwise, there are $I\subsetneq \{ 1, \ldots, |\ve| \}$ and $\kappa_1, \kappa_2\in\cols$ with $\arity{\kappa_1} = |I|$, $\arity{\kappa_2} = |\ve|-|I|$, such that
$\radd{\kappa_1}{\kappa_2}{I}{\rpred}(s) \in\mathcal{T}$ for $s = \fcn{cca}(\fcn{orig}(\ve))$, and $\ve \in \col_{s0}^{\mathcal{T}}(\kappa_1) \concat_{I} \col_{s1}^{\mathcal{T}}(\kappa_2)$.
	\end{enumerate}
	A $\colconsig$-well-decorated tree $\mathcal{T}$ in prenormal form is in \emph{normal form} if the following property holds: If $\unradd{\lambda}{\rpred}^{i}(s)\in\mathcal{T}$ for some $\rpred\in\Sigma$, $i\in\{ 0,1 \}$, and $s\in\{ 0,1 \}^*$, then $si$ is a pseudoleaf and $\lambda = \origcol_{\arity{\rpred}}$.\defend
\end{definition}

\begin{lemma}\label{lem:toprenf}
	Let $\Sigma$ be a finite signature, $\mathrm{Const}$ a finite set of constants and $\cols$ a finite set of colors. Let $\mathcal{T}$ be a $\colconsig$-well-decorated tree. Then there is a  $\colconsig$-well-decorated tree $\mathcal{T}'\supseteq \mathcal{T}$ that is in prenormal form and for which $\inst^{\mathcal{T}'} = \inst^{\mathcal{T}}$ holds.
\end{lemma}
\begin{proof} We will show how $\mathcal{T}'$ can be obtained by enriching $\mathcal{T}$ with additional facts.
	Every $\rpred(\ve)\in\inst^{\mathcal{T}}$ falls into one of three cases: Either
	\begin{enumerate}
		\item\label{item:lempwpnf3} $\ve = (e, \ldots, e)$ for some term $e$,
		\item\label{item:lempwpnf1} the above does not hold and there is some $\lambda^*\in\cols$ with $\arity{\lambda^*} = |\ve|$ and $\ve\in\col_{\varepsilon}^{\mathcal{T}}(\lambda^*)$, or
		\item\label{item:lempwpnf2} neither of the two above -- in particular, there is no  $\lambda^*\in\cols$ as described previously. 
	\end{enumerate}
    In Case~\ref{item:lempwpnf3}, let $s_e = \fcn{orig}(e)$, let $s'_e$ be the unique parent node of $s_e$ and let $i_e\in\{ 0,1 \}$ be such that $s_e = s'_ei_e$. Then we add for every fact $\rpred(e, \ldots, e)\in\inst^{\mathcal{T}}$ the atom $\unradd{\origcol_{\arity{\rpred}}}{\rpred}^{i_e}(s'_e)$ to $\mathcal{T}$ to obtain $\mathcal{T}_1$.
	\newline
	In Case~\ref{item:lempwpnf1}, $\ve$ must have been colored for the first time at $s_{\ve} = \fcn{cca}(\fcn{orig}(\ve))$. In particular, $\mathcal{T}$ must contain an atom  $\cadd{\kappa_1}{\kappa_2}{I}{\lambda}(s_{\ve})$ such that $\ve\in\col_{s_{\ve}}^{\mathcal{T}}(\lambda)$ for some $\kappa_1, \kappa_2\in\cols$ with $\arity{\kappa_1}+\arity{\kappa_2} = \arity{\rpred}$, $I\subsetneq \{ 1, \ldots, \arity{\rpred}\}$, and $\ve \in \col_{s0}^{\mathcal{T}}(\kappa_1) \concat_{I} \col_{s1}^{\mathcal{T}}(\kappa_2)$.
	Adding for every fact $\rpred(\ve)\in\inst^{\mathcal{T}}$ falling under this case the corresponding fact $\radd{\kappa_1}{\kappa_2}{I}{\rpred}(s_{\ve})$ to $\mathcal{T}_1$ we obtain $\mathcal{T}'$.
	\newline
	In Case~\ref{item:lempwpnf2}, $\radd{\kappa_1}{\kappa_2}{I}{\rpred}(s_{\ve})\in\mathcal{T}$ must already hold for $s_{\ve} = \fcn{cca}(\fcn{orig}(\ve))$, some $\kappa_1, \kappa_2\in\cols$ with $\arity{\kappa_1}+\arity{\kappa_2} = \arity{\rpred}$, $I\subsetneq \{ 1, \ldots, \arity{\rpred}\}$, and $\ve \in \col_{s0}^{\mathcal{T}}(\kappa_1) \concat_{I} \col_{s1}^{\mathcal{T}}(\kappa_2)$. Hence there is nothing to do here.
	\newline
	
	It is straightforward to see that $\mathcal{T}'$ is still a $\colconsig$-well-decorated tree and that it is in prenormal form.
	$\inst^{\mathcal{T}'}=\inst^{\mathcal{T}}$ clearly holds, as throughout the described construction we only introduced relations in $\mathcal{T}'$ closer to the introduction of the involved entities than or at the node as in $\mathcal{T}'$. %Furthermore, we do not change or add any instructions related to color introductions or recolorings.
\end{proof}
\begin{lemma}\label{lem:prenftonf}
	Let $\Sigma$ be a finite signature, $\mathrm{Const}$ a finite set of constants and $\cols$ a finite set of colors. Let $\mathcal{T}$ be a $\colconsig$-well-decorated tree in prenormalform. Then there is a  $\colconsig$-well-decorated tree $\mathcal{T}'\subseteq \mathcal{T}$ that is in normal form. Furthermore, $\inst^{\mathcal{T}'} = \inst^{\mathcal{T}}$.
\end{lemma}
\begin{proof}
	For every $\unradd{\lambda}{\rpred}^{i}(s)\in\mathcal{T}$ do the following:
	if $\lambda\neq\origcol_{\arity{\rpred}}$ or $si$ is no pseudoleaf of $\mathcal{T}$, remove $\unradd{\lambda}{\rpred}^{i}(s)$ from $\mathcal{T}$. After having done this, we obtain $\mathcal{T}'$ which is well-decorated and in normal form.
	$\inst^{\mathcal{T}'} = \inst^{\mathcal{T}}$ holds since by $\mathcal{T}$ being in prenormal form, the removed facts of the type $\unradd{\lambda}{\rpred}^{i}(s)$ are redundant.
\end{proof}
\begin{corollary}\label{cor:tonf}
	Let $\Sigma$ be a finite signature, $\mathrm{Const}$ a finite set of constants and $\cols$ a finite set of colors. Let $\mathcal{T}$ be a $\colconsig$-well-decorated tree. Then there is a  $\colconsig$-well-decorated tree $\mathcal{T}'$ that is in normal form. Furthermore, $\inst^{\mathcal{T}'} = \inst^{\mathcal{T}}$.
\end{corollary}
\begin{proof}
	Apply \cref{lem:toprenf} and \cref{lem:prenftonf} successively.
\end{proof}
\vspace{-0.3ex}

\subsection{Proofs for Section 6.3}
%\subsection{Proofs for \cref{{sec:preservation}}}

%\subsection*{Proof of \cref{prop:unary-replacements}}
\newcommand{\cwdec}{\mathrm{Dec}(\cols, \mathrm{Cnst}, \Sigma)}
\newcommand{\hcwdec}{\mathrm{Dec}(\cols, \mathrm{Cnst})}
\newcommand{\wdt}{\mathcal{T}}

\begin{customobs}{\ref{prop:unary-replacements}}
Finiteness of treewidth, cliquewidth, and \hypercliquewidth{} are all preserved under unary replacements.
\end{customobs}

We reformulate the above result into the following (more explicit) technical lemma, which we prove afterwards.

\begin{lemma}
Let $\fcn{w}$ be any of the following width measures: treewidth, cliquewidth, \ercwidth{}. Let $\signature$ be a signature and $\pred{U} \in \signature$ a unary predicate. Let $\inst$ be an instance over $\signature$ such that $\fcn{w}(\inst)$ is finite, and let $A \subseteq \adom{\inst}$ be an arbitrary set of domain elements. Then $\fcn{w}(\inst^{\pred{U}:=A})$ is finite.
\end{lemma}
\begin{proof}
We consider the case where $\fcn{w}$ is treewidth. Let $T = (V,E)$ be a tree decomposition of $\inst$ and $k$ be the width of $T$. Note that, by virtue of being a tree decomposition of $\inst$, $T$ satisfies the following:
\begin{enumerate}
	\item\label{item:urp-proof-item 1} each node $X \in V$ is a set of terms of $\inst$ with $\bigcup_{X \in V} X = \adom{\inst}$,
	\item\label{item:urp-proof-item 2} for each $\rpred(t_{1}, ...\,,t_{n}) \in \inst$, there is an $X \in V$ with $\{t_{1}, ...\,,t_{n}\} \subseteq X$,
	\item\label{item:urp-proof-item 3} for each term $t$ in $\inst$, the subgraph of $T$ consisting of the nodes $X$ with $t \in X$ is connected.
\end{enumerate}
We now argue that $T$ also is a tree de\-com\-po\-si\-tion of $\inst^{\pred{U}:=A}$, without requiring any modification.
Observe $\adom{\inst} = \adom{\inst^{\pred{U}:=A}}$, thus Item~(\ref{item:urp-proof-item 1}) is satisfied. Note that Item~(\ref{item:urp-proof-item 2}) is also satisfied, as the only atoms that were changed are unary. Finally, Item~(\ref{item:urp-proof-item 3}) is trivially satisfied as well.

\bigskip
For the case of $\fcn{w}$ being cliquewidth, let $\wdt$ be a $\colconsig$-well-decorated tree witnessing that $\inst$ has finite cliquewidth. We define $\wdt^{\pred{U}}$ as a result of the following transformation of $\wdt$.
\begin{itemize}
	\item Let $\fcn{f}$ be the isomorphism from a colored instance represented by $\wdt$ to $\inst$.
	
	\item For every $\kappa \in \cols$ and any node $s$ decorated with $\pred{Add}_{\pred{U}, \kappa}$ substitute the subtree of $\mathcal T$ that is rooted in $s$ with the subtree rooted in $s0$.
	
	\item For every node $s$ of $\wdt$ decorated with $\pred{*}_\kappa$ (for some $\kappa \in \cols$) for which $\pred{U}(\fcn{f}(s)) \in \inst^{\pred{U}:=A}$ holds, replace the decorator of $s$ with $\pred{Add}_{\pred{U}, \kappa}$ and replace the $\pred{Void}$ decorator of $s0$ with $\pred{*}_\kappa$.
	
	\item For every node $s$ of $\wdt$ decorated with $\pred{c}_\kappa$ (for some $\pred{c} \in \mathrm{Cnst}$ and $\kappa \in \cols$) for which $\pred{U}(\pred{c}) \in \inst^{\pred{U}:=A}$ holds, replace the decorator of $s$ with $\pred{Add}_{\pred{U}, \kappa}$ and replace the $\pred{Void}$ decorator of $s0$ with $\pred{c}_\kappa$.
\end{itemize}

Then $\wdt^{\pred{U}}$ represents a colored instance isomorphic to $\inst^{\pred{U}:=A}$.

\bigskip
In case $\fcn{w}$ is \hypercliquewidth{}, let $\wdt$ be a $\colconsig$-well-decorated tree witnessing that $\inst$ has finite \hypercliquewidth{}. By \cref{cor:tonf} we assume w.l.o.g. that $\mathcal{T}$ is in normal form. The well-decorated tree $\wdt^{\pred{U}}$ obtained from the following transformation satisfies $\inst^{\wdt^{\pred{U}}} \cong \inst^{\pred{U}:=A}$. Let $\fcn{iso}\colon \inst^{\mathcal{T}}\to \inst$ be an isomorphism. We define the following set:
\begin{align*}
	N_{\mathrm{add}} &=  \{ (s', i )\in (\{ 0,1 \}^*\times \{ 0,1 \}) \mid e\in A, s = \fcn{orig}(\fcn{iso}^{-1}(e)), i\in\{ 0,1 \}, s = s'i  \}.
\end{align*}
We then let
$$
\mathcal{T}^{\pred{U}} = \mathcal{T}  
\setminus \{\unradd{\kappa}{\pred{U}}^{i}(s) \mid s \in \{ 0,1 \}^*, i\in\{ 0,1 \}, \kappa \in \cols\}
\cup \{ \unradd{\origcol_{1}}{\pred{U}}^{i}(s') \mid (s',i)\in N_{\mathrm{add}} \} .
$$%
By construction and $\mathcal{T}$ being in normal form, we get $\inst^{\wdt^{\pred{U}}} \cong \inst^{\pred{U}:=A}$ witnessed by $\fcn{iso}$.
\end{proof}

\subsection{Proofs for Section 6.4}
%\subsection{Proofs for \cref{sec:hcw-vs-tw-cw}}

\begin{customthm}{\ref{thm:tree-width-implies-clique-width}}
	Let $\inst$ be a countable instance over a finite signature of arbitrary arity. Then the following hold:
	\begin{itemize}
		\item If $\inst$ has finite treewidth, then $\inst$ has finite \ercwidth{}.
		\item If $\inst$ has finite cliquewidth, then $\inst$ has finite \ercwidth{}.
	\end{itemize}
\end{customthm}

We split the proof of the theorem into two propositions.

\begin{proposition}\label{prop:tw-implies-pw}
	Let $\inst$ be a countable instance over a signature $\Sigma$. If $\inst$ has finite treewidth, then $\inst$ has finite \ercwidth{}.
\end{proposition}

\begin{proof}
	Courcelle showed~\cite[Theorem 8.1]{Cou89} that, given an MSO-sentence $\Psi$ talking about hypergraphs of treewidth $n$, the set of terms representing hypergraphs that satsify $\Psi$, is MSO-definable. In the course of this proof, he exhibits an MSO-definition that can be used in the follwing way for our proof: Given an instance $\inst$ of finite treewidth, first rewrite it as a hypergraph $\mathcal{G}_{\inst}$. Then there is a (term-)tree $\mathcal{T}$ representing $\mathcal{G}_{\inst}$, meaning that the (possibly infinite) term evaluates to $\mathcal{G}_{\inst}$, $\fcn{val}(\mathcal{T}) = \mathcal{G}_{\inst}$. We can then produce an MSO-definition of $\fcn{val}(\mathcal{T}) $ in $\mathcal{T}$. Since $\mathcal{T}$ is a tree, we obtain that $\mathcal{G}_{\inst}$ is MSO-tree-definable, hence MSO-tree-interpretable. By invoking \cref{thm:pw-coincides-mso-interpretability-trees} and translating our previous result to $\inst$, we obtain that $\inst$ has finite \partitionwidth{}, concluding the proof.
\end{proof}

In order to disambiguate between well-decorated trees in the \partitionwidth{} sense (as in \cref{def:well-dec}) and those in the cliquewidth sense~\cite[Definition 6]{ICDT2023}, let for a set of unary colors $\cols$, $\mathrm{Dec}_{\mathrm{cw}}\colconsig$ denote the set of decorators in the latter sense.

\begin{proposition}\label{prop:cw-implies-pw}
	Let $\inst$ be a countable instance over a signature $\Sigma$.
	If $\inst$ has finite cliquewidth, then $\inst$ has finite \ercwidth{}.
\end{proposition}

\begin{proof}
	Let $\inst$ be an instance of finite cliquewidth over signature $\Sigma$ and set of constants $\mathrm{Const}$. Then there is a finite set of (unary) colors $\cols$ and a $\colconsig$-well-decorated tree $\mathcal{T}$ in the cliquewidth sense~\cite[Definition 6]{ICDT2023} such that $\inst = \inst^{\mathcal{T}}$ in the cliquewidth sense~\cite[Definition 7]{ICDT2023}.
	We recall that for the notion of cliquewidth~\cite[Definition 8]{ICDT2023} an MSO-interpretation of instances in well-decorated trees (in the cliquewidth sense~\cite[Definition 6]{ICDT2023}) has been exhibited.
	Using the results obtained there, let $\varphi_{\top}(x_1) = \bigvee \big(\const{d}_k(x_1) \mid \const{d} \in \mathrm{Const} \cup\{\const{*}\}, k \in \cols\big)$; moreover, let  $ \varphi_{\const{c}}(x_1) = \const{c}_{k}(x_1)$ for all $\const{c}\in\mathrm{Cnst}, k\in\cols$; and let $\varphi_{=}(x_1, x_2) = (x_1 = x_2)$.
	For every $\rpred\in \Sigma$ we define a formula $\varphi_{\rpred}(x_1,\ldots, x_{\arity{\rpred}})$ as follows: Let
    	%$$
    	%\varphi^\mathrm{Dom}(X) = \forall x. \Big( x{\in} X \to \varphi_{\top}(x) %\Big),
    	%$$
    	%Furthermore, let
    	$$
    	\thephi_\fcn{up}(x,y) = \szero(y,x) \vee \sone(y,x)
    	$$
    	and
    	$$
    	\varphi^\mathrm{colIn}_{k}(x,y) = \varphi_{\top}(x) \wedge 
    	\forall_{\ell \in \cols} X_\ell. 
    	\Big(\big(\hspace{-3ex} \bigwedge_{\hspace{3ex}\psi \in M_{xy}}\hspace{-3ex}\psi\, \big) \to X_k(y)\Big),
    	$$%
    	where $X_\ell$ are MSO set variables and $M_{xy}$ contains the following formulae (for all respective elements of $\mathrm{Dec}_{\mathrm{cw}}\colconsig$):
    	\begin{align*}
    		\const{c}_k(x) & \to X_{k}(x)\\ 
    		\forall x'\!y'\!. X_k(x') \wedge \thephi_\fcn{up}(x'\!,y') \wedge \pred{Add}_{\rpred,k',k''}(y') & \to X_{k}(y')\\ 
    		\forall x'\!y'\!. X_k(x') \wedge \thephi_\fcn{up}(x'\!,y') \wedge \pred{Add}_{\rpred,k'}(y') & \to X_{k}(y')\\ 
    		\forall x'\!y'\!. X_k(x') \wedge \thephi_\fcn{up}(x'\!,y') \wedge \pred{Recolor}_{k\to k'}(y') & \to X_{k'}(y')\\ 
    		\forall x'\!y'\!. X_k(x') \wedge \thephi_\fcn{up}(x'\!,y') \wedge \pred{Recolor}_{k'\!\to k''}(y') & \to X_{k}(y')\hspace{1.5ex}\text{for }k{\neq} k'\\ 
    		\forall x'\!y'\!. X_k(x') \wedge \thephi_\fcn{up}(x'\!,y') \wedge \oplus(y') & \to X_{k}(y')
    	\end{align*}
    	For each predicate $\rpred\in \Sigma$ of arity $n$, we define
    	\begin{align*}
    		\varphi_{\rpred}(x_1,\ldots,x_n) & = 
    		\exists z.\hspace{-4ex}
    		\bigvee_{\hspace{2ex}
    			\scaleobj{0.9}{(k_1,\ldots,k_n) \in \cols^n}
    		}\hspace{-4ex} 
    		\scaleobj{0.8}{\Big(}
    		\,\pred{Add}_{\rpred,k_1,\ldots,k_n}(z) \wedge\! \scaleobj{0.8}{\bigwedge_{\scaleobj{1.25}{1 \leq i \leq n}}} \varphi^\mathrm{colIn}_{k_i}(x_i,z) 
    		\scaleobj{0.8}{\Big)}
    	\end{align*}
    	
    	%\vspace{-3ex}
    	%$\left.\right.$\defend
    	Then $\mathfrak{I}$ is the sequence of the MSO-formulae: $\varphi_{\rpred}(x_1, \ldots, x_{\arity{\rpred}})$ for every $\rpred\in\Sigma{\cup}\{ \top \}$; $\varphi_{\const{c}}(x_1)$ for every $\const{c}\in\mathrm{Const}$; and $\varphi_{=}(x_1, x_2)$ as described above.
    	Hence $\inst \leq_{\text{MSO}} \mathcal{T}$. Invoking \Cref{thm:pw-coincides-mso-interpretability-trees} we obtain that $\inst$ has finite \partitionwidth{}.
\end{proof}

\begin{customprop}{\ref{th:hcw=cw-bin}}
	Let $\inst$ be a countable instance over a binary signature $\Sigma$.
	$\inst$ has finite cliquewidth \iffi $\inst$ has finite \ercwidth{}.
\end{customprop}

%%%%%%%%%%%%%%%%%%%%%%%
%% IN CASE WE NEED COLOR TRACES
%%%%%%%%%%%%%%%%%%%%%%%
Before proving \cref{th:hcw=cw-bin}, we introduce some technical notions.

\begin{definition}\label{def:recolcycl}
	Let $\mathcal{T}$ be a $\colconsig$-well-decorated tree and $i\in\{ 0,1 \}$. Let $s$ be a node of $\mathcal{T}$ that is an ancestor of a pseudoleaf. We call $s$ \emph{recolor-cyclical for $i$}, if there is a $n\ge 2$ and a sequence of colors $\kappa_1, \ldots, \kappa_n\in\cols$ of the same arity, such that $\recol{\kappa_1}{\kappa_2}^{i}(s), \ldots, \recol{\kappa_{n-1}}{\kappa_n}^{i}(s), \recol{\kappa_n}{\kappa_1}^{i}(s)\in\mathcal{T}$. In this case we call $(i, \kappa_1, \ldots, \kappa_n)$ a \emph{recolor cycle at $s$ for $i$}. We call $s$ \emph{recolor-cyclical} if it is recolor-cyclical for some $i\in\{ 0,1 \}$. If $s$ is not recolor-cyclical, we also call $s$ \emph{non-recolor-cyclical}.
	We say an instruction $\recol{\kappa}{\lambda}^{i}$ \emph{participates in a recolor-cycle at $s$}, if there is some recolor-cycle $(i, \kappa_1, \ldots, \kappa_m)$ at $s$ and either some $1\le j \le m-1$ such that $\kappa = \kappa_j$ and $\lambda = \kappa_{j+1}$ or $\kappa = \kappa_m$ and $\lambda = \kappa_1$.\defend
\end{definition}
\allowdisplaybreaks
\begin{definition}\label{def:hcwtocw-bin}
	Let $\Sigma$ be a binary signature. Given a $\colconsig$-well-decorated tree $\mathcal{T}$ (in the \partitionwidth{}-sense, \cref{def:well-dec}) in normal form, we associate to it a $(\cols_{\mathrm{cw}}, \mathrm{Cnst}, \Sigma)$-well-decorated tree (in the cliquewidth sense~\cite[Definition 6]{ICDT2023}) inductively as detailed in the following.
	
	First define the set of (unary) colors $\cols_{\mathrm{cw}} = \cols_1 \uplus (\cols_1 \times \{ 0,1 \}) \uplus \{ \mu_{\mathrm{cb}}\}$, where $\cols_1\subseteq\cols$ is the subset of all unary colors of $\cols$ and $\mu_{\mathrm{cb}}$ is a still unused color.
	The following auxiliary maps will be defined in an interleaving top-down way along the construction (where $P$ denotes the pseudoleafs of $\mathcal{T}$ and $A_P$ all proper ancestors of pseudoleafs):  $\fcn{r}^{\mathcal{T}}\colon P\disun A_{P} \to \{ 0,1 \}^*$, $\fcn{suc}^{\mathcal{T}}_0\colon A_{P} \to \{ 0,1 \}^*$, $\fcn{suc}^{\mathcal{T}}_1\colon A_{P} \to \{ 0,1 \}^*$, and a map associating to every $s\in P\disun A_P$ an instance $\mathcal{T}_{\mathrm{cw}}^{s}$.
	Let $s\in P\disun A_{P}$ and $\fcn{r}^{\mathcal{T}}, \fcn{suc}^{\mathcal{T}}_0, \fcn{suc}^{\mathcal{T}}_1$ be defined for every proper ancestor $s'$ of $s$ but not for $s$. 
	If $s = \varepsilon$, initialize $\fcn{r}^{\mathcal{T}}\!(\varepsilon) = \varepsilon$. Otherwise define $\fcn{r}^{\mathcal{T}}\!(s) = \fcn{suc}_{i}^{\mathcal{T}}(s')$ where $s'i = s$, i.e. $s'$ is the unique predecessor of $s$.
	We distinguish between two cases:
	\begin{itemize}
		\item \emph{$s$ is a pseudoleaf}. Then set $\mathcal{T}_{\mathrm{cw}}^{s} = \{ \const{d}_{\origcol_1}(\fcn{r}^{\mathcal{T}}\!(s)) \mid \newconst{d}(s)\in\mathcal{T}, \const{d} \in \mathrm{Cnst} \cup \{\ast\} \}$.
		\item \emph{$s$ is not a pseudoleaf}. Then define with $i\in\{ 0,1 \}$ the sets
		\begin{align*}
			&\fcn{PRec}^{i}_{s} = \{ \recol{\kappa}{\lambda}^{i} \mid \kappa\in\cols_1,  \recol{\kappa}{\lambda}^{i}(s)\in\mathcal{T} \},  \hspace{5ex}\fcn{URel}^{i}_{s} = \{ \unradd{\lambda}{\rpred}^{i} \mid \rpred\in\Sigma \}, \text{ and} \\
			&\fcn{Rel}_{s} = \{ \radd{\kappa_1}{\kappa_2}{I}{\rpred} \mid \rpred\in\Sigma, \arity{\rpred} = 2,  \radd{\kappa_1}{\kappa_2}{I}{\rpred}(s)\in\mathcal{T}  \}.
		\end{align*}
		Define $\fcn{C}_{s}^{i}$ to be the set of all recolor cycles at $s$ for $i\in\{ 0, 1\}$. Set $\fcn{Rec}^{i}_{s} = \fcn{PRec}^{i}_s \cup  \fcn{C}_{s}^{i}$.
		Let $k_s = |\fcn{Rel}_s|$, $k_s^{i} = |\fcn{URel}^{i}_{s}|$, $\ell_s^{i} = |\fcn{Rec}_s^{i}|$, and let $\fcn{f}_s\colon \fcn{Rel}_{s} \to \{ 0,\ldots, k_s-1 \}$ and $\fcn{f}_s^{i}\colon \fcn{URel}_{s}^{i} \to \{ 0,\ldots, k_s^{i}-1 \}$ be bijections. Additionally let $\fcn{g}_s^{i}\colon \fcn{Rec}_{s}^{i} \to \{ 0,\ldots, \ell_s^{i}-1 \}$ be bijections with the following properties: For every $\mathrm{cyc}\in\fcn{C}^{i}_s$ where $\mathrm{cyc} = (i, \kappa_1, \ldots, \kappa_m)$:
		\begin{itemize}
			\item $\fcn{g}^{i}_{s}(\mathrm{cyc}) = \fcn{g}^{i}_{s}(\recol{\kappa_{m}}{\kappa_{1}}) + 1$ and $\fcn{g}^{i}_{s}(\recol{\kappa_{m}}{\kappa_{1}}) =  \fcn{g}^{i}_{s}(\recol{\kappa_{m-1}}{\kappa_{m}}) + 1$, 
			\item for every $1 \le k\le m-2$ we have $\fcn{g}^{i}_{s}(\recol{\kappa_{k+1}}{\kappa_{k+2}}) = \fcn{g}^{i}_{s}(\recol{\kappa_{k}}{\kappa_{k+1}}) + 1$.
		\end{itemize}
		Then we define $\fcn{suc}^{\mathcal{T}}_0(s)$ and $\fcn{suc}^{\mathcal{T}}_1(s)$ by letting $\fcn{suc}^{\mathcal{T}}_0(s) = \fcn{r}^{\mathcal{T}}\!(s)0^{2|\cols_1| +k_s+k_s^0+k_s^1+\ell_s^0+\ell_s^1+1 + |\cols_1|}$ and $\fcn{suc}^{\mathcal{T}}_1(s) = \fcn{r}^{\mathcal{T}}\!(s)0^{2|\cols_1| +k_s+k_s^0+k_s^1+\ell_s^0+\ell_s^1}10^{|\cols_1|}$. Then $\mathcal{T}_{\mathrm{cw}}^{s}$ is the union of the following instances over $\mathrm{Dec}_{\mathrm{cw}}(\cols^*, \mathrm{Cnst}, \Sigma)$:\small
		\begin{align*}
			\mathcal{T}_{\mathrm{pre}, 0}^{s} &=  \{ \pred{Recolor}_{\kappa \shortto (\kappa, 0)}(\fcn{r}^{\mathcal{T}}\!(s)0^{2|\cols_1| +k_s+k_s^0+k_s^1+\ell_s^0+\ell_s^1+1 + j}) \hspace{-3ex}&& \mid \kappa\in\cols_1, \fcn{coen}(\kappa) = j \} \\
			\mathcal{T}_{\mathrm{pre}, 1}^{s} &=  \{ \pred{Recolor}_{\kappa \shortto (\kappa, 1)}(\fcn{r}^{\mathcal{T}}\!(s)0^{2|\cols_1| + k_s+k_s^0+k_s^1+\ell_s^0+\ell_s^1}10^{j}) \!\! && \mid \kappa\in\cols_1, \fcn{coen}(\kappa) = j \} \\
			\mathcal{T}_{\mathrm{un}}^{s} & = \{ \oplus(\fcn{r}^{\mathcal{T}}\!(s)0^{2|\cols_1| +k_s+k_s^0+k_s^1+\ell_s^0+\ell_s^1}) \} &&\\
			\mathcal{T}_{\pred{Add}}^{s} &=  \{ \pred{Add}_{\rpred, (\kappa_1, i_1),(\kappa_2, i_2)}(\fcn{r}^{\mathcal{T}}\!(s)0^{2|\cols_1| + k_s^0+k_s^1+\ell_s^0+\ell_s^1+j}) &&\mid \radd{\kappa_1}{\kappa_2}{I}{\rpred}\in\fcn{Rel}_{s},I = \{i_1 +1 \}, \\[-1ex]
			&&& \hspace{3ex} i_1+i_2 = 1,\fcn{f}_s(\radd{\kappa_1}{\kappa_2}{I}{\rpred}) = j   \} \\
			\mathcal{T}_{\pred{UUAdd},0}^{s} &=  \{ \pred{Add}_{\rpred, (\kappa, 0)}(\fcn{r}^{\mathcal{T}}\!(s)0^{2|\cols_1| +k_s^1+\ell_s^0+\ell_s^1+j}) &&\hspace{-7ex}\mid \arity{\rpred} = 1,  \unradd{\kappa}{\rpred}^0\in\fcn{URel}_{s}^0,\fcn{f}_s^{0}(\unradd{\kappa}{\rpred}^{0}) = j   \} \\
\mathcal{T}_{\pred{UBAdd},0}^{s} &=  \{ \pred{Add}_{\rpred, (\kappa, 0), (\kappa, 0)}(\fcn{r}^{\mathcal{T}}\!(s)0^{2|\cols_1| +k_s^1+\ell_s^0+\ell_s^1+j}) \!\!\! &&\hspace{-7ex}\mid \arity{\rpred} = 2,  \unradd{\kappa}{\rpred}^0\in\fcn{URel}_{s}^0,\fcn{f}_s^{0}(\unradd{\kappa}{\rpred}^{0}) = j   \} \\
\mathcal{T}_{\pred{UUAdd},1}^{s} &=  \{ \pred{Add}_{\rpred, (\kappa, 1)}(\fcn{r}^{\mathcal{T}}\!(s)0^{2|\cols_1| +\ell_s^0+\ell_s^1+j}) &&\hspace{-7ex}\mid \arity{\rpred} = 1,  \unradd{\kappa}{\rpred}^1\in\fcn{URel}_{s}^1,\fcn{f}_s^{1}(\unradd{\kappa}{\rpred}^{1}) = j   \} \\
\mathcal{T}_{\pred{UBAdd},1}^{s} &=  \{ \pred{Add}_{\rpred, (\kappa, 1), (\kappa, 1)}(\fcn{r}^{\mathcal{T}}\!(s)0^{2|\cols_1| +\ell_s^0+\ell_s^1+j}) &&\hspace{-7ex}\mid \arity{\rpred} = 2,  \unradd{\kappa}{\rpred}^1\in\fcn{URel}_{s}^1,\fcn{f}_s^{1}(\unradd{\kappa}{\rpred}^{1}) = j   \} \\
		\end{align*}
\pagebreak
		\begin{align*}
			\mathcal{T}_{\pred{Rec}, 0}^{s} &=  \{ \pred{Recolor}_{(\kappa, 0)\shortto(\lambda, 0)}(\fcn{r}^{\mathcal{T}}\!(s)0^{2|\cols_1| + \ell_s^1+ j}) &&\mid \recol{\kappa}{\lambda}^{0}\in\fcn{Rec}_s^{0}, \fcn{g}_s^{0}(\recol{\kappa}{\lambda}^{0}) = j   \\[-1ex] &&& \hspace{3ex}\forall\mathrm{cyc}\in\fcn{Rec}_s^{0}. \mathrm{cyc} = (0, \kappa_1, \ldots, \kappa_m) \rightarrow \kappa \neq \kappa_1 \} \\
			\mathcal{T}_{\pred{Rec}, 1}^{s} &=  \{ \pred{Recolor}_{(\kappa, 1)\shortto(\lambda, 1)}(\fcn{r}^{\mathcal{T}}\!(s)0^{2|\cols_1|  + j}) &&\mid \recol{\kappa}{\lambda}^{1}\in\fcn{Rec}_s^{1},\fcn{g}_s^1(\recol{\kappa}{\lambda}^{1}) = j , \\[-1ex] &&& \hspace{3ex}\forall\mathrm{cyc}\in\fcn{Rec}_s^{1}. \mathrm{cyc} = (1, \kappa_1, \ldots, \kappa_m) \rightarrow \kappa \neq \kappa_1 \} \\
			\mathcal{T}_{\pred{CRecOp}, 0}^{s} &=  \{ \pred{Recolor}_{(\kappa_1, 0)\shortto\mu_{\mathrm{cb}}}(\fcn{r}^{\mathcal{T}}\!(s)0^{2|\cols_1| + \ell_s^1+ j}) &&\mid \mathrm{cyc}\in\fcn{Rec}_s^{0}, \fcn{g}_s^{0}(\mathrm{cyc}) = j, \mathrm{cyc} = (0, \kappa_1, \ldots, \kappa_m)  \} \\
			\mathcal{T}_{\pred{CRecClo}, 0}^{s} &=  \{ \pred{Recolor}_{\mu_{\mathrm{cb}} \shortto (\kappa_2, 0)}(\fcn{r}^{\mathcal{T}}\!(s)0^{2|\cols_1| + \ell_s^1+ j}) &&\mid (0, \kappa_1, \ldots, \kappa_m)\in\fcn{Rec}_s^{0} , \fcn{g}_s^{0}(\recol{\kappa_1}{\kappa_2}^0) = j \} \\
			\mathcal{T}_{\pred{CRecOp}, 1}^{s} &=  \{ \pred{Recolor}_{(\kappa_1, 1)\shortto\mu_{\mathrm{cb}}}(\fcn{r}^{\mathcal{T}}\!(s)0^{2|\cols_1|  + j}) &&\mid  \mathrm{cyc}\in\fcn{Rec}_s^{1}, \fcn{g}_s^{1}( \mathrm{cyc}) = j, \mathrm{cyc} = (1, \kappa_1, \ldots, \kappa_m) \} \\
			\mathcal{T}_{\pred{CRecClo}, 1}^{s} &=  \{ \pred{Recolor}_{\mu_{\mathrm{cb}} \shortto (\kappa_2, 1)}(\fcn{r}^{\mathcal{T}}\!(s)0^{2|\cols_1| + j}) &&\mid (1, \kappa_1, \ldots, \kappa_m)\in\fcn{Rec}_s^{1} , \fcn{g}_s^{1}(\recol{\kappa_1}{\kappa_2}^1) = j \} \\
			\mathcal{T}_{\mathrm{pos}, 0}^{s} &=  \{ \pred{Recolor}_{(\kappa, 0)\shortto \kappa}(\fcn{r}^{\mathcal{T}}\!(s)0^{|\cols_1|+j}) &&\mid \kappa\in\cols_1, \fcn{coen}(\kappa) = j \} \\
			\mathcal{T}_{\mathrm{pos}, 1}^{s} &=  \{ \pred{Recolor}_{(\kappa, 1)\shortto \kappa}(\fcn{r}^{\mathcal{T}}\!(s)0^{j}) &&\mid \kappa\in\cols_1, \fcn{coen}(\kappa) = j \}.
		\end{align*}
	\end{itemize}
	Thus we obtain a map that maps every $s\in P\oplus A_P$ to an instance $\mathcal{T}_{\mathrm{cw}}^s$.
	We define $\mathcal{T}_{\mathrm{cw}}^{\mathrm{aux}} = \bigcup_{s\in P\disun A_P}\mathcal{T}_{\mathrm{cw}}^{s}$. Finally,
	$$\mathcal{T}_{\mathrm{cw}} = \mathcal{T}_{\mathrm{bin}} \cup \mathcal{T}_{\mathrm{cw}}^{\mathrm{aux}} \cup \{ \pred{Void}(s) \mid \forall \pred{Dec} \in\mathrm{Dec}_{\mathrm{cw}}(\cols^*, \mathrm{Cnst}, \Sigma).  \pred{Dec}(s)\notin \mathcal{T}_{\mathrm{cw}}^{\mathrm{aux}} \}.$$%
	Additionally, we define a map $\fcn{h}^{\mathcal{T}}\colon \adom{\inst^{\mathcal{T}}} \to \adom{\inst^{\mathcal{T}_{\mathrm{cw}}}}$ as follows:
	$$\fcn{h}^{\mathcal{T}}\!(e) = \begin{cases} \const{c}, &\text{if } e = \const{c} \text{ for some } \const{c}\in\mathrm{Const}, \\ \fcn{r}^{\mathcal{T}}\!(e), &\text{otherwise} \end{cases}$$%
	where $\inst^{\mathcal{T}_{\mathrm{cw}}}$ is interpreted in the cliquewidth sense~\cite[Definition 7]{ICDT2023}.\defend
\end{definition}
\begin{observation}\label{obs:hcwtocw-bin}
	We observe the following:
	\begin{enumerate}
		\item \label{item:obshcwtocw-bin-1} $\mathcal{T}_{\mathrm{cw}}$ is a $(\cols_{\mathrm{cw}}, \mathrm{Const}, \Sigma)$-well-decorated tree (in the cliquewidth sense~\cite[Definition 6]{ICDT2023}).
		\item \label{item:obshcwtocw-bin-2} By construction, $\fcn{h}^{\mathcal{T}}$ is a bijection between the terms of $ \inst^{\mathcal{T}}$ and $\inst^{\mathcal{T}_{\mathrm{cw}}}$.
	\end{enumerate}
\end{observation}
\begin{lemma}\label{lem:hcwtocw-bin-cols}
	Let $\Sigma$ be a binary signature and $\mathcal{T}$ a $\colconsig$-well-decorated tree. Then for every node $s$ of $\mathcal{T}$ that is a pseudoleaf or an ancestor of a pseudoleaf, every $e\in\ent^{\mathcal{T}}\!(s)$, and every $\kappa\in\cols_1$ the following are equivalent:
	\begin{enumerate}
		\item $e\in\col_s^{\mathcal{T}}(\kappa)$
		\item $\col_{\fcn{r}^{\mathcal{T}}\!(s)}^{\mathcal{T}_{\mathrm{cw}}}(\fcn{h}^{\mathcal{T}}\!(e)) = \kappa$ (in the cliquewidth sense~\cite[Definition 7]{ICDT2023}).
	\end{enumerate}
\end{lemma}
\begin{proof}
	Simple induction over the distance of $s$ from $s_e = \fcn{orig}(e)$.
	\begin{itemize}
		\item $s = s_e$. Then $e\in\col_s^{\mathcal{T}}(\kappa)$ \iffi $\kappa = \origcol_1$. This is equivalent to $\const{c}(s)\in\mathcal{T}$ holding for some $\const{c}\in\mathrm{Cnst}{\cup}\{ \ast \}$. This is equivalent by construction to $\const{c}_{\origcol_1}(\fcn{r}^{\mathcal{T}}\!(s))\in\mathcal{T}_{\mathrm{cw}}$. Which is equivalent to $\col_{\fcn{r}^{\mathcal{T}}\!(s)}^{\mathcal{T}_{\mathrm{cw}}}(\fcn{h}^{\mathcal{T}}\!(e)) = \kappa$. 
		\item $s$ is a proper ancestor of $s_e$ and the assertion holds for every node $s^{\ast}$ being an ancestor of $s_e$ such that $s'$ is a proper ancestor of $s^{\ast}$. Consequently also for the unique $s'$ such that $s'= si $ and $s'$ is an ancestor of $s_e$. 
		Assume $e\in\col_{s}^{\mathcal{T}}(\kappa)$.We have two cases to consider:
		\begin{itemize}
			\item $e\in\col_{si}^{\mathcal{T}}(\kappa)$. Then $\recol{\kappa}{\lambda}^{i}(s)\notin\mathcal{T}$ for all $\lambda\in\cols_1$, $\lambda\neq\kappa$. Hence for all nodes $v$ of $\mathcal{T}_{cw}$ that are proper ancestors of $\fcn{r}^{\mathcal{T}}\!(si)$ and for whom $\fcn{r}^{\mathcal{T}}\!(s)$ is an ancestor, $\pred{Recolor}_{(\kappa, i)\shortto(\lambda, i)}(v)\!\notin\mathcal{T}_{\mathrm{cw}}$ for all $\lambda\neq \kappa$. As $\col_{\fcn{r}^{\mathcal{T}}\!(si)}^{\mathcal{T}_{\mathrm{cw}}}(\fcn{h}^{\mathcal{T}}\!(e)) = \kappa$ by assumption, we conclude $\col_{\fcn{r}^{\mathcal{T}}\!(s)}^{\mathcal{T}_{\mathrm{cw}}}(\fcn{h}^{\mathcal{T}}\!(e)) = \kappa$. The other direction is completely analogous.
			\item $e\in\col_{si}^{\mathcal{T}}(\lambda)$ for $\lambda\in\cols_1$, $\lambda\neq\kappa$. Then $\recol{\lambda}{\kappa}^{i}(s)\in\mathcal{T}$. By construction  there is a node $v$ of $\mathcal{T}_{cw}$ that is a proper ancestors of $\fcn{r}^{\mathcal{T}}\!(si)$ and for whom $\fcn{r}^{\mathcal{T}}\!(s)$ is an ancestor, such that $\pred{Recolor}_{(\lambda, i)\shortto(\kappa, i)}(v)\in\mathcal{T}_{\mathrm{cw}}$. As $\col_{\fcn{r}^{\mathcal{T}}\!(si)}^{\mathcal{T}_{\mathrm{cw}}}(\fcn{h}^{\mathcal{T}}\!(e)) = \lambda$ by assumption, we conclude $\col_{\fcn{r}^{\mathcal{T}}\!(s)}^{\mathcal{T}_{\mathrm{cw}}}(\fcn{h}^{\mathcal{T}}\!(e)) = \kappa$. The other direction is completely analogous. \qedhere
		\end{itemize}
	\end{itemize}
\end{proof}

\begin{customprop}{\ref{th:hcw=cw-bin}}
	Let $\inst$ be a countable instance over a binary signature $\Sigma$.
	$\inst$ has finite cliquewidth \iffi $\inst$ has finite \ercwidth{}.
\end{customprop}
\begin{proof}
	By \cref{prop:cw-implies-pw} it suffices to prove that, given an instance $\inst$ over a binary signature $\Sigma$ with finite \partitionwidth{}, $\inst$ also has finite cliquewidth.
	So let $\Sigma$ be a finite binary signature, $\cols$ finite set of colors (of arity $\le 2$) and $\mathrm{Cnst}$ a finite set of constants. 
	Assume $\mathcal{T}$ to be a $\colconsig$-well-decorated tree (in the \partitionwidth{} sense, \cref{def:well-dec}) in normal form such that $\inst = \inst^{\mathcal{T}}$. This can be achieved by renaming (if necessary) of the nulls in $\inst$. 
	Let $\cols^{\ast} = \cols_1\disun (\cols_1\times \{ 0,1 \})\disun \{ \mu_{\mathrm{cb}}\}$ and $\mathcal{T}_{\mathrm{cw}}$ as in \cref{def:hcwtocw-bin}. If we show $\fcn{h}^{\mathcal{T}}\colon \inst^{\mathcal{T}} \to \inst^{\mathcal{T}_{\mathrm{cw}}}$ to be an isomorphism we are done. Since by \cref{obs:hcwtocw-bin}, \cref{item:obshcwtocw-bin-1} we know $\fcn{h}^{\mathcal{T}}$ to be bijective, we only need to prove it to be structure preserving in both ways. In the following, we use all the defined auxiliary maps, sets and instances as denoted in the construction in \cref{def:hcwtocw-bin}.

	Let $\rpred(\ve)\in\inst$. Note that this automatically implies, that $\varepsilon$ is not a pseudoleaf (in $\mathcal{T}$). We consider two cases:
	\begin{itemize}
		\item $\ve = (e, e)$ or $\ve = e$ for some entity $e$. Then let $s_e = \fcn{orig}(e)$, $s_e'\in\{ 0,1 \}^*$, and $i_e\in\{ 0,1 \}$ such that $s_e = s'_ei_e$. Since $\mathcal{T}$ is in normal form, we conclude that $\unradd{\origcol_{\arity{\rpred}}}{\rpred}^{i_e}(s_e')\in\mathcal{T}$. As $s_e$ is a pseudoleaf by definition, there is a unique $\const{c}\in\mathrm{Cnst}\disun \{ \ast \}$ such that $\const{c}(s)\in\mathcal{T}$. By construction, $\const{c}_{\origcol_1}(\fcn{r}^{\mathcal{T}}\!(s))\in\mathcal{T}_{\mathrm{cw}}$. Additionally, $\pred{Recolor}_{\origcol_{1} \shortto (\origcol_1, i_e)}(v)\in\mathcal{T}^{s_e'}_{\mathrm{cw}}$ for some node $v\in\{ 0,1 \}^{\ast}$ of $\mathcal{T}_{\mathrm{cw}}$ such that 
		\begin{itemize}
			\item $\fcn{r}^{\mathcal{T}}\!(s'_e)$ is a proper ancestor of $v$, and
			\item $v$ is a proper ancestor of $\fcn{r}^{\mathcal{T}}\!(s)$.
		\end{itemize}
		Finally, since $\unradd{\origcol_{\arity{\rpred}}}{\rpred}^{i_e}(s_e')\in\mathcal{T}$ there is another node $v'$ satisfying both 
		\begin{itemize}
			\item $\fcn{r}^{\mathcal{T}}\!(s'_e)$ is a proper ancestor of $v'$, and
			\item $v'$ is a proper ancestor of $v$,
		\end{itemize}
		such that either $\pred{Add}_{\rpred, (\origcol_1, i_e), (\origcol_1, i_e)}(v')\in\mathcal{T}_{\mathrm{cw}}$ or $\pred{Add}_{\rpred, (\origcol_1, i_e)}(v')\in\mathcal{T}_{\mathrm{cw}}$ depending on whether $\arity{\rpred}= 2$ or $\arity{\rpred}= 1$. By construction of $\mathcal{T}_{\mathrm{cw}}$ and $\mathrm{Atoms}_{v'}$~\cite[Definition 7]{ICDT2023}, we obtain $\rpred(\fcn{h}^{\mathcal{T}}\!(\ve))\in\mathrm{Atoms}_{v'}$ and hence $\rpred(\fcn{h}^{\mathcal{T}}\!(\ve))\in\inst^{\mathcal{T}_{\mathrm{cw}}}$.
		\item $\ve \neq (e, \ldots, e)$ for every entity $e$. This implies immediately that $\arity{\rpred} = 2$. Hence $\ve = (e_1, e_2)$ with $e_1\neq e_2$. Furthermore, as $\mathcal{T}$ is in normal form, for $s_{\ve} = \fcn{cca}(\fcn{orig}(\ve))$ there are some $\kappa_1, \kappa_2\in\cols$ with $\arity{\kappa_1} = \arity{\kappa_2} = 1$ and $I = {i^*}$ with $i^*\in\{ 1,2 \}$ such that 
		\begin{itemize}
			\item $\radd{\kappa_1}{\kappa_2}{\{ i^* \}}{\rpred}(s_{\ve})\in\mathcal{T}$,
			\item $e_{i^{*}} \in\col_{s_{\ve}0}^{\mathcal{T}}(\kappa_1)$, and
			\item $e_{j}\in\col_{s_{\ve}1}^{\mathcal{T}}(\kappa_2)$ such that $\{ i^*, j \} = \{ 1,2 \}$.
		\end{itemize}
		By construction, there are nodes $v, v_0, v_1\in\{ 0,1 \}^{\ast}$ from $\mathcal{T}_{\mathrm{cw}}$ such that
		\begin{itemize}
			\item $v$ is a proper ancestor to $v_0$ and $v_1$,
			\item $\fcn{r}^{\mathcal{T}}\!(s)$ is a proper ancestor of $v$,
			\item $v_0$ is a proper ancestor $\fcn{r}^{\mathcal{T}}\!(s0)$, $v_1$ is a proper ancestor of $\fcn{r}^{\mathcal{T}}\!(s1)$,
			\item if $i^* = 1$, $\pred{Add}_{\rpred, (\kappa_{1}, 0), (\kappa_{2}, 1)}(v)\in\mathcal{T}_{\mathrm{cw}}$, if $i^* = 2$, $\pred{Add}_{\rpred, (\kappa_{2}, 1), (\kappa_{1}, 0)}(v)\in\mathcal{T}_{\mathrm{cw}}$, and
			\item $\pred{Recolor}_{\kappa_1\shortto (\kappa_1, 0)}(v_0), \pred{Recolor}_{\kappa_2\shortto (\kappa_2, 1)}(v_1)$.
		\end{itemize}
		Applying \cref{lem:hcwtocw-bin-cols} we obtain $ \col_{\fcn{r}^{\mathcal{T}}\!(s_{\ve}0)}^{\mathcal{T}_{\mathrm{cw}}}(e_{i^{\ast}}) = \kappa_1$ and $\col_{\fcn{r}^{\mathcal{T}}\!(s_{\ve}1)}^{\mathcal{T}_{\mathrm{cw}}}(e_{j}) = \kappa_2$. Hence by construction $ \col_{v0}^{\mathcal{T}_{\mathrm{cw}}}(\fcn{h}^{\mathcal{T}}\!(e_{i^{\ast}})) = (\kappa_1, 0)$ and $\col_{v0}^{\mathcal{T}_{\mathrm{cw}}}(h(e_{j})) = (\kappa_2, 1)$. Consequently $\rpred(\fcn{h}^{\mathcal{T}}\!\ve)\in\mathrm{Atoms}_{v}$ and thus $\rpred(\fcn{h}^{\mathcal{T}}\!(\ve))\in\inst^{\mathcal{T}_{\mathrm{cw}}}$.
	\end{itemize}
	Now assume $\rpred(\fcn{h}^{\mathcal{T}}\!(\ve))\in\inst^{\mathcal{T}_{\mathrm{cw}}}$ for $\ve$ from $\adom{\inst^{\mathcal{T}}}$. Note that this captures all the atoms in $\inst^{\mathcal{T}_{\mathrm{cw}}}$ as $\fcn{h}^{\mathcal{T}}$ is a bijection on the sets of entities corresponding to the instances $\inst = \inst^{\mathcal{T}}$ and $\inst^{\mathcal{T}_{\mathrm{cw}}}$. We again consider two cases:
	\begin{itemize}
		\item $\ve = (e, e)$ or $\ve = e$ for some entity $e$ of $\inst$. Let $s = \fcn{orig}(e)$ be the node of $\mathcal{T}$ introducing entity $e$. By construction there is a node $v$ in $\mathcal{T}_{\mathrm{cw}}$ such that $\fcn{r}^{\mathcal{T}}\!(s) = v$ and there is a $\const{c}\in\mathrm{Cnst} {\disun} \{ \ast \}$ such that $\const{c}(s)\in\mathcal{T}$ and $\const{c}_{\origcol_1}(v)\in\mathcal{T}_{\mathrm{cw}}$. Let $s'$ be the parent of $s$ in $\mathcal{T}$ and $i\in\{ 0,1 \}$ such that $s = s'i$. Again by construction, there are nodes $v_{i}, v'$ of $\mathcal{T}_{\mathrm{cw}}$ such that
		\begin{itemize}
			\item $\fcn{r}^{\mathcal{T}}\!(s')$ is a proper ancestor of $v', v_{i}, v$,
			\item $v'$ is a proper ancestor of $v_{i}$ and $v$,
			\item $v_{i}$ is a proper ancestor of $v$,
			\item $\pred{Recolor}_{\origcol_1 \shortto (\origcol_1, i)}(v_{i})\in\mathcal{T}_{\mathrm{cw}}$,
			\item $\pred{Add}_{\rpred, (\origcol_1, i), (\origcol_1, i)}$ or $\pred{Add}_{\rpred, (\origcol_1, i)}$ depending on $\arity{\rpred} = 2$ or $\arity{\rpred} = 1$.
		\end{itemize}
		By construction we may obtain $\unradd{\origcol_2}{\rpred}^{i}(s')\in\mathcal{T}$ or  $\unradd{\origcol_1}{\rpred}^{i}(s')\in\mathcal{T}$ depending on $\arity{\rpred} = 2$ or $\arity{\rpred} = 1$. Consequently, $\rpred(\ve)\in\inst^{\mathcal{T}} = \inst$.
		\item $\ve \neq (e, \ldots, e)$ for every entity $e$. This implies immediately that $\arity{\rpred} = 2$. Hence $\ve = (e_1, e_2)$ with $e_1 \neq e_2$. Let $s_{\ve} = \fcn{cca}(\fcn{orig}(\ve))$. Since $\rpred(\fcn{h}^{\mathcal{T}}\!(\ve))\in\inst^{\mathcal{T}_{\mathrm{cw}}}$ there is by construction a node $v$ in $\mathcal{T}_{\mathrm{cw}}$ such that $\pred{Add}_{\rpred, (\kappa_1, i_1), (\kappa_2, i_2)}(v)\in\mathcal{T}_{\mathrm{cw}}$ with $\{ i_1, i_2 \} = \{ 1, 2 \}$, $ \col_{v}^{\mathcal{T}_{\mathrm{cw}}}(\fcn{h}^{\mathcal{T}}\!(e_{1})) = (\kappa_1, i_1)$ and $\col_{v}^{\mathcal{T}_{\mathrm{cw}}}(\fcn{h}^{\mathcal{T}}\!(e_{2})) = (\kappa_2, i_2)$. Furthermore, $v$ satisfies the following property: $v$ is a proper ancestor to both $\fcn{r}^{\mathcal{T}}\!(s0)$ and $\fcn{r}^{\mathcal{T}}\!(s1)$ and $\fcn{r}^{\mathcal{T}}\!(s)$ is a proper ancestor to $v$. Also, by construction we immediately obtain $ \col_{\fcn{r}^{\mathcal{T}}(s(i_1-1))}^{\mathcal{T}_{\mathrm{cw}}}(\fcn{h}^{\mathcal{T}}\!(e_{1})) = \kappa_1$ and $\col_{\fcn{r}^{\mathcal{T}}(s(i_2 - 1))}^{\mathcal{T}_{\mathrm{cw}}}(\fcn{h}^{\mathcal{T}}\!(e_{2})) = \kappa_2$. Again by construction we now that $\radd{\kappa_1}{\kappa_2}{\{i_1 \}}{\rpred}(s)\in\mathcal{T}$. Finally, applying \cref{lem:hcwtocw-bin-cols} we obtain $e_1\in\col_{s(i_1-1)}^{\mathcal{T}}(\kappa_1)$ and $e_2\in\col_{s(i_2-1)}^{\mathcal{T}}(\kappa_2)$. Together this implies $\rpred(e_1, e_2)\in\mathrm{Atoms}_{s}$ and hence $\rpred(\ve)\in\inst^{\mathcal{T}} = \inst$.
	\end{itemize}
	This proves that $\fcn{h}^{\mathcal{T}}$ is an isomorphism between $\inst = \inst^{\mathcal{T}}$ and $\inst^{\mathcal{T}_{\mathrm{cw}}}$. As $|\cols_{\mathrm{cw}}| = 3\cdot |\cols_1| + 1 < 3 \cdot |\cols| + 1$ we conclude that since $\inst$ has finite \partitionwidth{}, it has also finite cliquewidth.
\end{proof}

%\section{Proof for \cref{sec:strat}}
\section{Proof for Section 9}

\begin{customlem}{\ref{lem:n-cut-layered-chase}}
Let $\ruleset$ be a ruleset. If $\ruleset = \ruleset_{1} \cut \cdots \cut \ruleset_{n}$, then for any database $\db$, the following holds:
$$
\sa(\db,\ruleset) = \sa(\ldots\sa(\sa(\db,\ruleset_{1}),\ruleset_{2}) \ldots, \ruleset_{n}).
$$
\end{customlem}

\begin{proof} For $n=1$ the statement is trivial. For all other $n$, we prove the result by induction.

\textit{Base case.} Let $\ruleset = \ruleset_{1} \cut \ruleset_{2}$. We argue two cases: we first argue (i) that $\sa(\db,\ruleset) \subseteq \sa(\sa(\db,\ruleset_{1}),\ruleset_{2})$, and then show (ii) that $\sa(\sa(\db,\ruleset_{1}),\ruleset_{2}) \subseteq \sa(\db,\ruleset)$. 

For case (i), let us suppose that $\rpred(\vec{t}) \in \sa(\db,\ruleset)$. If $\rpred(\vec{t}) \in \db$, then $\rpred(\vec{t}) \in \sa(\sa(\db,\ruleset_{1}),\ruleset_{2})$, so let us assume that there exists an $i$ such that $\rpred(\vec{t}) \in \chase_{i}(\db,\ruleset)$, but $\rpred(\vec{t}) \not\in \chase_{j}(\db,\ruleset)$ for $j < i$. It follows that there exists a trigger $(\rho,h)$ of $\chase_{i-1}(\db,\ruleset)$ such that $\rpred(\vec{t}) \in \overline{h}(\head(\rho))$. Either $\rho \in \ruleset_{1}$ or $\rho \in \ruleset_{2}$. If $\rho \in \ruleset_{1}$, then since no rule of $\ruleset_{1}$ depends on a rule in $\ruleset_{2}$, it follows that $\rpred(\vec{t}) \in \sa(\db,\ruleset_{1})$, meaning $\rpred(\vec{t}) \in \sa(\sa(\db,\ruleset_{1}),\ruleset_{2})$. If $\rho \in \ruleset_{2}$, then $\rpred(\vec{t}) \in \sa(\sa(\db,\ruleset_{1}),\ruleset_{2})$, which completes the proof of case (i). 

For case (ii), let us suppose that $\rpred(\vec{t}) \in \sa(\sa(\db,\ruleset_{1}),\ruleset_{2})$. If $\rpred(\vec{t}) \in \db$, then $\rpred(\vec{t}) \in \sa(\db,\ruleset)$, so let us suppose that $\rpred(\vec{t}) \not\in \db$. Either $\rpred(\vec{t}) \in \sa(\db,\ruleset_{1})$ or $\rpred(\vec{t}) \in \sa(\sa(\db,\ruleset_{1}),\ruleset_{2})$. Regardless of the case, one can see that $\rpred(\vec{t}) \in \sa(\db,\ruleset)$.

\textit{Inductive step.} Suppose that $\ruleset = \ruleset_{1} \cut \cdots \cut \ruleset_{n}$, and let $\ruleset' = \bigcup_{i \in \{1, \ldots n-1\}} \ruleset_{i}$, from which it follows that $\ruleset' = \ruleset_{1} \cut \cdots \cut \ruleset_{n-1}$. By the base case, we know that $\sa(\db,\ruleset) = \sa(\sa(\db,\ruleset'),\ruleset_{n})$, and by IH, we know that $$\sa(\db,\ruleset') = \sa(\ldots\sa(\sa(\db,\ruleset_{1}),\ruleset_{2}) \ldots, \ruleset_{n-1}).$$Therefore, $\sa(\db,\ruleset) = \sa(\ldots\sa(\sa(\db,\ruleset_{1}),\ruleset_{2}) \ldots, \ruleset_{n}).$
\end{proof}

%%EXTENSIONS
%\input{definability-appendix}

%%NEW APP SECTION
%% NOT NEEDED, AS ONE CAN APPLY BAGET ET AL's ARGUMENT
%\section{Undecidability of HCS Containment}
%\input{app-undecidability}

%%NEW APP SECTION
%\section{Section 4 Proofs}\label{app:section-4-proofs}

%%NEW APP SECTION
%\section{Proofs for \cref{sec:FO-rewritability}}\label{app:section-5-proofs}
\section{Proofs for Section 10}\label{app:section-5-proofs}

\begin{customprop}{\ref{prop:frone-fusddl}}
Let $\ruleset = \ruleset_{\smash{\mathrm{fr1}}} \cup \ruleset_{\smash{\mathrm{ff}}}$ be a ruleset where $\ruleset_{\smash{\mathrm{fr1}}}$ is frontier-one and $\ruleset_{\smash{\mathrm{ff}}}$ is full and $\fus$. Then $\ruleset$ is $\fhcs$ and thus allows for decidable HUSOMSOQ entailment.  
\end{customprop}

\begin{proof} We define an atom $\rpred(\vec{t}) \in \chase_{\infty}(\db,\ruleset_{\smash{\mathrm{fr1}}} \cup \ruleset_{\smash{\mathrm{ff}}})$ to be an \emph{ff-atom} \iffi $\rpred(\vec{t})$ was introduced at some step of the chase by a rule in $\ruleset_{\smash{\mathrm{ff}}}$. We define 
$$
\chase_{\smash{\mathrm{fr1}}}(\db,\ruleset_{\smash{\mathrm{fr1}}} \cup \ruleset_{\smash{\mathrm{ff}}}) = \chase_{\infty}(\db,\ruleset_{\smash{\mathrm{fr1}}} \cup \ruleset_{\smash{\mathrm{ff}}}) \setminus \{\rpred(\vec{t}) \ | \ \rpred(\vec{t}) \text{ is an ff-atom} \}.
$$%
 It is straightforward to confirm that $\chase_{\smash{\mathrm{fr1}}}(\db,\ruleset_{\smash{\mathrm{fr1}}} \cup \ruleset_{\smash{\mathrm{ff}}})$ has finite treewidth, and thus, the instance has finite \hypercliquewidth{} by \cref{thm:tree-width-implies-clique-width}. 
 
 Let us now take the ruleset $\ruleset_{\smash{\mathrm{ff}}}$ and observe that every rule in $\ruleset_{\smash{\mathrm{ff}}}$ is a full rule of the form $\forall \vec{x} \vec{y}(\phi(\vec{x},\vec{y}) \rightarrow \rpred_{1}(\vec{z}_{1}) \land \cdots \land \rpred_{n}(\vec{z}_{n}))$ with $\vec{z}_{1}, \ldots, \vec{z}_{n}$ all variables in $\vec{y}$. We can therefore rewrite each such rule as the following, equivalent set of single-headed rules:
$$
\{\forall \vec{x} \vec{y}(\phi(\vec{x},\vec{y}) \rightarrow \rpred_{1}(\vec{z}_{1})), \ldots, \forall \vec{x} \vec{y}(\phi(\vec{x},\vec{y}) \rightarrow \rpred_{n}(\vec{z}_{n}))\}.
$$%
Based on this fact, we may assume w.l.o.g. that every rule in $\ruleset_{\smash{\mathrm{ff}}}$ is a single-headed rule.
 Let us now rewrite $\ruleset_{\smash{\mathrm{ff}}}$ relative to itself. To do this, we first define the rewriting of each rule in $\ruleset_{\smash{\mathrm{ff}}}$. For each rule $\rho = \forall \vec{x} \vec{y}(\phi(\vec{x},\vec{y}) \rightarrow \rpred(\vec{y})) \in \ruleset_{\smash{\mathrm{ff}}}$, we define the rewriting of $\rho$ accordingly:
$$
\fcn{rew}(\rho) := \{\psi(\vec{x}, \vec{y}) \rightarrow \rpred(\vec{y}) \ | \ \exists \psi(\vec{x},\vec{y}) \in \fcn{rew}_{\ruleset_{\smash{\mathrm{ff}}}}(\rpred(\vec{y})) \}.
$$%
We then define the rewriting of $\ruleset_{\smash{\mathrm{ff}}}$ as follows: $\ruleset_{\smash{\mathrm{ff}}}' := \bigcup_{\rho \in \ruleset_{\smash{\mathrm{ff}}}} \fcn{rew}(\rho)$. Since $\ruleset_{\smash{\mathrm{ff}}}$ is $\fus$, it follows that $\ruleset_{\smash{\mathrm{ff}}}'$ is finite. Additionally, $\ruleset_{\smash{\mathrm{ff}}}'$ is a full ruleset consisting solely of single-headed rules by definition, and so, we assume that the ruleset has the following form:
$$
\ruleset_{\smash{\mathrm{ff}}}' = \{\psi_{1}(\vec{x}, \vec{y}) \rightarrow \rpred_{1}(\vec{y}), \ldots, \psi_{n}(\vec{x}, \vec{y}) \rightarrow \rpred_{n}(\vec{y})\}.
$$

We now show the following: For any $\rpred(\vec{t}) \in \chase_{\infty}(\db,\ruleset_{\smash{\mathrm{fr1}}} \cup \ruleset_{\smash{\mathrm{ff}}}) \setminus \chase_{\smash{\mathrm{fr1}}}(\db,\ruleset_{\smash{\mathrm{fr1}}} \cup \ruleset_{\smash{\mathrm{ff}}})$, there exists a trigger $(\rho,h)$ in $\chase_{\smash{\mathrm{fr1}}}(\db,\ruleset_{\smash{\mathrm{fr1}}} \cup \ruleset_{\smash{\mathrm{ff}}})$ with $\rho \in \ruleset_{\smash{\mathrm{ff}}}'$ such that applying $(\rho,h)$ adds $\rpred(\vec{t})$ to $\chase_{\smash{\mathrm{fr1}}}(\db,\ruleset_{\smash{\mathrm{fr1}}} \cup \ruleset_{\smash{\mathrm{ff}}})$. To see this, observe that since $\ruleset_{\smash{\mathrm{ff}}}$ is $\fus$, we know by definition that for any predicate $\rpred$ occurring in any rule head of $\ruleset_{\smash{\mathrm{ff}}}$ and any tuple $\vec{t}$ of terms in $\chase_{\smash{\mathrm{fr1}}}(\db,\ruleset_{\smash{\mathrm{fr1}}} \cup \ruleset_{\smash{\mathrm{ff}}})$, it holds that   $\chase_{\infty}(\chase_{\smash{\mathrm{fr1}}}(\db,\ruleset_{\smash{\mathrm{fr1}}} \cup \ruleset_{\smash{\mathrm{ff}}}), \ruleset_{\smash{\mathrm{ff}}}) \models \rpred(\vec{t})$ \iffi there exists a CQ $\psi(\vec{x},\vec{y})$ in $\fcn{rew}_{\ruleset}(\rpred(\vec{y}))$ such that $\chase_{\smash{\mathrm{fr1}}}(\db,\ruleset_{\smash{\mathrm{fr1}}} \cup \ruleset_{\smash{\mathrm{ff}}}) \models \exists \vec{x} \psi(\vec{x},\vec{t})$. It therefore follows that a trigger exists for $\rho \in \ruleset_{\smash{\mathrm{ff}}}'$ in $\chase_{\smash{\mathrm{fr1}}}(\db,\ruleset_{\smash{\mathrm{fr1}}} \cup \ruleset_{\smash{\mathrm{ff}}})$.

 Based on the above fact, the following equation holds:
$$
\chase_{\infty}(\db,\ruleset_{\smash{\mathrm{fr1}}} \cup \ruleset_{\smash{\mathrm{ff}}}) = \chase_{1}(\chase_{\smash{\mathrm{fr1}}}(\db,\ruleset_{\smash{\mathrm{fr1}}} \cup \ruleset_{\smash{\mathrm{ff}}}), \ruleset_{\smash{\mathrm{ff}}}').
$$%
 Recall that $\ruleset_{\smash{\mathrm{ff}}}'$ is of the form $\{\psi_{i}(\vec{x}, \vec{y}) \rightarrow \rpred_{i}(\vec{y}) \ | \ 1 \leq i \leq n\}$. The right hand side of the above equation can be rewritten as a finite sequence of FO-extensions (see~\cref{sec:preservation}):%\pagebreak 
$$
\chase_{1}(\chase_{\smash{\mathrm{fr1}}}(\db,\ruleset_{\smash{\mathrm{fr1}}} \cup \ruleset_{\smash{\mathrm{ff}}}), \ruleset_{\smash{\mathrm{ff}}}') = (\ldots(\chase_{\smash{\mathrm{fr1}}}(\db,\ruleset_{\smash{\mathrm{fr1}}} \cup \ruleset_{\smash{\mathrm{ff}}}))^{\rpred_{1}:= \rpred_{1} \cup \psi_{1}})\ldots)^{\rpred_{n}:= \rpred_{n} \cup \psi_{n}}.
$$%
 By what was argued above, $\chase_{\smash{\mathrm{fr1}}}(\db,\ruleset_{\smash{\mathrm{fr1}}} \cup \ruleset_{\smash{\mathrm{ff}}})$ has finite \hypercliquewidth{}, and since this property is preserved under FO-extensions (by \cref{thm:closureproperties}), it follows that  $$\chase_{1}(\chase_{\smash{\mathrm{fr1}}}(\db,\ruleset_{\smash{\mathrm{fr1}}} \cup \ruleset_{\smash{\mathrm{ff}}}), \ruleset_{\smash{\mathrm{ff}}}')$$has finite \ercwidth{} and hence, $\chase_{\infty}(\db,\ruleset_{\smash{\mathrm{fr1}}} \cup \ruleset_{\smash{\mathrm{ff}}})$ has finite \ercwidth{}. Therefore, by \cref{thm:decidablequerying}, $\ruleset_{\smash{\mathrm{fr1}}} \cup \ruleset_{\smash{\mathrm{ff}}}$ allows for decidable HUSOMSOQ entailment.
 %Therefore, by \cref{cor:fcsdecidable}, $\ruleset_{\smash{\mathrm{fr1}}} \cup \ruleset_{\smash{\mathrm{ff}}}$ allows for decidable DiDaMSOQ entailment. 
\end{proof}

\begin{proposition}
Every full $\fus$ ruleset is $\fods$.
\end{proposition}
\begin{proof}
Consider a full $\fus$ ruleset $\ruleset$ and observe that every rule in $\ruleset$ is a full rule of the form $\forall \vec{x} \vec{y}(\phi(\vec{x},\vec{y}) \rightarrow \rpred_{1}(\vec{z}_{1}) \land \cdots \land \rpred_{n}(\vec{z}_{n}))$ with $\vec{z}_{1}, \ldots, \vec{z}_{n}$ all variables in $\vec{y}$. We can therefore rewrite each such rule as the following, equivalent set of single-headed rules:
$$
\{\forall \vec{x} \vec{y}(\phi(\vec{x},\vec{y}) \rightarrow \rpred_{1}(\vec{z}_{1})), \ldots, \forall \vec{x} \vec{y}(\phi(\vec{x},\vec{y}) \rightarrow \rpred_{n}(\vec{z}_{n}))\}.
$$%
Based on this fact, we may assume w.l.o.g. that every rule in $\ruleset$ is a single-headed rule.

We will next rewrite $\ruleset$ relative to itself. To do this, we first define the rewriting of each rule in $\ruleset$. For each rule $\rho = \forall \vec{x} \vec{y}(\phi(\vec{x},\vec{y}) \rightarrow \rpred(\vec{y})) \in \ruleset$, we define the rewriting of $\rho$ accordingly:
$$
\fcn{rew}(\rho) := \{\psi(\vec{x}, \vec{y}) \rightarrow \rpred(\vec{y}) \mid  \psi(\vec{x},\vec{y}) \in \fcn{rew}_{\ruleset}(\rpred(\vec{y})) \}.
$$%
We then define the rewriting of $\ruleset$ as follows: $\ruleset' := \bigcup_{\rho \in \ruleset} \fcn{rew}(\rho)$. Since $\ruleset$ is $\fus$ by assumption, it follows that $\ruleset'$ is finite. Additionally, $\ruleset'$ is a full ruleset consisting solely of single-headed rules by definition, and so, we assume that the ruleset has the following form:
$$
\ruleset' = \{\psi_{1}(\vec{x}, \vec{y}) \rightarrow \rpred_{1}(\vec{y}), \ldots, \psi_{n}(\vec{x}, \vec{y}) \rightarrow \rpred_{n}(\vec{y})\}.
$$

We now show the following: For any $\inst$ and any $\rpred(\vec{t}) \in \chase_{\infty}(\inst,\ruleset) \setminus \inst$, there exists a trigger $(\rho,h)$ in $\inst$ with $\rho \in \ruleset'$ such that applying $(\rho,h)$ adds $\rpred(\vec{t})$ to $\inst$. To see this, observe that since $\ruleset$ is $\fus$, we know by definition that for any predicate $\rpred$ occurring in any rule head of $\ruleset$ and any tuple $\vec{t}$ of terms in $\inst$, it holds that 
$\chase_{\infty}(\inst, \ruleset) \models \rpred(\vec{t})$ \iffi 
there exists a CQ $\exists\vec{x}.\psi(\vec{x},\vec{y})$ in $\fcn{rew}_{\ruleset}(\rpred(\vec{y}))$ such that  $\inst \models \exists \vec{x} \psi(\vec{x},\vec{t})$. It therefore follows that a trigger exists for the corresponding $\rho \in \ruleset_{\smash{\mathrm{ff}}}'$ in $\inst$.

Based on the above fact, we obtain
$
\chase_{\infty}(\inst,\ruleset) = \chase_{1}(\inst,\ruleset').
$
Recall that $\ruleset_{\smash{\mathrm{ff}}}'$ is of the form $\{\psi_{i}(\vec{x}, \vec{y}) \rightarrow \rpred_{i}(\vec{y}) \ | \ 1 \leq i \leq n\}$. Yet, the right hand side of the above equation can be rewritten as a finite sequence of FO-extensions as shown below:
$$
\chase_{1}(\inst, \ruleset') = (\ldots(\inst)^{\rpred_{1}:= \rpred_{1} \cup \psi_{1}})\ldots)^{\rpred_{n}:= \rpred_{n} \cup \psi_{n}},
$$%
which finishes our proof.
\end{proof}

\begin{customthm}{\ref{thm:fcs-fus-decidable}}
	Let $\ruleset = \ruleset_{1} \cut \ruleset_{2}$. If $\ruleset_{1}$ is $\fhcs$ and $\ruleset_{2}$ is $\fus$, then the query entailment problem $(\db,\ruleset) \models q$ for databases $\db$, and CQs $q$ is decidable.
\end{customthm}

\begin{proof} Let us suppose that $\ruleset = \ruleset_{1} \cut \ruleset_{2}$, $\ruleset_{1}$ is $\fhcs$, and $\ruleset_{2}$ is $\fus$. Moreover, let $\db$ be an arbitrary database. Since $\ruleset = \ruleset_{1} \cut \ruleset_{2}$, we have $\sa(\db,\ruleset) = \sa(\sa(\db,\ruleset_{1}),\ruleset_{2})$. As $\ruleset_{2}$ is $\fus$, we know that for any tuple $\vec{\const{a}}$ of constants from $\db$, $\sa(\sa(\db,\ruleset_{1}),\ruleset_{2}) \models q(\vec{\const{a}})$ \iffi $\sa(\db,\ruleset_{1}) \models \rewrs{\ruleset_{2}}{\query(\vec{\const{a}})}$, where $\rewrs{\ruleset_{2}}{\query}$ is a rewriting of $\query$. By \cref{thm:decidablequerying}, we can decide if $\sa(\db,\ruleset_{1}) \models \rewrs{\ruleset_{2}}{\query}$ as $\rewrs{\ruleset_{2}}{\query}$ is a UCQ, which is a type of HUSOMSOQ. Hence, the query entailment problem $(\db,\ruleset) \models q$ is decidable. %By \cref{cor:fcsdecidable}, we can decide if $\sa(\db,\ruleset_{1}) \models \rewrs{\ruleset_{2}}{\query}$ as $\rewrs{\ruleset_{2}}{\query}$ is a UCQ, which is a type of DiDaMSOQ. Hence, the query entailment problem $(\db,\ruleset) \models q$ is decidable.
\end{proof}

\end{document}